\documentclass[11pt]{elsarticle}
\journal{Journal of Mathematical Economics}

\biboptions{sort}

\usepackage[letterpaper]{geometry}
\usepackage{array}
\usepackage{setspace}

\usepackage{amsthm}
\usepackage{amsmath}
\usepackage{amssymb}

\theoremstyle{plain}
\newtheorem{lemma}{Lemma}[section]
\newtheorem{corollary}{Corollary}[section]
\newtheorem*{conjecture*}{Conjecture}

\theoremstyle{definition}
\newtheorem{example}{Example}[section]
\newtheorem{defn}{Definition}[section]

\theoremstyle{remark}
\newtheorem{claim}{Claim}[section]

\usepackage{mathpazo}   

\usepackage[crop=off]{auto-pst-pdf}
\usepackage{pstricks,pst-node,pst-tree}
\usepackage{graphics,graphicx}
\usepackage[overload]{textcase}

\makeatletter
\def\Square{\pst@object{Square}}
\def\Square@i(#1,#2)#3{{\use@par\solid@star\psframe[origin={#1,#2}](#3,#3)}}
\def\squarediagonal{\pst@object{squarediagonal}}
\def\squarediagonal@i(#1,#2)#3{{\use@par\solid@star\Square[linestyle=solid,opacity=0.5](#1,#2){#3}\psline[linestyle=dashed,origin={#1,#2}](0,#3)(#3,0)}}
\makeatother

\newcommand{\squarepool}[1]{\Square[fillstyle=solid,fillcolor=blue,linecolor=blue]#1{1}}
\newcommand{\smallsquarepool}[1]{\Square[fillstyle=solid,fillcolor=blue,linecolor=blue]#1{0.4}}

\newcommand{\squarepoolwithlabel}[2]{
	\rput(0.05,0.05)
	{\Square[fillstyle=solid,fillcolor=blue,linecolor=blue]#1{0.2}}
	\uput{7pt}[0]#1{\tiny{#2}}
}

\newcommand{\steptitle}[1] {
}

\newcommand{\valuenotes}[3][-0.5] {}

\psset{dotsep=2pt,linecolor=blue}

\def\covernum(#1,#2){\text{CoverNum}(#1,#2)}
\def\egal(#1,#2,#3,#4){\text{Egal}(#1,#2,#3,#4)}
\def\egalrel(#1,#2,#3,#4){\text{RelEgal}(#1,#2,#3,#4)}
\def\prop(#1,#2,#3){\text{Prop}(#1,#3,#2)}
\def\propsame(#1,#2,#3){\text{PropSame}\allowbreak(#1,#3,#2)}
\def\propcoeff(#1,#2){\text{PropCoeff}(#1,#2)}
\def\proprel(#1,#2,#3){\text{RelProp}(#1,#3,#2)}
\def\proprelsame(#1,#2,#3){\text{RelPropSame}(#1,#3,#2)}
\def\overlap(#1){\text{Overlap}\allowbreak(#1)}

\newcommand{\South}{bottom }

\newcommand{\West}{left }

\newcommand{\range}[2]{\in\{#1,\dots,#2\}}
\newcommand{\spiece}{$S$-piece}
\newcommand{\spieces}{$S$-pieces}

\makeatother

\providecommand{\tabularnewline}{\\}
\usepackage{enumitem}		
\usepackage[shortcuts]{extdash}
\usepackage{caption}
\usepackage{multirow}
\usepackage{bm}

\usepackage{nohyperref}  
\usepackage{url}         

\providecommand{\doi}[1]{%
  \begingroup
    \let\bibinfo\@secondoftwo
    \urlstyle{rm}%
    \href{http://dx.doi.org/#1}{%
      doi:\discretionary{}{}{}%
    }%
  \endgroup
}

\newcommand{\n}[1]{~~#1}

\begin{document}

\begin{frontmatter}{}

\title{Fair and Square: Cake-Cutting in Two Dimensions}

\author{Erel Segal-Halevi\footnote{erelsgl@cs.biu.ac.il . Corresponding author.},
Shmuel Nitzan\footnote{nitzans@biu.ac.il}, Avinatan Hassidim\footnote{avinatan.hassidim@biu.ac.il}
and Yonatan Aumann\footnote{aumann@cs.biu.ac.il}}

\address{Bar-Ilan University, Ramat-Gan 5290002, Israel }
\begin{abstract}
We consider the classic problem of fairly dividing a heterogeneous good ("cake") among several agents with different valuations. Classic cake-cutting procedures either allocate each agent a collection of disconnected pieces, or assume that the cake is a one-dimensional interval. In practice, however, the two-dimensional shape of the allotted pieces is important. In particular, when building a house or designing an advertisement in printed or electronic media, squares are more usable than long and narrow rectangles. We thus introduce and study the problem of fair two-dimensional division wherein the allotted pieces must be of some restricted two-dimensional geometric shape(s), particularly squares and fat rectangles. Adding such geometric constraints re-opens most questions and challenges related to cake-cutting. Indeed, even the most elementary fairness criterion --- \emph{proportionality} --- can no longer be guaranteed. In this paper we thus examine the level of proportionality that \emph{can} be guaranteed, providing both impossibility results and constructive division procedures. 

\end{abstract}
\end{frontmatter}{}
\pagebreak{}

\section{Introduction}
Fair division of land has been an important issue since the dawn of
history. One of the classic fair division procedures, ``I cut and you choose'', is already alluded to in the Bible (Genesis 13) as a method for dividing land between two people. The modern study of this problem, commonly termed \emph{cake cutting}, began in the 1940's. The first challenge was conceptual --- how should ``fairness'' be defined when the cake is heterogeneous and different people may assign different values to subsets of the cake? \citet{Steinhaus1948Problem} introduced the elementary and most basic fairness requirement, now termed \emph{proportionality}:
each of the $n$ agents should get a piece which he values as worth at least $1/n$ of the value of the entire cake. He also presented a procedure, suggested by Banach and Knaster, for proportionally dividing a cake among an arbitrary number of agents. Since then, many other desirable properties of cake partitions have been studied, including:
envy-freeness \citep[e.g. ][]{Weller1985Fair,Brams1996Fair,Su1999Rental,Barbanel2004Cake},
social welfare maximization \citep[e.g. ][]{Cohler2011Optimal,Bei2012Optimal,Caragiannis2012Efficiency} and strategy-proofness \citep[e.g. ][]{Mossel2010Truthful,Chen2013Truth,Cole2013Mechanism}. See the books by \citet{Brams1996Fair,Robertson1998CakeCutting,Barbanel2005Geometry,Brams2007Mathematics} and a recent survey by \citet{Procaccia2015Cake} for more information.

Many economists regard land division as an important application of division procedures \citep[e.g. ][]{Berliant1988Foundation,Berliant1992Fair,Legut1994Economies,Chambers2005Allocation,DallAglio2009Disputed,Husseinov2011Theory,Nicolo2012Equal}).
Hence, they note the importance of imposing some geometric constraints
on the pieces allotted to the agents. The most well-studied constraint
is \emph{connectivity} --- each agent should receive a single connected
piece. The cake is usually assumed to be the one-dimensional interval
$[0,1]$ and the allotted pieces are sub-intervals \citep[e.g. ][]{Stromquist1980How,Su1999Rental,Nicolo2008Strategic,Azrieli2014Rental}).
This assumption is usually justified by the reasoning that higher-dimensional
settings can always be projected onto one dimension, and hence fairness
in one dimension implies fairness in higher dimensions.\footnote{In the words of \citet{Woodall1980Dividing}: ``the cake
is simply a compact interval which without loss of generality I shall take to be [0,1]. If you find this thought unappetizing, by all means think of a three-dimensional cake. Each point P of division
of my cake will then define a plane of division of your cake: namely, the plane through P orthogonal to [0,1]''.} However, projecting back from the one dimension, the resulting two-dimensional
plots are thin rectangular slivers, of little use in most practical applications; it is hard to build a house on a $10\times1,000$ meter plot even though its area is a full hectare, and a thin 0.1-inch wide
advertisement space would ill-serve most advertises regardless of its height.

We claim that the \emph{two-dimensional shape} of the allotted piece is of prime importance. Hence, we seek divisions in which the allotted pieces must be of some restricted family of ``usable'' two-dimensional shapes, e.g. squares or polygons of balanced length/width ratio.

\begin{figure}
\psset{unit=1mm,dotsep=1pt,linecolor=blue}
\centering
\begin{pspicture}(50,45)
\rput[l](0,35){\tiny{(a) Two disjoint rectangles worth 1/2}}
\Square[linestyle=solid,linecolor=brown](0,0){30}
\psframe[linestyle=dashed,linecolor=red](1,1)(14,29)
\psframe[linestyle=dashed,linecolor=blue](16,1)(29,29)
\end{pspicture}
\begin{pspicture}(50,45)
\rput[l](0,35){\tiny{(b) Two disjoint squares worth 1/4}}
\Square[linestyle=solid,linecolor=brown](0,0){30}
\Square[linestyle=dashed,linecolor=red](1,1){14}
\Square[linestyle=dashed,linecolor=blue](15,15){14}
\end{pspicture}
\begin{pspicture}(50,45)
\rput[l](0,35){\tiny{(c) No two disjoint squares worth more than 1/4}}
\Square[linestyle=solid,linecolor=brown](0,0){30}
\Square[linestyle=dashed,linecolor=red](1,1){17}
\Square[linestyle=dashed,linecolor=blue](12,12){17}
\end{pspicture}

\protect\caption{\label{fig:impossibility-intro} Dividing a square cake to two agents.}
\end{figure}
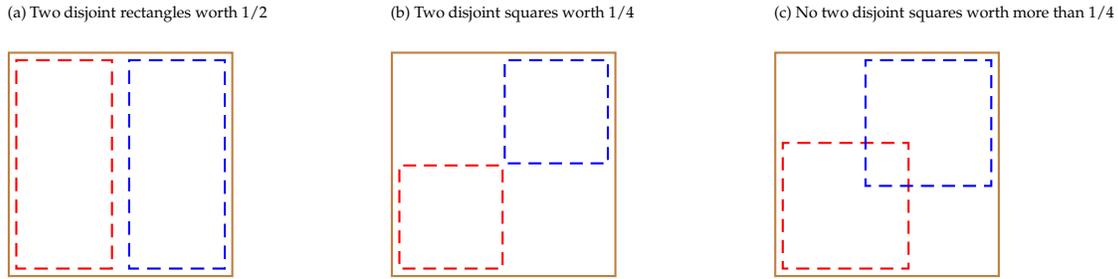

Adding a two-dimensional geometric constraint re-opens most questions
and challenges related to cake-cutting. Indeed, even the elementary
proportionality criterion can no longer be guaranteed.
\begin{example}
\label{exm:intro}A homogeneous square land-estate has to be divided
between two heirs. Each heir wants to use his share for building a
house with as large an area as possible, so the utility of each heir
equals the area of the largest house that fits in his piece (see Figure
\ref{fig:impossibility-intro}). If the houses can be rectangular,
then it is possible to give each heir $1/2$ of the total utility
(a); if the houses must be square, it is possible to give each heir
$1/4$ of the total utility (b) but impossible to give both heirs \emph{more} than $1/4$ the total utility (c). In particular, when the allotted pieces must be square, a proportional division does not exist.\footnote{\citet{Berliant2004Foundation} use a very similar example to prove the nonexistence of a competitive equilibrium when the pieces must be square.}
\end{example}
This example invokes several questions. What happens when the land-estate is heterogeneous and each agent has a different utility function?
Is it always possible to give each agent a 2-by-1 rectangle worth for him at least $1/2$ the total value? Is it always possible to give each agent a square worth for him at least $1/4$ the total value? Is it even possible to guarantee a positive fraction of the total value? If it is possible, what division procedures
can be used? How does the answer change when there are more than two agents? Such questions are the topic of the present paper.

We use the term \emph{proportionality} to describe the fraction that
can be guaranteed to every agent. So when the shape of the pieces
is unrestricted, the proportionality is always $1/n$, but when the
shape is restricted, the proportionality might be smaller. Naturally, the attainable proportionality depends on both the shape of the cake and the desired shape of the allotted pieces. For every combination of cake shape and piece shape, one can prove \emph{impossibility results}
(for proportionality levels that cannot be guaranteed) and \emph{possibility results }(for the proportionality that can be guaranteed). While we
examined many such combinations, the present paper focuses on several representative scenarios which, in our opinion, demonstrate the richness of the two-dimensional cake-cutting task. 

\subsection{Walls and unbounded cakes}

In Example \ref{exm:intro}, the two pieces had to be contained in
the square cake. One can think of this situation as dividing a square
island surrounded in all directions by sea, or a square land-estate
surrounded by 4 walls: no land-plot can overlap the sea or cross a
wall. 

In practical situations, land-estates often have less than 4 walls.
For example, consider a square land-estate that is bounded by sea
to the west and north but opens to a desert to the east and south.
Allocated land-plots may not flow over the sea shore, but they may
flow over the borders to the desert. 

Cakes with less than 4 walls can also be considered as unbounded cakes.
For example, the above-mentioned land-estate with 2 walls can be considered
a quarter-plane. The total value of the cake is assumed to be finite
even when the cake is unbounded. When considering unbounded cakes,
the pieces are allowed to be ``generalized squares'' with an infinite
side-length. For example, when the cake is a quarter-plane (a square
with 2 walls), we allow the pieces to be squares or quarter-planes.
When the cake is a half-plane (a square with 1 wall), we also allow
the pieces to be half-planes, etc.  The terms ``square with 2 walls''
and ``quarter-plane'' are used interchangeably throughout the paper.

\subsection{Fat objects}\label{sub:Fat-pieces}

Intuitively, a piece of cake is usable if its lengths in all dimensions are balanced --- it is not too long in one dimension and too short in another dimension. This intuition is captured by the concept of \emph{fatness},
which we adapt from the computational geometry literature \citep[e.g. ][]{Agarwal1995Computing,Katz19973D}:
\begin{defn}
\label{def:rfat}A $d$-dimensional object is called \emph{$R$-fat}, for $R\geq1$, if 
it contains a $d$-dimensional cube $c^-$ and is contained in a parallel $d$-dimensional cube $c^+$, such that the ratio between the side-lengths of the cubes is at most $R$: $\text{len}(c^+)/\text{len}(c^-)\leq R$
.
\end{defn}
A two-dimensional cube is a square. So, for example, a square is 1-fat, a 10-by-20 rectangle is 2-fat, a right-angled isosceles triangle is 2-fat and a circle is $\sqrt{2}$-fat. 

Note that $R$ is an upper bound, so if $R_2\geq R_1,$ every $R_1$-fat piece is also $R_2$-fat. So a square is also 2-fat, but a 10-by-20 rectangle is not 1-fat.

Our long-term research plan is to study various families of fat shapes. As a first step, we study the simplest fat shape, which is the square (hence the name of the paper). Despite its simplicity, it is still challenging. We also present results for fat rectangles, which are almost identical to the results for squares.
\subsection{Results}
\label{sub:results}
Our results can be broadly summarized as follows.
\begin{itemize}
\item \textbf{Negative results:} when the pieces have to be squares or fat rectangles, a proportional division is usually
\footnote{
We have proved this for most, but not all the cases that we have studied. The exception is when the cake is an unbounded plane and the pieces are non-parallel squares: in this case, we do not know whether a proportional division always exists. See Table \ref{tab:Proportionality-bounds-for-square-cakes} below.
}
not guaranteed to exist. Moreover, there is a constant $A>1$ that depends on the shape of the cake and usable pieces, such that 
, for some sets of value-functions,
it is impossible to give all agents a value of more than $1/(A\cdot n)$.
\item \textbf{Positive results:} when the pieces have to be squares or fat rectangles, a constant-factor approximation to a proportional division is usually guaranteed to exist. This means that there is a constant $B>1$ that depends on the shape of the cake and usable pieces, such that,
for all sets of value-function,
it is possible to give all agents a value of at least $1/(B\cdot n)$.
\end{itemize}
The constant $A$ in our negative results is at most 2, and the constant $B$ in our positive results is at least 2; this leads us to conjecture that the ``real'' constant is 2, i.e, a half-proportional division with square pieces always exists, and half-proportionality is the best that can be guaranteed. Currently we can prove this conjecture only in several restricted scenarios, that are presented below.
\begin{table}
\noindent \psset{unit=0.3mm,dotsep=1pt,linewidth=1pt}
\def\xf{0}
\def\xt{30}
\def\yf{0}
\def\yt{30}
\newcommand\openO{
\begin{pspicture}(0,0)(35,35)
\psline[linestyle=solid,linecolor=blue](\xf,\yf)(\xf,\yt)(\xt,\yt)(\xt,\yf)(\xf,\yf)
\end{pspicture}
}
\newcommand\openI{
\begin{pspicture}(35,35)
\psline[linestyle=solid,linecolor=blue](\xf,\yt)(\xf,\yf)(\xt,\yf)(\xt,\yt)
\psline[linestyle=dotted,linecolor=brown](\xt,\yt)(\xf,\yt)
\end{pspicture}
}
\newcommand\openII{
\begin{pspicture}(35,35)
\psline[linestyle=solid,linecolor=blue](\xf,\yt)(\xf,\yf)(\xt,\yf)
\psline[linestyle=dotted,linecolor=brown](\xt,\yf)(\xt,\yt)(\xf,\yt)
\end{pspicture}
}
\newcommand\openIII{
\begin{pspicture}(35,35)
\psline[linestyle=solid,linecolor=blue](\xf,\yf)(\xt,\yf)
\psline[linestyle=dotted,linecolor=brown](\xt,\yf)(\xt,\yt)(\xf,\yt)(\xf,\yf)
\end{pspicture}
}
\newcommand\openIIII{
\begin{pspicture}(35,35)
\psline[linestyle=dotted,linecolor=brown](\xf,\yf)(\xf,\yt)(\xt,\yt)(\xt,\yf)(\xf,\yf)
\end{pspicture}
}{}%
\begin{tabular}{|c|c|c|c|c|c|}
\hline 
\multicolumn{2}{|l|}{{Cake $\downarrow$}} & \multicolumn{2}{c|}{{Impossibility}} & \multicolumn{2}{c|}{{Possibility}}
\tabularnewline
\cline{3-6}
\multicolumn{2}{|r|}{
	} & 
{Square pieces} & {\shortstack{$R$-Fat rects\\ ($R\geq2$)}} & 
{Square pieces} & {\shortstack{$R$-Fat rects\\ ($R\geq2$)}}
\tabularnewline
\hline 
\hline 
{\openO} & \shortstack{4 walls\\(Square)} & 
{$1/(2n)$ *} & {$1/(2n-1)$} & 
\shortstack{$1/(4n-4)$  *\\ \emph{same}: $1/(2n)$  *} & \shortstack{$1/(4n-5)$ \\ \emph{same}: $1/(2n-1)$} 
\tabularnewline
\hline 
{\openI} & {3 walls} & \multicolumn{2}{c|}{{$1/(2n-1)$}} & \multicolumn{2}{c|}{
$1/(2n-1)$} 
\tabularnewline
\hline 
{\openII} & \shortstack{2 walls \\ (quarter-plane)} & \multicolumn{2}{c|}{{$1/(2n-1)$}} & \multicolumn{2}{c|}{{$1/(2n-1)$}} 
\tabularnewline
\hline 
{\openIII} & \shortstack{1 wall \\ (half-plane)} & \multicolumn{2}{c|}{{$1/({3\over 2}n-1)$}} & \multicolumn{2}{c|}{{$1/(2n-2)$}}
\tabularnewline
\hline 
\multirow{2}{*}{
	\openIIII
} 
& 
\multirow{2}{*}{
\shortstack{0 walls\\(plane)} 
}
&
\multicolumn{2}{c|}{
 	\emph{axes-parallel:} $1/({10\over 9}n-1)$
}
& 
\multicolumn{2}{c|}{
	\shortstack{
	\\
	\\
	{$1/\max(2n-4,n)$}
}
}
\tabularnewline
\cline{3-4}
& 
& \shortstack{
	\emph{parallel:} $1/({30\over 29}n-1)$
	\\
	\emph{general:} ?
} 
& 
?
&
\multicolumn{2}{c|}{}
\tabularnewline
\hline 
\end{tabular}
		\caption[Summary of results for square cakes: upper and lower bounds on the level of attainable proportionality.]{
			\label{tab:Proportionality-bounds-for-square-cakes}
			Summary of results for square cakes: upper and lower bounds on the level of attainable proportionality.
			\\All results assume that there are at least two agents ($n\geq 2$).
			\\** means that the results are valid not only for square pieces but also for $R$-fat rectangles with $R<2$.
			\\ ? means that we do not have a non-trivial impossibility result for this case
			.
		}
\end{table}

\subsubsection{Square cakes bounded or unbounded} 
\label{sub:intro-square-cakes}
In the first set of results, the cake is a square bounded in zero or more sides. Table \ref{tab:Proportionality-bounds-for-square-cakes}
summarizes our negative and positive results:

The \textbf{Impossibility} column shows upper bounds on the attainable proportionality. Each upper bound is proved by showing a specific scenario in which it is impossible to give all agents more than the mentioned fraction of their total value. The upper bound for a square with 4 walls and $n=2$ is $1/(2n)=1/4$, as was already seen in Example \ref{exm:intro}. The upper bounds for an unbounded plane are valid only when the pieces must be squares parallel to a pre-specified coordinate system, or parallel to each other (as is common in urban planning). The other upper bounds are valid even when the squares are allowed to be non-parallel.

The \textbf{Possibility} column shows our positive results. Each such result is proved constructively by an explicit division procedure that gives each agent at least the mentioned fraction of their total value. 
The \emph{same} result means that there exists a different division procedure that guarantees a larger fraction per agent, but this procedure works only when all agents have the same valuations. We do not know whether there exists a division procedure that guarantees this larger fraction for agents with different valuations.

Note that all our impossibility results hold even for agents with the same valuations, and all our division procedures return axes-parallel pieces.


Intuitively, one may think that allowing rectangles instead of just squares should considerably increase the attainable proportionality level. But this is not the case if the pieces need to be fat. As seen in the table, most results for fat rectangles are almost the same as for squares. The only exception is the impossibility result for an unbounded plane, which we have not managed to extend to $R$-fat rectangles.

For $n=2$, the proportionality levels in our possibility results are equal to the impossibility results. For a cake with two or three walls the guaranteed proportionality is equal to the impossibility result for every $n$. This means that in these cases, our procedures are optimal in their worst-case guarantee. For a cake with 4 walls, the guaranteed proportionality for agents with the same value measure is optimal. In the other cases, there is a multiplicative gap of at most 2 between the possibility and the impossibility result.

A secondary consideration in geometric division problems, in addition to value, is the type of cuts used for implementing the division. In some cases, \emph{guillotine cuts} are preferred. Guillotine cuts are axis-parallel cuts running from one end to the opposite end of an already cut piece. They are considered easier to implement \cite[e.g.][]{AlvarezValdes2002Tabu,Cui2008Recursive,Hifi2011High}. In the industry, guillotine cuts are used for cutting stock such as plates of glass. In the context of land division, guillotine cuts may be desired because they may make it easier to build fences between land-plots. Our procedures for a cake with 4 walls find divisions that can be implemented using guillotine cuts. The other procedures use general cuts, and we do not know if it is possible to attain the same value guarantees using guillotine cuts.


\subsubsection{Bounded cakes of any shape}\label{sub:intro-any-shape}
While some states in the USA are rectangular (e.g. Colorado or Wyoming), most land-estates have irregular shapes. In such cases, it may be impossible to guarantee any positive proportionality. For example, consider Robinson Crusoe arriving at a circular island. Assume that Robinson's value measure is such that all value is concentrated in a very thin strip along the shore, as in Figure \ref{fig:A-circular-cake}. The value contained in any single square might be arbitrarily small. Clearly, no division procedure for $n$ agents can guarantee a better fraction of the total value.

\begin{figure}
\psset{unit=1mm}
\begin{center}
\begin{pspicture}(-5,-5)(35,35)\rput(0,-5){
\pscircle[fillstyle=solid,fillcolor=blue](15,15){15}
\pscircle[fillstyle=solid,fillcolor=white](15,15){14}
\psframe[linestyle=dashed,linecolor=red](5,5)(25,25)
}\end{pspicture}
\end{center}

\protect\caption{\label{fig:A-circular-cake}A circular cake where all value is near the perimeter. No positive value can be guaranteed to an agent who wants a square piece.}
\end{figure}
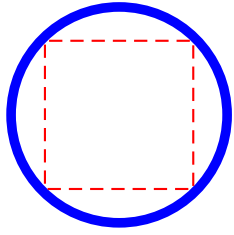

Therefore, for arbitrary cakes we use a \emph{relative} rather than absolute fairness measure. For each agent, we calculate the maximum value that this agent can attain in a square piece if he doesn't need to share the cake with other agents. We guarantee the agent a certain fraction of this value, rather than a certain fraction of the entire
cake value. This fairness criterion is similar to the \emph{uniform preference externalities} criterion suggested by \citet{Moulin1990Uniform}. Similar criteria have been recently studied in the context of indivisible item assignment \citep{Budish2011Combinatorial,Procaccia2014Fair,Bouveret2015Characterizing}.

Table \ref{tab:Relative-proportionality-bounds-} summarizes our bounds on relative proportionality. The impossibility results follow trivially from those for square cakes. The possibility results require new division procedures. They are valid for any cake that is a compact (closed
and bounded) subset of the plane. The guarantees are better
when the pieces are required to be axis-parallel. This is in accordance
with the common practice in urban planning, in which axis-parallel
plots are usually preferred.

\begin{table}
\psset{unit=0.3mm,dotsep=1pt,linewidth=1pt,linestyle=solid,linecolor=blue}
\newcommand\ParSquares{
\begin{pspicture}(30,30)
\Square(0,0){11}
\Square(16,8){13}
\Square(4,16){9}
\end{pspicture}
}
\newcommand\RotSquares{
\begin{pspicture}(-10,0)(30,30)
\rput{45}{\Square(0,0){11}}
\rput{-22}{\Square(12,8){13}}
\Square(4,16){9}
\end{pspicture}
}
\newcommand\ParRects{
\begin{pspicture}(30,30)
\psframe(0,0)(11,15)
\psframe(16,8)(30,18)
\psframe(4,16)(13,25)
\end{pspicture}
}

\noindent %
\begin{tabular}{|>{\centering}b{3cm}c|>{\centering}b{4cm}|>{\centering}b{4cm}|}
\hline 
\multicolumn{2}{|c|}{Pieces $\downarrow$} & Impossibility & Possibility\tabularnewline
\hline 
\hline 
Parallel\\
squares  & \ParSquares & $1/(2n)$ & $1/(8n-6)$

\emph{same: $1/(2n)$}\tabularnewline
\hline 
General\\
squares & \RotSquares & $1/(2n)$ & $1/(16n-14)$

\emph{same: $1/(2n)$}\tabularnewline
\hline 
Parallel\\
$R$-fat rectangles & \ParRects & $1/(2n-1)$ & $1/([4R+4][n-1]+2)$

\emph{same: $1/(2n)$}\tabularnewline
\hline 
\end{tabular}
\caption
[Summary of results for arbitrary compact cakes: bounds on the level of attainable relative proportionality.]
{
	\label{tab:Relative-proportionality-bounds-}Summary of results for arbitrary compact cakes: bounds on the level of attainable relative proportionality.
}
\end{table}

\subsection{Techniques and their economic meaning}\label{sub:auctions-intro}
Most of our division procedures can be presented as sequences of \emph{auctions}.\footnote{The relation between division procedures and auctions has already been mentioned by \citet{Brams1996Fair}.} The general process is as follows. Initially, each of the $n$ agents receives a ticket with an entitlement to share a certain cake, $C$, in a group of $n$ agents. Then, the divider performs a well-designed sequence of auctions. In each auction, the winning agents exchange their ticket for another ticket with an entitlement to share a smaller cake $C'\subset C$ in a smaller group of $n'<n$ agents. This goes on until finally each agent holds a private entitlement for a single piece of the cake. Note that there are no monetary payments: the winners 'pay' only by giving away their tickets.

We use auctions of two types: \emph{mark auction} and \emph{eval auction}.\footnote{The two auction types are analogous to the two query types --- \emph{mark query} and \emph{eval query} --- used in the cake-cutting literature in computer science, e.g.  \citet{Robertson1998CakeCutting,Woeginger2007Complexity}.
In fact, each mark/eval auction can be implemented by $n$ mark/eval queries. Therefore, all our division procedures require $O(poly(n))$ queries. 
We prefer to use auctions because their economic meaning is clearer.
}
They are presented briefly below; formal definitions and detailed examples are given in Section \ref{sec:auctions}.
\begin{itemize}
\item In a \emph{mark auction}, each agent bids by marking a piece of cake. All bids must satisfy a given geometric constraint (such as ``mark a square at the bottom-left corner''). An agent bidding a piece $X_i$ is interpreted as saying ``I am willing to give my ticket in exchange for $X_i$''. The agent bidding the smallest piece is the winner. The winner receives his bid and goes home, while the remaining agents continue to divide the remaining cake.
\item In an \emph{eval auction}, the divider specifies a piece $C'\subset C$, and each agent bids by declaring his/her evaluation of $C'$. An agent bidding a value $V$ is interpreted as saying ``I am willing to give my ticket for sharing $C$ in a group of $n$ agents, in exchange for a ticket for sharing $C'$ in a group of up to $f(V)$ agents''. Here $f: \mathbb{R}^+\to \mathbb{Z}^+$ is some weakly-increasing function that depends on the situation (the same function for all agents). The agent or agents bidding the highest values are the winners, since they are willing to share $C'$ with the largest number of other agents. The number of winners is determined as the largest value $n'$ such that the $n'$ highest winners are willing to share $C'$ in a group of $n'$. These winners go on and divide $C'$ among them, while the remaining $n-n'$ agents continue to compete on $C\setminus C'$.
\end{itemize}
The geometric constraints are carefully designed in order to guarantee that the final pieces are usable. A key geometric concept here is the \emph{cover number} --- the minimum number of squares required to cover a given region. By making sure that all sub-pieces have a sufficiently small cover-number, we ensure that they can be divided effectively. See Section \ref{sec:auctions} for details.

For the sake of simplicity, our division procedures are presented as if all agents bid according to their true value functions. However, the guarantees of our procedures are stronger: they are valid for any \emph{single} agent bidding according to his/her true value function, regardless of what the other agents do. This is the common practice in the cake-cutting world.\footnote{In the words of \citet{Steinhaus1948Problem}: ``The greed, the ignorance, and the envy of other partners cannot deprive him of the part due to him in his estimation; he has only to keep to the methods described above. Even a conspiracy of all other partners with the only aim to wrong him, even against their own interests, could not damage him.''} On the other hand, our procedures are not \emph{dominant-strategy truthful}, since an agent who knows the other agents' valuations may gain from under-bidding, just like in a first-price auction. Designing truthful
cake-cutting mechanisms is known to be a difficult problem even with a 1-dimensional cake \cite{Branzei2015Dictatorship}.

\subsection{Paper structure}
The remainder of the paper is structured as follows. The introduction section is concluded by reviewing some related research. The model is formally presented in Section \ref{sec:model}. Impossibility results are proved in Section \ref{sec:Impossibility-Results}. Section \ref{sec:auctions} presents the basic building-blocks for the division procedures: the two auction types and the geometric covering concept. These building-blocks are then used in the division procedures of Section \ref{sec:procedures}. Section \ref{sec:future} discusses several directions for future research.

\subsection{Related work}
The most prominent geometric constraint in cake-cutting is one-dimensional: the pieces must be contiguous intervals. Several authors studied a circular cake \citep{Thomson2007Children,Brams2008Proportional,Barbanel2009Cutting}, but it is still a one-dimensional circle and the pieces are one-dimensional arcs. Only few cake-cutting papers explicitly consider a two-dimensional cake. 

\citet{Hill1983Determining,Beck1987Constructing,Webb1990Combinatorial,Berliant1992Fair}
study the problem of dividing a disputed territory between several bordering countries, with the constraint that each country should get a piece that is adjacent to its border.

\citet{Iyer2009Procedure} describe a procedure that asks each of the $n$ agents to draw $n$ disjoint rectangles on the map of the two-dimensional cake. These rectangles are supposed to represent the desired regions of the agent. The procedure tries to give each agent one of his $n$ desired areas. However, the procedure does not succeed unless each rectangle proposed by an individual intersects at most one other rectangle drawn by any other agent. If even a single rectangle of Alice intersects two rectangles of George (for example), then the procedure fails and no agent receives any piece.

\citet{Berliant1992Fair,Ichiishi1999Equitable,DallAglio2009Disputed} acknowledge the importance of having nicely-shaped pieces in resolving land disputes. They prove that, if the cake is a simplex in any number of dimensions, then there exists an envy-free and proportional partition of the cake into polytopes. However, this proof is purely existential when the cake has two or more dimensions. Additionally, there are no restrictions on the fatness of the allocated polytopes and apparently these can be arbitrarily thin triangles. \citet{Berliant2004Foundation} studies the existence of competitive equilibrium with utility functions that may depend on geometric shape; their \emph{nonwasteful partitions assumption} explicitly excludes fat shapes such as squares. \citet{Devulapalli2014Geometric} studies a two-dimensional division problem in which the geometric constraints are connectivity, simple-connectivity and convexity.

In our model (see Section \ref{sec:model}), the utility functions depend on geometry, which makes them non-additive. They are not even sub-additive like in the models of \citet{Maccheroni2003How,DallAglio2005Fair,DallAglio2009Disputed}.
\footnote{\citet{DallAglio2009Disputed} do not explicitly require sub-additivity, but they require \emph{preference for concentration}: if an agent is indifferent between two pieces X and Y, then he prefers 100\% of X to 50\% of X plus 50\% of Y. This axiom may be incompatible with geometric constraints: an agent who wants square pieces will give away 100\% of a $20\times 10$ rectangle, in exchange for 50\% of a $20\times 20$ square which is the union of two such rectangles. We are grateful to Marco Dall'Aglio for his help in clarifying this issue.}
Previous papers about cake-cutting with non-additive utilities can be roughly divided to two kinds: some \citep{Berliant2004Foundation,Sagara2005Equity,Husseinov2013Existence} handle general non-additive utilities but provide only pure existence
results. Others \citep{Su1999Rental,Caragiannis2011Towards,Mirchandani2013Superadditivity} provide constructive division procedures but only for a 1-dimensional cake. Our approach is a middle ground between these extremes. Our utility functions are more general than the 1-dimensional model but less general than the arbitrary utility model; for this class of utility functions, we provide both existence results and constructive division procedures. 

Besides fair division problems, geometric methods have been used in many other economics problems,\footnote{We are thankful to Steven Landsburg, Michael Greinecker, Kenny LJ,
Alecos Papadopoulos, B Kay and Martin van der Linden for contributing these references in economics.stackexchange.com website (http://economics.stackexchange.com/q/6254/385).} such as voting \citep{Plott1967Notion}, trade theory and growth theory \citep[e.g. ][]{Johnson1971Trade}, tax burdens \citep{Hines1995From}, social choice \citep{Cantillon2002Graphical}, mechanism design \citep{Goeree2011Geometric}, public good/bad allocation \citep[e.g. ][]{Ozturk2013Strategyproof,Ozturk2014Location,Chatterjee2016Locating}, utility theory \citep{Abe2012Geometric} and general economics models \citep{Michaelides2006Euclidean}.

With square pieces a proportional allocation may not exist, so we have to settle for partial\-/proportionality. Other goals that justify partial\-/proportionality are speed of computation \citep{Edmonds2006Balanced,Edmonds2008Confidently}, improving the social welfare \citep{Zivan2011Can,Arzi2012Cake} and guaranteeing a minimum-length constraint of a 1-dimensional piece \citep{Caragiannis2011Towards}.

\section{Model and Terminology}
\label{sec:model}

The \emph{cake} $C$ is a Borel subset of the two-dimensional Euclidean plane $\mathbb{R}^{2}$. Usually $C$ is a polygonal domain. \emph{Pieces} are Borel subsets of $\mathbb{R}^{2}$. \emph{Pieces of $C$} are Borel subsets of $C$.

There is a pre-specified family $S$ of pieces that are considered \emph{usable}. An \emph{\spiece{}} is an element of $S$. In the present paper, $S$ is usually the family of squares or fat rectangles. 

$C$ has to be divided among $n\geq1$ \emph{agents}. Each agent $i\in\left\{ 1,...,n\right\}$ has a value-density function $v_{i}$, which is an integrable, non-negative and bounded function on $C$. The \emph{value} of a piece $X_i$ to agent $i$ is marked by $V_i(X_i)$ and it is the integral of the value-density over the piece: 
\[
V_i(X_i)=\iint_{X_i}v_{i}(x,y)dxdy
\]
When $C$ is unbounded, we assume that the $v_i$ are nonzero only in a bounded subset of $C$. Hence the $V_i$ are always finite. The $V_i$ are also continuous --- a zero-area piece has a value of zero to all agents. This means that the value of a piece is the same whether or not it contains its boundary.

Based on $V_i$ and $S$ we define the following shape-based \emph{utility} function, which assigns to each piece $X_i\subseteq C$ the value of the most valuable usable piece contained in $X_i$:
\[
V_i^{S}(X_i)=\sup_{q\in S\,\text{{and}}\,q\subseteq X_i}V_i(q)
\]
For example, suppose $S$ is the family of squares. If Alice wants to build a square house but gets a piece $X_i$ that is not square, then she builds her house on the most valuable square contained in $X_i$. Hence her utility is the value of that most valuable square. 

The value function $V$ is additive, but the utility function $V^S$ is usually not additive. Hence, classic cake-cutting results, which require additivity, are not applicable. If the cake $C$ itself is an \spiece{}, then $V^S(C)=V(C)$; otherwise, usually $V^S(C)<V(C)$.

An \emph{$S$-allocation} is an $n$-tuple of \spiece{}s, $X=(X_1,...,X_n)$, one piece per agent, such that $X_{i}\subseteq C$ and the $X_{i}$ are pairwise-disjoint. Some parts of the cake may remain unallocated (free disposal is assumed). Since $X_{i}$ is an \spiece{}, $V^{S}(X_{i})=V(X_{i})$. 

The fairness of an allocation is determined by the agents' \emph{normalized} values. Values can be normalized in two ways: either divide them by the \emph{absolute} cake value for the agent and get $V_i(X_i)/V_i(C)$, or divide them by the \emph{relative }cake utility for the agent and get $V_i(X_i)/V_i^{S}(C)$. Throughout the paper absolute normalization is used, except in Subsection \ref{sub:Selection-algorithms} where relative normalization is used.

An allocation is called \emph{proportional} if the normalized value of every agent is at least $1/n$. Example \ref{exm:intro} shows that a proportional allocation is not always attainable (whether absolute or relative normalization is used). Hence, we define:
\begin{defn}
\emph{\label{def:abs-prop}(Absolute proportionality)} For a cake
$C$, a family of usable pieces $S$ and an integer $n\geq1$: 

(a) The \emph{proportionality level }of $C$, $S$ and $n$, marked
$\prop(C,n,S)$, is the largest fraction $r\in[0,1]$ such that, for every $n$ value measures $(V_i,...,V_n)$, there exists an $S$-allocation $(X_1,...,X_n)$ for which $\forall i:\,V_i(X_i)/V_i(C)\geq r$.\footnote{Shortly: $\prop(C,n,S)=\inf_{V}\sup_{X}\min_{i}V_{i}(X_{i})/V_{i}(C)$, where the infimum is on all combinations of $n$ value measures $(V_{1},...,V_{n})$, the supremum is on all $S$-allocations $(X_{1},...,X_{n})$ and the minimum is on all agents $i\in\{1,...,n\}$.}

(b) The \emph{same-value proportionality level }of $C$, $S$ and $n$, marked $\propsame(C,n,S)$, is the largest fraction $r\in[0,1]$ such that, for every single value measure $V$, there exists an $S$-allocation $(X_1,...,X_n)$ for which $\forall i:\,V(X_{i})/V(C)\geq r$.
\end{defn}
The analogous definition for relative proportionality is given in Subsection \ref{sub:Selection-algorithms}. 

Obviously, for every $C$, $S$ and $n$: $\prop(C,n,S)\leq\propsame(C,n,S)\leq1/n$.

Applying this notation, classic cake-cutting results \citep[e.g. ][]{Steinhaus1948Problem}
imply that for every cake $C$
\begin{align*}
\prop(C,\,n,\,All)=\propsame(C,\,n,\,All)=1/n
\end{align*}
where \emph{All} is the collection of all pieces. That is: when there are no geometric constraints on the pieces,
for every cake $C$ and every combination of $n$ continuous value
measures there is a division in which each agent receives a utility
of $1/n$, which is the best that can be guaranteed. One-dimensional
procedures with contiguous pieces \citep[e.g. ][]{Even1984Note}  prove that $\prop(Interval,\,n,\,intervals)=1/n$
and when translated to two dimensions they yield: 

\[
\prop(Rectangle,\,n,\,rectangles)=\propsame(Rectangle,\,n,\,Rectangles)=1/n
\]
However, these procedures do not consider constraints that are two-dimensional in nature, such as squareness. Such two-dimensional constraints are the focus of the present paper. 

Our challenge in the rest of this paper will be to establish bounds
on $\prop(C,n,S)$ and $\propsame(C,n,S)$ for various cake shapes
$C$ and piece families $S$. Two types of bounds are provided: 
\begin{itemize}
\item Impossibility results (upper bounds), of the form $\prop(C,n,S)\leq f(n)$
where $f(n)\in[0,1]$, are proved by showing a set of $n$ value measures
on $C$, such that in any $S$-allocation, the value of one or more
agents is \emph{at most} $f(n)$. Such bounds are established in Section
\ref{sec:Impossibility-Results}. 
\item Positive results (lower bounds), of the form $\prop(C,n,S)\geq g(n)$
where $g(n)\in[0,1]$, are proved by describing a division procedure which finds, for every set of $n$ value measures on $C$, an $S$-allocation in which the value of every agent is \emph{at least }$g(n)$. Such bounds are established in Sections \ref{sec:auctions}-\ref{sec:procedures}. 
\end{itemize}

\section{Impossibility Results}
\label{sec:Impossibility-Results}
Our impossibility results are based on the following scenario.
\begin{itemize}
\item The cake $C$ is a desert with only $k$ water-pools; the set of pools is denoted $P_k$.
\item Each pool in $P_k$ is a square with side-length $\epsilon>0$ containing $1$ unit of water.
\item There are $n$ agents with the same value measure: the value of a piece equals the total amount of water in the piece. So the value of each pool in $P_k$ is 1 and the total cake value is $k$.
\item We say that a piece $X_i$ is \textbf{supported by $P_k$} if $X_i$ contains strictly more than 1 unit of water from $P_k$. This implies that $X_i$ touches at least two pools of $P_k$.
\item We say that $P_k$ \textbf{supports $m$ squares} if there exists a collection of $m$ pairwise-disjoint squares each of which 
contains strictly more than one unit of water from $P_k$.
\end{itemize}
The latter definition implies the following lemma:
\begin{lemma} \label{lem:support}
A collection of $k$ pools supports at most $k-1$ squares.
\end{lemma}
\begin{proof}
Let $P_k$ be a collection of $k$ pools and suppose that it supports $m$ squares. This means that there exists a collection of $m$ pairwise-disjoint squares, each of which contains more than one unit of water from $P_k$. So the union of these squares contains strictly more than $m$ units of water from $P_k$. Since each pool in $P_k$ contains exactly one unit of water, necessarily $k\geq m+1$ so $m\leq k-1$.
\end{proof}

In each impossibility result, we present a set $P_k$ and prove that it supports at most $n-1$ squares. This implies that, in every allocation of $n$ pairwise-disjoint squares, at least one agent receives a piece not supported by $P_k$ --- a piece with at most 1 unit of water. The value of this agent is at most a fraction $1/k$ of the total cake value. This implies that $\propsame(C,\,n,\,Squares)\leq1/k$, which implies that $\prop(C,\,n,\,Squares)\leq1/k$. 

\subsection{Impossibility results for two, three and four walls}
We start with impossibility results for two agents.

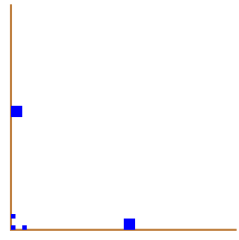
\begin{figure}
\psset{unit=1.5mm,dotsep=1pt}
\centering
\renewcommand{\squarepool}[1]{\Square[fillstyle=solid,fillcolor=blue,linecolor=blue]#1{1}} 
\newcommand{\quartplanecake}{
  \psline[linecolor=brown](0,20)(0,0)(20,0)
}
\begin{pspicture}(30,25)
\rput[l](0,24){\tiny{a. \prop(C,\n{2},Squares) $~\leq~ 1/3$}} 
\quartplanecake 
\squarepool{(0,0)} 
\squarepool{(10,0)} 
\squarepool{(0,10)} 
\rput(0.5,0.5){
\psframe[linestyle=dashed,linecolor=red](0,0)(10,10) 
\pspolygon[linestyle=dotted,linecolor=blue](10,0)(0,10)(10,20)(20,10) 
}
\rput(6,6){x} 
\end{pspicture}
\hskip 2cm
\begin{pspicture}(30,25)
\rput[l](0,24){\tiny{b. \prop(C,\n{3},Squares) $~\leq~ 1/5$}} 
\quartplanecake
\smallsquarepool{(0,0)}
\smallsquarepool{(1,0)}
\smallsquarepool{(0,1)}
\squarepool{(10,0)} 
\squarepool{(0,10)} 
\end{pspicture}

\protect\caption{\label{fig:impossibility-quart} Impossibility results in a quarter-plane
cake for $n=2$ and $n=3$ agents.}
\end{figure}

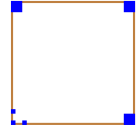
\begin{figure}
\psset{unit=1.5mm,dotsep=1pt}
\centering
\newcommand{\squarecake}{
  \psframe[linecolor=brown](0,0)(11,11)
}
\begin{pspicture}(30,15)
\rput[l](0,14){\tiny{a. \prop(C,\n{2},Squares) $~\leq~ 1/4$}} 
\squarecake
\squarepool{(0,0)} 
\squarepool{(10,0)} 
\squarepool{(0,10)} 
\squarepool{(10,10)}
\rput(5.5,5.5){x} 
\psframe[linestyle=dashed,linecolor=red](1.5,0.5)(10.5,9.5)
\psframe[linestyle=dotted,linecolor=green](0.5,1.5)(9.5,10.5)
\end{pspicture}
\hskip 2cm
\begin{pspicture}(30,15)
\rput[l](0,14){\tiny{b. \prop(C,\n{3},Squares) $~\leq~ 1/6$}} 
\squarecake
\smallsquarepool{(0,0)}
\smallsquarepool{(1,0)}
\smallsquarepool{(0,1)}
\squarepool{(10,0)} 
\squarepool{(0,10)} 
\squarepool{(10,10)}
\end{pspicture}

\protect\caption{\label{fig:impossibility-square} Impossibility results in a square cake
	for $n=2$ and $n=3$ agents.}
\end{figure}
 
\begin{claim}
\label{claim:neg-quart-2}
\[
\propsame(Quarter\,plane,\,\n{2},\,Squares)~\leq~1/3
\]
\end{claim}
\begin{proof}
Let $P_3$ be the set of 3 pools shown in Figure \ref{fig:impossibility-quart}/a,
where the bottom-left corners of the pools are in $(0,0)$, $(10,0)$, $(0,10)$. Every square in $C$ touching two pools of $P_3$ must contain e.g. the point $(6,6)$ in its interior (marked by x in the figure). Hence, every two squares touching two pools of $P_3$ must overlap. Hence, $P_3$ supports at most one square. Hence, in any allocation of squares to two agents, at least one square touches at most one pool of $P_3$; the agent receiving such a square has at most $1/3$ of the total value.
\end{proof}
\begin{claim}
\label{claim:neg-square-2}
\begin{align*}
\propsame(Square,\,\n{2},\,Squares)~\leq~1/4
\end{align*}
\end{claim}
\begin{proof}
Analogous to the previous claim, based on the set $P_4$ shown in Figure \ref{fig:impossibility-square}/a.
\end{proof}
To extend these results to $n>2$ agents, we construct new sets of pools by shrinking existing sets into pools of other sets. 

As an example, consider $P_3$ from the proof of Claim \ref{claim:neg-quart-2}. Suppose the entire plane is shrunk (deflated) towards the origin. If the deflation factor is sufficiently large, all three pools of the shrunk $P_3$ are contained in $[0,\epsilon]\times[0,\epsilon]$,
which is a pool of the original $P_3$. The cake itself (the quarter-plane)
is not changed by the deflation. By adding the other two pools of $P_3$, namely $(10,0)$ and $(0,10)$, we get a larger pool set,
$P_5$, which is depicted in Figure \ref{fig:impossibility-quart}/b.
We already know that the shrunk $P_3$ supports at most one square. The additional two pools support at most one additional square, since there is at most one square touching two new pools or a new pool and a shrunk pool. Hence, $P_5$ supports at most two squares. This proves that $\propsame(Quarter\,plane,\,\n{3},\,Squares)~\leq~1/5$. The following claim generalizes this construction.
\begin{claim}
\label{claim:neg-quart-n}For every $n\geq1$: 
\begin{align*}
\propsame(Quarter\,plane,\,n,\,Squares)\leq\frac{1}{2n-1}
\end{align*}
 \end{claim}
\begin{proof}
\footnote{We are grateful to Boris Bukh for the idea underlying this proof.}It
is sufficient to prove that for every $n$ there is an arrangement of $2n-1$ pools in $C$ that supports at most $n-1$ squares. The proof is by induction on $n$. The base
case $n=1$ is trivial (and the case $n=2$ is Claim \ref{claim:neg-quart-2}).
For $n>2$, assume there is an arrangement of $2(n-1)-1$ pools that supports at most $n-2$ squares. Deflate the entire arrangement towards the origin until it is contained in $[0,\epsilon]\times[0,\epsilon]$, where $\epsilon>0$ is a sufficiently small constant.

Add two new pools with side-length $\epsilon$ cornered at $(10,0)$ and $(0,10)$. We now have an arrangement of $2n-1$ pools. Every square touching a new pool and another pool (either new or old), must contain e.g. the point $(6,6)$ in its interior, so every two such squares must overlap.
Hence, the additional pools support at most one additional square. All in all, the new arrangement of $2n-1$ pools supports at most $(n-2)+1 = n-1$ squares.
\end{proof}
The upper bound for two walls is also trivially true when the cake
is a square with three walls, since adding walls cannot increase the
proportionality:
\[
\propsame(Square\,with\,3\,walls,\,n,\,Squares)~\leq~\frac{1}{2n-1}
\]
The bound also holds for a square with 4 walls, but in this case a slightly tighter bound is true:
\begin{claim}
\label{claim:neg-square}For every $n\geq2$, 
\[
\propsame(Square\,with\,4\,walls,\,n,\,Squares)~\leq~\frac{1}{2n}
\]
\end{claim}
\begin{proof}
W.l.o.g. assume $C$ is the square $[0,10+\epsilon]\times[0,10+\epsilon]$.
Create the arrangement of $2(n-1)-1$ pools from the induction step
of Claim \ref{claim:neg-quart-n}. Deflate it into to $[0,\epsilon]\times[0,\epsilon]$. The shrunk collection supports at most $n-2$ squares.
Add \emph{three} new pools with side-length $\epsilon$ cornered at
$(10,0)$, $(0,10)$ and $(10,10)$, as in Figure \ref{fig:impossibility-square}/b.
 Every square in $C$ touching a new pool and another pool must contain $(5,5)$ in its
interior. Hence, the three additional pools allow us to support at most one additional square. All in all, the new arrangement of $2n$ pools supports at most $n-1$ squares.
\end{proof}

\subsection{Impossibility results for one wall}

\newcommand{\halfplanecake}{
  \psline[linecolor=brown](-20,0)(0,0)(20,0)
}
\begin{figure}
\psset{unit=1.5mm,dotsep=1pt}
\centering
\begin{pspicture}(50,25) 
\rput(20,0){
\rput(0,22){\tiny{a. $\prop(C,\n{3},Squares)\leq 1/4$}} 
\halfplanecake 
\squarepool{(0,0)} 
\squarepool{(5,0)} 
\squarepool{(0,10)} 
\squarepool{(-5,0)} 
\rput(0.5,0.5){
\psframe[linestyle=dashed,linecolor=red](0,0)(5,5) \pspolygon[linestyle=dotted,linecolor=blue](5,0)(0,10)(10,15)(15,5) 
\psframe[linestyle=dashed,linecolor=red](0,0)(-5,5) \pspolygon[linestyle=dotted,linecolor=blue](-5,0)(0,10)(-10,15)(-15,5) 
\pspolygon[linestyle=dotted,linecolor=gray](0,0)(4.9,4.9)(0,9.8)(-4.9,4.9)
}
\rput(4,4.5){x} 
\rput(-4,4.5){-x} 
}
\end{pspicture}
\begin{pspicture}(50,25) 
\rput(20,0){
\rput(0,22){\tiny{b. $\prop(C,\n{5},Squares)\leq 1/7$}} 
\halfplanecake 
\smallsquarepool{(0,0)}
\smallsquarepool{(1,0)}
\smallsquarepool{(0,2)}
\smallsquarepool{(-1,0)}
\squarepool{(5,0)} 
\squarepool{(0,10)} 
\squarepool{(-5,0)} 
}
\end{pspicture}

\protect\caption{\label{fig:impossibility-half-plane} Impossibility results in a half-plane cake for $n=3$ and $n=5$ agents. See Claims \ref{claim:neg-half-3}-\ref{claim:neg-half-n}.}
\end{figure}
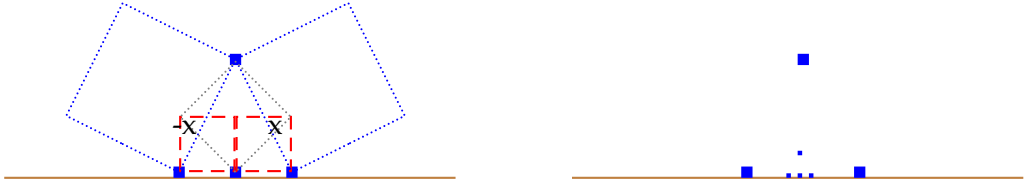

\begin{claim}
\label{claim:neg-half-3}
\[
\propsame(Half\,plane,\,\n{3},\,Squares)\leq1/4
\]
\end{claim}
\begin{proof}
Let $P_4$ be the set of 4 pools shown in Figure \ref{fig:impossibility-half-plane}/a. Assume the side-length of each pool is $\epsilon\leq0.01$ and that their bottom-left corner is in $(-5,0)$, $(0,0)$,  $(0,10)$, $(5,0)$. We prove that $P_4$ supports at most 2 squares. Examine the squares in $C$ that touch two pools of $P_4$:
\begin{itemize}
\item Every square touching $(5,0)$ and another pool must contain the point x $(4,4.5)$ in its interior.
\item Every square touching $(-5,0)$ and another pool must contain the point -x $(-4,4.5)$.
\item Every square touching $(0,0)$ and another pool must touch either x or -x or both.
\end{itemize}
Hence, in every set of three squares, each of which touches two pools of $P_4$, at least two squares must overlap. Hence, $P_4$ supports at most two squares. Hence, in any allocation to three agents, at least one of them receives at most 1/4 of the total value.
\end{proof}
\begin{claim}
\label{claim:neg-half-n}For every $n\geq2$:
\[
\propsame(Half\,plane,\,n,\,Squares)\leq\frac{1}{(3/2)n-1}
\]
\end{claim}
\begin{proof}
The proof is analogous to that of Claim \ref{claim:neg-quart-n}. With
each induction step, the current arrangement of pools is shrunk towards the central pool at the origin, \emph{three} new pools are added, but only \emph{two} new squares are supported. Hence the coefficient of $n$ is $3/2$.  The $-1$ ensures that the right-hand side is a correct upper bound for every $n\geq 2$.
\end{proof}
Figure \ref{fig:impossibility-half-plane}/b shows the set of 7 pools for the case $n=5$.

\subsection{Impossibility results for zero walls}
\psset{unit=7mm,dotsep=1pt}
\def\maxy{2.5}
\def\maxrot{18}
\def\maxyR{5}
\newcommand{\axes}{
	\psline[linestyle=dotted,linecolor=black!30](0,-5)(0,5)
	\psline[linestyle=dotted,linecolor=black!30](-5,0)(5,0)
}
\newcommand{\axeslong}{
	\psline[linestyle=dotted,linecolor=black!30](0,-5)(0,5)
	\psline[linestyle=dotted,linecolor=black!30](-15,0)(15,0)
}
\newcommand{\pfour}{
	\squarepoolwithlabel{(-3,0)}{B}
	\squarepoolwithlabel{(-1,0)}{C}
	\squarepoolwithlabel{(1,0)}{C'}
	\squarepoolwithlabel{(3,0)}{B'}
}
\newcommand{\psix}{
	\pfour{}
	\squarepoolwithlabel{(0,\maxy)}{A}
	\squarepoolwithlabel{(0,-\maxy)}{A'}
}
\newcommand{\psixnoletters}{
	\squarepoolwithlabel{(-3,0)}{}
	\squarepoolwithlabel{(-1,0)}{}
	\squarepoolwithlabel{(1,0)}{}
	\squarepoolwithlabel{(3,0)}{}
	\squarepoolwithlabel{(0,\maxy)}{}
	\squarepoolwithlabel{(0,-\maxy)}{}
}

\newcommand{\psixteen}{
	\pfour
	\rput(0,\maxy){
		\scalebox{0.1}{
		\rput{0}{\psix{}}
		}
	}
	\rput(0,-\maxy){
		\scalebox{0.1}{
		\rput{0}{\psix{}}
		}
	}
}
\newcommand{\psixR}{
	\pfour{}
	\squarepoolwithlabel{(0,\maxyR)}{A}
	\squarepoolwithlabel{(0,-\maxyR)}{A'}
}

Finding an impossibility result for an unbounded cake is a  challenging task. The main difficulty is that, when there are no walls, any arrangement of pools can be rotated arbitrarily, as will be explained shortly.

We begin with impossibility results for the restricted case in which the squares must be parallel to a specific coordinate system. Such a restriction may be meaningful, for example, in the installation of solar power-plants or the building of houses with electric solar panels, where the positioning relative to the sun is important.


\begin{claim}
Given a fixed coordinate system in the plane:
\label{claim:neg-plane-5}
\begin{align*}
\propsame(Plane,\,\n{5},\,Axes\,Parallel\,Squares) \leq 1/6
\end{align*}

\end{claim}
\begin{proof}
Let $P_6$ be the set of 6 pools:
A(0,\maxy), B(-3,0), C(-1,0), C'(1,0), B'(3,0), A'(0,-\maxy). We prove that $P_6$ supports at most 4 axes-parallel squares. First, consider the squares that touch two pools of $P_6$: 
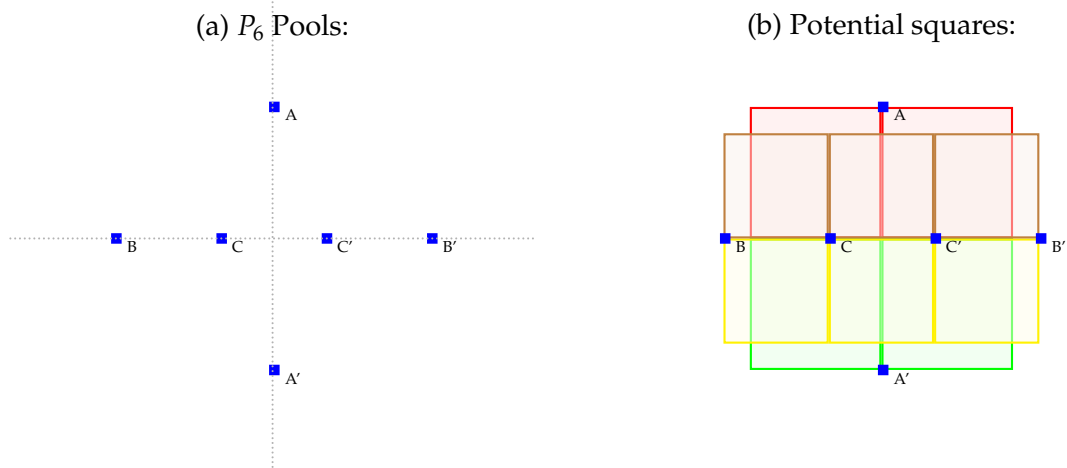
\begin{center}
\begin{pspicture}(-5,-4)(5,4.5) 
\rput(0,4){(a) $P_6$ Pools:}
\rput(-0.05,-0.15){
\psix{}
}
\axes{}
\end{pspicture}
\hskip 1cm
\begin{pspicture}(-5,-4)(5,4.5)
\rput(0,4){(b) Potential squares:}
\psset{opacity=0.5}
\Square[fillstyle=solid,fillcolor=red!10,linecolor=red](-\maxy,0){\maxy}
\Square[fillstyle=solid,fillcolor=red!10,linecolor=red](0,0){\maxy}
\Square[fillstyle=solid,fillcolor=green!10,linecolor=green](-\maxy,-\maxy){\maxy}
\Square[fillstyle=solid,fillcolor=green!10,linecolor=green](0,-\maxy){\maxy}
\Square[fillstyle=solid,fillcolor=yellow!10,linecolor=yellow](-3,-2){2}
\Square[fillstyle=solid,fillcolor=yellow!10,linecolor=yellow](-1,-2){2}
\Square[fillstyle=solid,fillcolor=yellow!10,linecolor=yellow](1,-2){2}
\Square[fillstyle=solid,fillcolor=brown!10,linecolor=brown](-3,0){2}
\Square[fillstyle=solid,fillcolor=brown!10,linecolor=brown](-1,0){2}
\Square[fillstyle=solid,fillcolor=brown!10,linecolor=brown](1,0){2}
\psset{opacity=1}
\rput(-0.05,-0.15){
\psix{}
}
\end{pspicture}
\end{center}
We can ignore squares that contain other squares or that contain pools in their interior, since such squares can be shrunk without interfering with other squares. Hence, any set of supported squares must contain a subset of the following:
\begin{itemize}
\item At most two disjoint ``top squares'' (squares touching pool A) and two disjoint ``bottom squares'' (touching pool A'). Each such square has a side-length of $\maxy$.
\item At most one ``left square'' (touching pools B and C), one ``right square'' (touching pools B' and C') and one ``central square'' (touching C and C'). Each such square has a side-length of 2 and can be located anywhere between $y=-2$ and $y=2$.
\footnote{
	While there can two disjoint squares touching pools B+C, Lemma \ref{lem:support} implies that the pools B+C can support at most one square. The same is true for the pools B'+C' and C+C'.
}
\end{itemize}
We prove that at most four of these squares can be supported simultaneously. There are two cases:

\textbf{Case \#1:} there are no bottom squares. The pool A' is not used, so only 5 pools are used. By Lemma \ref{lem:support}, these pools can support at most 4 squares. The situation is similar if there are no top squares, since in this case the pool A is not used.

\textbf{Case \#2:} there is at least one bottom square (e.g, a square supported by A' and C') and at least one top square (e.g, supported by A and C). These two squares leave no room for a central square. Hence, there is room for at most two additional squares: one above the x axis (e.g, supported by A and C', or C' and B'), and one below the x axis (e.g, supported by A' and C, or C and B).

In all cases, $P_6$ supports at most 4 axes-parallel squares.
\end{proof}

\begin{claim}
\label{claim:neg-plane-16}
Given a fixed coordinate system in the plane, for every $k\geq 0$:
\begin{align*}
\propsame(Plane,\,\n{5+9k},\,Axes\,Parallel\,Squares) ~~\leq~~ 1/(6+10 k)
\end{align*}
\end{claim}
\begin{proof}
We prove that for every $k\geq 0$, there exists an arrangement of $6+10k$ pools that supports at most $4+9k$ axes-parallel squares. The proof is by induction on $k$. The base $k = 0$ is proved by $P_6$ from Claim \ref{claim:neg-plane-5}. Assume that there exists an arrangement $P_{6+10(k-1)}$ which supports at most $9+4k$ squares. Construct a new arrangement $P_{6+10k}$ in the following way. Take $P_6$, replace the pool A a with shrunk copy of $P_6$ and the pool A' with a shrunk copy of $P_{6+10(k-1)}$. The following illustration shows $P_{16}$, the arrangement for $k=1$ (the shrunk copies are enlarged for the sake of clarity):
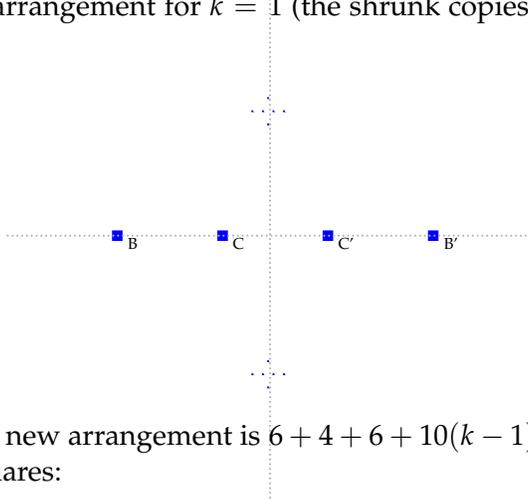
\begin{center}
\begin{pspicture}(-5,-3)(5,3) 
\rput{0}(-0.05,-0.15){
\psixteen
}
\axes
\end{pspicture}
\end{center}
The number of pools in the new arrangement is $6+4+6+10(k-1)=6+10k$. We claim that it supports at most $4+9k$ squares:
\begin{itemize}
\item The shrunk copy of $P_6$ supports at most 4 squares;
\item The shrunk copy of $P_{6+10(k-1)}$ supports at most $4+9(k-1)$ squares, by the induction assumption;
\item The four pools B C C' B' in the large $P_6$ support at most 3 large squares;
\item If there is a top large square then there is at most one additional large square above the x axis, and if there is a bottom large square then there is at most one additional large square below the x axis.
\end{itemize}
All in all, at most $4+4+9(k-1)$ squares are supported by the shrunk copies and at most 3+2=5 additional large squares are supported by the outer arrangement, so the total number of supported squares is at most $4+9k$. 
\end{proof}
In general, every 10 additional pools support at most 9 additional squares. Hence:
\begin{align*}
\propsame(Plane,\,n,\,Axes\,Parallel\,Squares) \leq {1 \over (10/9)n - 1} \approx \frac{9}{10}\cdot \frac{1}{n}
\end{align*}
This implies that any division procedure which works in a pre-specified coordinate system cannot guarantee a proportional division of the plane with square pieces. 

In our next results, we relax the axes-parallel restriction and only require that the squares be parallel to each other. While this is still not the most general setting, it is natural e.g. in urban planning. Equivalently, we still require that the squares be parallel to the axes, but allow the arrangement of pools to rotate.

Note that the proof of Claim \ref{claim:neg-plane-5} (Case 2) relies on the fact that any pair of a top-square and a bottom-square leaves no room for a central square. This follows from the facts that A and A' lie horizontally between C and C', and the horizontal distance between C and C' is larger than the vertical distance between B and B'. These facts are still true if the entire arrangement is rotated by at most $\maxrot^\circ$ to either direction:\footnote{The calculation was done using GeoGebra \citep{GeoGebra,GeoGebra5}. The worksheet is available here:
https://tube.geogebra.org/m/zzNY3ag4 }

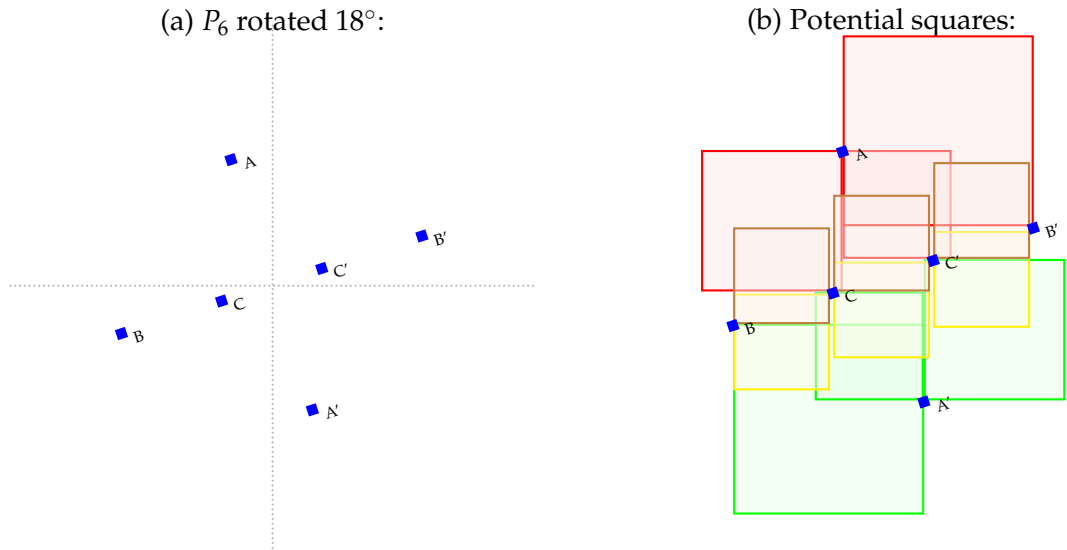
\begin{center}
\begin{pspicture}(-5,-5)(5,6) 
\rput(0,5){(a) $P_6$ rotated $\maxrot^\circ$:}
\rput{\maxrot}(-0.05,-0.15){
\psix{}
}
\axes{}
\end{pspicture}
\hskip 1cm
\begin{pspicture}(-5,-5)(5,6)
\rput(0,5){(b) Potential squares:}
\psset{opacity=0.5}
\rput(0.1,0.2){
\Square[fillstyle=solid,fillcolor=red!10,linecolor=red](-3.46,-0.31){2.69}  
\Square[fillstyle=solid,fillcolor=red!10,linecolor=red](-0.77,0.31){2.07} 
\Square[fillstyle=solid,fillcolor=red!10,linecolor=red](-0.77,0.93){3.63} 
\Square[fillstyle=solid,fillcolor=green!10,linecolor=green](-2.85,-4.55){3.63} 
\Square[fillstyle=solid,fillcolor=green!10,linecolor=green](-1.3,-2.38){2.07} 
\Square[fillstyle=solid,fillcolor=green!10,linecolor=green](0.77,-2.38){2.69}   
\Square[fillstyle=solid,fillcolor=yellow!10,linecolor=yellow](-2.85,-2.19){1.84}
\Square[fillstyle=solid,fillcolor=yellow!10,linecolor=yellow](-0.95,-1.58){1.84}
\Square[fillstyle=solid,fillcolor=yellow!10,linecolor=yellow](0.95,-1){1.84}
\Square[fillstyle=solid,fillcolor=brown!10,linecolor=brown](-2.85,-0.93){1.84}
\Square[fillstyle=solid,fillcolor=brown!10,linecolor=brown](-0.95,-0.31){1.84}
\Square[fillstyle=solid,fillcolor=brown!10,linecolor=brown](0.95,0.31){1.84}
}
\rput{\maxrot}{
\psset{opacity=1}
\psix{}
}
\end{pspicture}
\end{center}
For every angle $\theta$, define $ParallelSquares[\theta]$ as the family of squares rotated at exactly $\theta$ degrees (counter-clockwise) relative to the axes. Then, the proofs of Claim \ref{claim:neg-plane-5} and \ref{claim:neg-plane-16} and the above explanation imply:
\begin{claim} For every $\theta\in[-\maxrot^\circ,+\maxrot^\circ]$ and every $k\geq 0$:
\begin{align*}
\propsame(Plane,\,\n{5+9k},\,ParallelSquares[\theta])~~ \leq ~~ 1/(6+10 k)
\end{align*}
\end{claim}
The arrangement $P_{6+10k}$ ``covers'' a range of rotation-angles of size $36^\circ$. By using three copies of $P_{6+10k}$ rotated in different angles, we can cover the entire range of relevant rotation angles. We use this idea to prove an impossibility result for rotated parallel squares.
\begin{claim}
\label{claim:neg-plane-31}
For every $k\geq 0$:
\begin{align*}
\propsame(Plane,\,\n{18+29k},\,Parallel Squares)~~ \leq~~ 1/(18+30 k)
\end{align*}
\end{claim}
\begin{proof}
Construct an arrangement $P_{18+30k}$ from three copies of $P_{6+10k}$:
\begin{itemize}
\item A \emph{leftmost copy} --- rotated by $-27^\circ$ and translated by $(-300,0)$;
\item A \emph{central copy} --- not rotated;
\item A \emph{rightmost copy} --- rotated by $+27^\circ$ and translated by $(+300,0)$.
\end{itemize}
The following illustration shows $P_{18}$ (the construction for $k=0$) with the three copies enlarged for the sake of clarity:
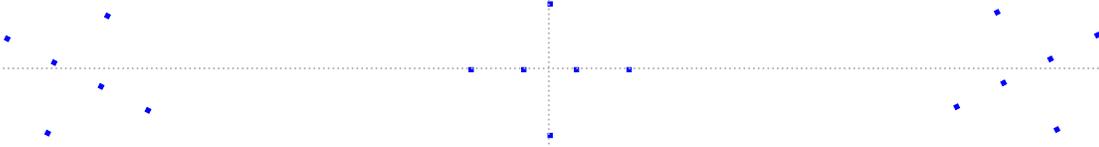
\begin{center}
\begin{pspicture}(-12,-2)(12,2)
\rput{-27}(-9.05,-0.15){
\scalebox{0.5}{\psixnoletters}
}
\rput(-0.05,-0.1){
\scalebox{0.5}{\psixnoletters}
}
\rput{27}(+9.05,-0.15){
\scalebox{0.5}{\psixnoletters}
}
\axeslong
\end{pspicture}
\end{center}

We claim that if $P_{18+30k}$ is rotated by any angle $\theta\in[-45^\circ,45^\circ]$, then the rotated arrangement supports at most $18+29k$ axes-parallel squares. Consider three cases:

(a) $P_{18+30k}$ is rotated by $\theta\in [-45^{\circ},-9^{\circ}]$. Then, the rightmost copy is $P_{6+10k}$ rotated by $\theta+27^\circ\in[-18^\circ,18^\circ]$, so it supports at most $4+9k$ squares. 

(b) $P_{18+30k}$ is rotated by $\theta\in [-18^{\circ},+18^{\circ}]$. Then the central copy supports at most $4+9k$ squares. 

(c) $P_{18+30k}$ is rotated by $\theta\in [+9^{\circ},+45^{\circ}]$. Then the leftmost copy is $P_{6+10k}$ rotated by $\theta-27^\circ\in[-18^\circ,18^\circ]$, so it supports at most $4+9k$ squares.

In all cases, one of the copies supports at most $4+9k$ squares. Each of the other two copies has $6+10k$ pools, so by Lemma \ref{lem:support} it supports at most $5+10k$ squares. Additionally, between the three copies there can be at most four (huge) pairwise-disjoint squares: two above and two below the x axis. All in all, the number of supported squares is at most $(4+9k)+(5+10k)+(5+10k)+4=18+29k$.

Therefore, for any angle $\theta\in[-45^\circ,45^\circ]$, if the family $S$ of usable pieces is the family of squares rotated by $\theta$, then  $P_{18+30k}$ supports at most $18+29k$ \spieces{}. But, \emph{any} square is identical to a square rotated by $\theta\in[-45^\circ,45^\circ]$. Therefore, the existence of $P_{18+30k}$ proves the claim.
\end{proof}
In Claim \ref{claim:neg-plane-31}, for every 30 new pools, at most 29 new squares can be supported. Therefore, 
\begin{claim}
For every $n\geq 1$:
\label{claim:neg-plane-61-n}
\begin{align*}
\propsame(Plane,\,n,\,Parallel\,Squares) \leq {1 \over (30/29)n - 1} \approx \frac{29}{30}\cdot \frac{1}{n}
\end{align*}
\end{claim}

\subsection{Impossibility results with fat rectangles}
\label{sub:negative-fat}
Our impossibility results so far have assumed that $S$ is the family of squares. One could think that allowing fat rectangles, instead of just squares, can overcome these impossibility results. But this is not necessarily true. Claim \ref{claim:neg-quart-2} holds as-is for $R$-fat rectangles:

\begin{claim}
\label{claim:neg-quart-2-fatrect}For every finite $R\geq1$: 
\[
\propsame(Quarter\,plane,\,2,\,R\,fat\,rectangles)\leq1/3
\]
\end{claim}
\begin{proof}
Let $P_{3}$ be the arrangement of 3 pools from the proof of Claim \ref{claim:neg-quart-2}:
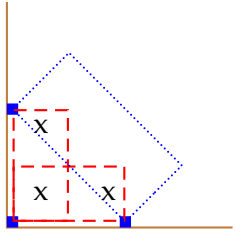
\begin{center}
\psset{unit=1.5mm,dotsep=1pt}
\renewcommand{\squarepool}[1]{\Square[fillstyle=solid,fillcolor=blue,linecolor=blue]#1{1}} 
\newcommand{\quartplanecake}{
  \psline[linecolor=brown](0,20)(0,0)(20,0)
}
\begin{pspicture}(30,20) 
\quartplanecake 
\squarepool{(0,0)} 
\squarepool{(10,0)} 
\squarepool{(0,10)} 
\rput(0.5,0.5){
\psframe[linestyle=dashed,linecolor=red](0,0)(10,5) 
\psframe[linestyle=dashed,linecolor=red](0,0)(5,10) 
\pspolygon[linestyle=dotted,linecolor=blue](10,0)(0,10)(5,15)(15,5) 
}
\rput(3,3){x} 
\rput(3,9){x} 
\rput(9,3){x} 
\end{pspicture}
\end{center}
 The side-length of each pool is $\epsilon>0$. Every $R$-fat rectangle touching the two \South pools must have
a height of at least $(10-2\epsilon)/R$ and thus, when $\epsilon$ is sufficiently small, it must contain the point $(5/R,5/R)$ and the point $(10-10/R,\,5/R)$. Every $R$-fat rectangle touching the two \West pools must contain the point $(5/R,5/R)$ and the point
$(5/R,\,10-10/R)$. Every $R$-fat rectangle touching the top-left and the bottom-right pools must contain $(10-10/R,\,5/R)$ and $(5/R,\,10-10/R)$. Hence, in every allocation of disjoint $R$-fat rectangles, at most one rectangle touches two or more pools.
\end{proof}
Claim \ref{claim:neg-quart-n} is based on Claim \ref{claim:neg-quart-2},
so it holds as-is for $R$-fat rectangles. The same is true for the 3-walls result. The 1-wall claims \ref{claim:neg-half-3} and \ref{claim:neg-half-n} can be generalized in a similar way:
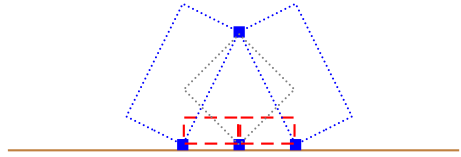
\begin{center}
\psset{unit=1.5mm,dotsep=1pt}
\begin{pspicture}(50,16)
\rput(20,0){
\halfplanecake 
\squarepool{(0,0)} 
\squarepool{(5,0)} 
\squarepool{(0,10)} 
\squarepool{(-5,0)} 
\rput(0.5,0.5){
\psframe[linestyle=dashed,linecolor=red](0,0)(5,2.5) \pspolygon[linestyle=dotted,linecolor=blue](5,0)(0,10)(5,12.5)(10,2.5) 
\psframe[linestyle=dashed,linecolor=red](0,0)(-5,2.5) \pspolygon[linestyle=dotted,linecolor=blue](-5,0)(0,10)(-5,12.5)(-10,2.5) 
\pspolygon[linestyle=dotted,linecolor=gray](0,0)(4.9,4.9)(0,9.8)(-4.9,4.9)
}
}
\end{pspicture}
\end{center}
We omit the details. We obtain:
\begin{claim}
\label{claim:neg-fatrect}For every $R\geq1$:
\end{claim}
\begin{align*}
\propsame(Square\,with\,1\,wall,\,n,\,R\,fat\,rectangles)\leq\frac{1}{(3/2)n-1}
\\
\propsame(Square\,with\,2\,walls,\,n,\,R\,fat\,rectangles)\leq\frac{1}{2n-1}
\\
\propsame(Square\,with\,3\,walls,\,n,\,R\,fat\,rectangles)\leq\frac{1}{2n-1}
\end{align*}
Claims \ref{claim:neg-square-2} and \ref{claim:neg-square} hold whenever $R<2$, since in this case, every $R$-fat rectangle touching one of the corner-pools must contain the central point of the cake in its interior, as shown below:
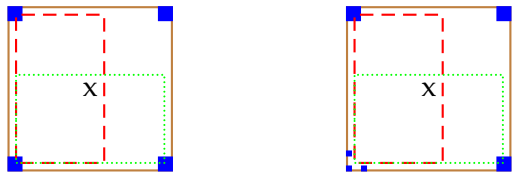
\psset{unit=2mm,dotsep=1pt}
\begin{center}
\newcommand{\squarecake}{
	\psframe[linecolor=brown](0,0)(11,11)
}
\begin{pspicture}(12,12)
\squarecake
\squarepool{(0,0)} 
\squarepool{(10,0)} 
\squarepool{(0,10)} 
\squarepool{(10,10)}
\rput(5.5,5.5){x} 
\psframe[linestyle=dashed,linecolor=red](0.5,0.5)(6.5,10.5)
\psframe[linestyle=dotted,linecolor=green](0.5,0.5)(10.5,6.5)
\end{pspicture}
\hskip 2cm
\begin{pspicture}(12,12)
\rput[l](0,14){\tiny{b. \prop(C,3,Squares) $\leq 1/6$}} 
\squarecake
\smallsquarepool{(0,0)}
\smallsquarepool{(1,0)}
\smallsquarepool{(0,1)}
\squarepool{(10,0)} 
\squarepool{(0,10)} 
\squarepool{(10,10)}
\rput(5.5,5.5){x} 
\psframe[linestyle=dashed,linecolor=red](0.5,0.5)(6.5,10.5)
\psframe[linestyle=dotted,linecolor=green](0.5,0.5)(10.5,6.5)
\end{pspicture}
\end{center}
This gives:
\begin{claim}
\label{claim:neg-square-fatrects-R<2}For every $R$ such that $1\leq R<2$:
\begin{align*}
\propsame(Square\,with\,4\,walls,\,n,\,R\,fat\,rectangles)\leq\frac{1}{2n}
\end{align*}
\end{claim}
When $R\geq 2$, the following slightly weaker result follows immediately from Claim \ref{claim:neg-fatrect}
(since adding walls cannot increase the proportionality):
\begin{claim}
\label{claim:neg-square-fatrects}For every $R\geq 2$: \footnote{
By classic cake-cutting protocols, $\propsame(Square,\,n,\,\infty\,fat\,rectangles)=1/n$ (an $\infty$-fat rectangle is just an arbitrary rectangle). The PropSame function is thus discontinuous at $R=\infty$. If the agents agree to use any rectangular piece, they can receive their proportional share of $1/n$, but if they insist on using $R$-fat rectangles, even when $R$ is very large, they might have to settle for about half of this share.}
\begin{align*}
\propsame(Square\,with\,4\,walls,\,n,\,R\,fat\,rectangles)\leq\frac{1}{2n-1}
\end{align*}
\end{claim}

The impossibility results for an unbounded plane are different for $R$-fat rectangles. Consider first Claim \ref{claim:neg-plane-5}, which assumes that the pieces must be axes-parallel. When the pieces have to be squares, the set $P_6$ supports at most 2 pieces above the x axis and 2 pieces below the x axis. But when the pieces may be $R$-fat rectangles and $R\geq \maxy$, it is possible to support 3 pieces above or below the x axis, e.g:
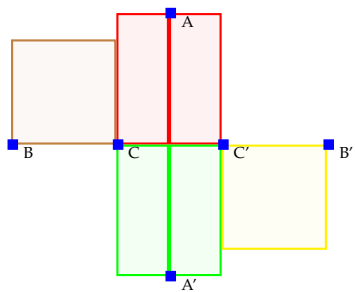
\begin{center}
\psset{unit=7mm,dotsep=1pt}
	\begin{pspicture}(-5,-2.7)(5,2.7)
	\psset{opacity=0.5}
	\psframe[fillstyle=solid,fillcolor=red!10,linecolor=red](-1,0)(0,\maxy)
	\psframe[fillstyle=solid,fillcolor=red!10,linecolor=red](0,0)(1,\maxy)
	\psframe[fillstyle=solid,fillcolor=green!10,linecolor=green](-1,0)(0,-\maxy)
	\psframe[fillstyle=solid,fillcolor=green!10,linecolor=green](0,0)(1,-\maxy)
	\Square[fillstyle=solid,fillcolor=yellow!10,linecolor=yellow](1,-2){2}
	\Square[fillstyle=solid,fillcolor=brown!10,linecolor=brown](-3,0){2}
	\psset{opacity=1}
	\rput(-0.05,-0.15){
		\psix{}
	}
	\end{pspicture}
\end{center}
The impossibility result can be maintained by locating the pool $A$ at $(\maxy R,0)$ instead of $(\maxy,0)$, and the pool $A'$ at $(-\maxy R,0)$ instead of $(-\maxy ,0)$:
\begin{center}
\psset{unit=4mm,dotsep=1pt}
	\begin{pspicture}(-5,-5.4)(5,5.4)
	\psset{opacity=0.5}
	\psframe[fillstyle=solid,fillcolor=red!10,linecolor=red](-2,0)(0,\maxyR)
	\psframe[fillstyle=solid,fillcolor=red!10,linecolor=red](0,0)(2,\maxyR)
	\psframe[fillstyle=solid,fillcolor=green!10,linecolor=green](-2,0)(0,-\maxyR)
	\psframe[fillstyle=solid,fillcolor=green!10,linecolor=green](0,0)(2,-\maxyR)
	\Square[fillstyle=solid,fillcolor=yellow!10,linecolor=yellow](1,-2){2}
	\Square[fillstyle=solid,fillcolor=brown!10,linecolor=brown](-3,0){2}
	\psset{opacity=1}
	\rput(-0.05,-0.15){
		\psixR{}
	}
	\end{pspicture}
\end{center}
So Claim \ref{claim:neg-plane-5}, and hence Claim \ref{claim:neg-plane-16}, are valid for $R$-fat rectangles, and we obtain:
\begin{claim}
	\label{claim:neg-plane-16-fat}
	Given a fixed coordinate system in the plane, for every $R\geq 1$:
\begin{align*}
\propsame(Plane,\,n,\,Axes\,Parallel\,R\,fat\,rectangles) \leq {1 \over (10/9)n - 1} \end{align*}
\end{claim}
However, the angle-range in which Claim \ref{claim:neg-plane-16-fat} holds is no longer $[-\maxrot^\circ,\maxrot^\circ]$ --- the range becomes smaller as a (complicated) function of $R$. This means that more copies may be needed to ``cover'' the entire range of $[-45^\circ,45^\circ]$. Therefore, the upper bound for parallel squares will probably be a complicated function of $R$. We leave this issue for future work.

\section{Auctions and Covers}  \label{sec:auctions} \label{sub:Queries}
Our cake-cutting procedures are composed of two types of auctions. In a \emph{mark auction}, each agent bids by marking a piece of the cake; the winner is the agent marking the smallest piece. In an \emph{eval auction}, each agent bids by declaring a value for a pre-specified piece of cake; the winners are the agents declaring the highest value. 
As usual in the cake-cutting literature, no monetary transfers are involved; the agents effectively 'pay' with their entitlements for a share of the cake. Below we explain each auction type in detail.

\subsection{Mark auction}\label{sub:mark}
In a mark auction, the divider specifies a geometric constraint and a value $v$. Each agent has to mark a piece of the cake which satisfies the geometric constraint and is worth for him exactly $v$. The geometric constraint guarantees that the marked pieces are totally ordered by containment (i.e. for every two agents $i,j$, the bid of $i$ either contains or is contained in the bid of $j$). Hence, there is a \emph{smallest bid} --- a bid contained in all other bids. There can be more than one smallest bid; in this case, one smallest bid is selected arbitrarily. The agent making the selected smallest bid is the \emph{winner}; he is allocated his bid and goes home. The remaining cake is divided among the remaining $n-1$ agents.

\begin{example} \label{exm:mark} 
\textbf{Dividing a rectangle to rectangles.}
The cake $C$ is a rectangle and $S$ is the family of rectangles. We normalize the valuations of all agents such that the value of the entire cake is $n$. We show how a sequence of mark auctions can be used to give each agent a rectangle with a value of at least 1. 

The proof is by induction on the number of agents $n$. When $n=1$, $C$ can just be given to the single agent. Suppose we already know how to divide a rectangle to $n-1$ agents who value it as $n-1$. Now we are given $n$ agents who value the cake as $n$. We do a mark auction with the following geometric constraint: \emph{mark a rectangle whose rightmost edge coincides with the rightmost edge of $C$}. The auction value is $v=1$. The continuity of the valuations guarantees that all agents can indeed bid as required, and the geometric constraint guarantees that the bids are totally ordered by containment. An example is illustrated below, where there are four bids marked by dotted lines:
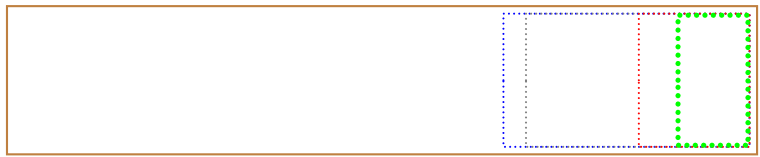
\begin{center}
\psset{unit=1mm,dotsep=0.4}
\begin{pspicture}(100,20)
\psframe[linestyle=solid,linecolor=brown](0,0)(100,20)
\psframe[linestyle=dotted,linecolor=blue](66,1)(99,19)
\psframe[linestyle=dotted,linecolor=gray](69,1)(99,19)
\psframe[linestyle=dotted,linecolor=red](84,1)(99,19)
\psframe[linestyle=dotted,linecolor=green, linewidth=2pt](89,1)(99,19)
\end{pspicture}
\end{center}
the winning bid --- the smallest rectangle --- is marked by a thick dotted line. The winner is given his bid, so he now has a rectangle with a value of exactly 1, as required (recall that our guarantees are valid for every agent bidding truthfully, regardless of what the other agents do). Since the $n-1$ losing bids contain the winning bid, the $n-1$ losers value the winning bid as at most 1. By additivity, they value the remaining cake as at least $n-1$. Hence, by the induction assumption we can divide the remaining cake among them in a similar way, finally giving each agent a rectangle with a value of at least 1.\footnote{Example \ref{exm:mark} shows that $\prop (Rectangle,n,Rectangles)=1/n$. This result is not new since it follows immediately from known results on 1-dimensional cake-cutting. It is presented here to show that it fits well into the auction framework.} \qed
\end{example}

A mark auction has the following interpretation. Initially, each agent holds an entitlement for a piece of cake. An agent bidding a piece $X_i$ is interpreted as saying ``I am willing to give my entitlement in exchange for piece $X_i$''. The agent marking the smallest piece is effectively offering the highest ``price'' per unit area; hence this agent is the winner. He pays for the win by giving up his entitlement and leaving the remaining cake to the remaining agents.

\subsection{Eval auction}\label{sub:eval}
In an eval auction, the divider specifies a piece of cake $C'\subset C$. Each agent $i$ has to declare the value $V_i(C')$. The agents are ordered in a descending order of their bids, such that $V_1(C')\geq V_2(C')\geq\cdots\geq V_n(C')$. The procedure calculates the number of winners $n'$ (we explain shortly how this number is calculated). The $n'$ highest bidders, $1,\dots,n'$, are the winners. The remaining $n-n'$ agents are the losers. The procedure then divides $C'$ among the winners and $C\setminus C'$ among the losers.

To calculate the number of winners $n'$, we should already have a plan for dividing $C'$ among each possible number of winners $n'\leq n$. Specifically, we should have a procedure for dividing $C'$ among $n'$ agents, each of whom values $C'$ as at least $F(n')$ (where $F: \mathbb{Z}^+\to \mathbb{R}^+$ is some increasing function), such that each agent is guaranteed a piece with a value of at least 1.
\footnote{
As explained in the introduction, an agent bidding a value $V$ is interpreted as saying ``I am willing to share $C'$ in a group of up to $f(V)$ agents'', where $f: \mathbb{R}^+\to \mathbb{Z}^+$. The function $f$ is an inverse of $F$ in the following sense: $f(V)$ is largest integer such that $V \geq F(f(V))$.
}
Assuming that we have such a procedure, the number of winners is defined as the largest integer $n'$ such that:
\begin{align*}
V_{n'}(C') \geq F(n')
\end{align*}
or 0 if already $V_1(C')<F(1)$.\footnote{$n'$ is somewhat analogous to the \emph{h-index} used to evaluate an academic researcher --- the largest integer $h$ such that the researcher has at least $h$ publications with at least $h$ citations.} Since $V_{n'}(C')$ is a decreasing sequence, the definition implies that:
\begin{itemize}
\item For every winner $i\in\{1,\dots,n'\}$: $V_i(C')\geq F(n')$
\item For every loser  $i\in\{n'+1,\dots,n\}$: $V_i(C')< F(n'+1)$
\end{itemize}
(this is true even when $n'=0$). Hence, the set of winners is a largest set of agents for whom we can divide $C'$ in a way which guarantees each of them a value of at least 1.

\begin{example} \label{exm:eval}
\textbf{Dividing an archipelago to rectangles.} 
The cake $C$ is an \emph{archipelago} --- a union of $m$ disjoint rectangular \emph{islands}. $S$ is the family of rectangles. We normalize the valuations of all agents such that the value of the entire archipelago is $n+m-1$. We show how a sequence of eval auctions can be used to give each agent a rectangle, contained in one of the islands, with a value of at least 1. 

The proof is by induction on the number of islands $m$. When $m=1$, $C$ is a single rectangle and all agents value it as at least $n$, so the procedure of Example \ref{exm:mark} can be used to give each agent a rectangle with a value of at least 1. Suppose we already know how to divide an archipelago of $m-1$ islands. Given an archipelago of $m$ islands, pick one island arbitrarily and call it $C'$. Do an eval auction on $C'$. Order the bids in descending  order, and let $n'$ be the largest index such that:
\begin{align*}
V_{n'}(C') \geq n'
\end{align*}
or 0 if already $V_1(C')<1$. If $n'=0$ then just discard $C'$; otherwise use the procedure of Example \ref{exm:mark} to divide $C'$ among the $n'$ winners. By definition, each winner values $C'$ as at least $n'$ so he is guaranteed a rectangular piece of $C'$ with a value of at least 1.

All $n-n'$ losers value $C'$ as less than $n'+1$, so they value the remaining archipelago $C\setminus C'$ as more than $(n+m-1)-(n'+1)=(n-n')+(m-1)-1$. This is an archipelago of $m-1$ islands, so by the induction assumption we can divide it among the remaining $n-n'$ agents giving each agent a rectangle with a value of at least 1. Note that this is true even when $n'=0$.\footnote{Example \ref{exm:eval} shows that $\prop(m\,\,disjoint\,rectangles,\,n,\,Rectangles)\geq 1/(n+m-1)$. It is easy to construct an arrangement of pools, analogous to the ones in Section \ref{sec:Impossibility-Results}, proving that this is the best proportionality that can be guaranteed.}
\qed
\end{example}
An eval auction has the following interpretation. Initially, each agent has an entitlement to share the entire cake $C$ with $n$ agents (including the agent himself). An agent bidding a value $V$ is interpreted as saying ``I am willing to give my entitlement in exchange for an entitlement to share $C'$ with at most $n'$ agents, where n' is the largest integer such that $V\geq F(n')$.'' The agents with the highest bids are actually offering a higher ``price'' for $C'$, since they are willing to share $C'$ with a larger number of other agents. Hence, the highest bidders are the winners. They pay for their win by giving up their entitlement to $C\setminus C'$ and leaving it to the remaining agents.

\subsection{Cover numbers} \label{sub:cover}
The last ingredient we need for our division procedures, in addition to the two auction types, is the \emph{cover number}. It is a well-known concept in computational geometry (see \citet{Keil2000Polygon} for a survey).
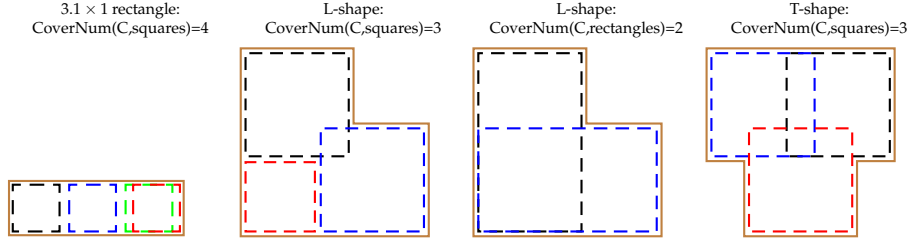
\begin{figure}
\psset{unit=0.5mm,linecolor=brown}
\centering
\begin{pspicture}(60,65)
\psframe(0,0)(47,15)
\psframe[linecolor=black,linestyle=dashed](1,1)(14,14)
\psframe[linecolor=blue,linestyle=dashed](16,1)(29,14)
\psframe[linecolor=green,linestyle=dashed](31,1)(44,14)
\psframe[linecolor=red,linestyle=dashed](33,1)(46,14)
\rput(30,60){\tiny{$3.1\times 1$ rectangle:}}
\rput(30,55){\tiny{CoverNum(C,squares)=4}}
\end{pspicture}
\begin{pspicture}(60,65)
\pspolygon(0,0)(50,0)(50,30)(30,30)(30,50)(0,50)
\psframe[linecolor=black,linestyle=dashed](1,49)(29,21)
\psframe[linecolor=blue,linestyle=dashed](49,1)(21,29)
\psframe[linecolor=red,linestyle=dashed](1,1)(20,20)
\rput(30,60){\tiny{L-shape:}}
\rput(30,55){\tiny{CoverNum(C,squares)=3}}
\end{pspicture}
\begin{pspicture}(60,65)
\pspolygon(0,0)(50,0)(50,30)(30,30)(30,50)(0,50)
\psframe[linecolor=black,linestyle=dashed](1,49)(29,1)
\psframe[linecolor=blue,linestyle=dashed](49,1)(1,29)
\rput(30,60){\tiny{L-shape:}}
\rput(30,55){\tiny{CoverNum(C,rectangles)=2}}
\end{pspicture}
\begin{pspicture}(60,65)
\pspolygon(10,0)(10,20)(0,20)(0,50)(50,50)(50,20)(40,20)(40,0)
\psframe[linecolor=black,linestyle=dashed](21,49)(49,21)
\psframe[linecolor=blue,linestyle=dashed](1,21)(29,49)
\psframe[linecolor=red,linestyle=dashed](11,1)(39,29)
\rput(30,60){\tiny{T-shape:}}
\rput(30,55){\tiny{CoverNum(C,squares)=3}}
\end{pspicture}

\protect\caption{\label{fig:cover-numbers}Cover numbers of various polygons.}
\end{figure}

\begin{defn}
\label{def:covernum}Let $C$ be a cake and $S$ a family of pieces.\\
(a) An \emph{$S$-cover of $C$} is a set of \spieces{}, possibly overlapping, whose union equals $C$.\\
(b) The \emph{$S$-cover number}\textbf{\emph{ }}\emph{of $C$}, $\covernum(C,S)$,
is the minimum cardinality of an $S$-cover of $C$.
\end{defn}
Some examples are depicted in Figure \ref{fig:cover-numbers}.

The cover number is related to the utility that a single agent can derive from a given cake:
\begin{lemma}
\label{lemma:pos-cover}(Covering Lemma) For every cake $C$ and family
$S$: 
\[
\prop(C,\n{1},S)\geq\frac{1}{\covernum(C,S)}
\]
\end{lemma}
\begin{proof}
Let $k=\covernum(C,S)$ and let $\{C_{1},...,C_{k}\}$ be an $S$-cover of $C$. By definition $\cup_{j=1}^{k}C_{j}=C$. By additivity, if an agent's valuation function is $V$, then:
\begin{align*}
\sum_{j=1}^{k}V(C_{j})\geq V(C)
\end{align*}
so the average value of the left-hand side is at least $V(C)/k$. By the properties of the average, at least one summand must be weakly larger than the average value, i.e, there exists $j$ for which $V(C_{j})\geq V(C)/k$. This $C_{j}$, which is an $S$-piece, gives the single agent a utility of at least $1/k$ of the total cake value.
\end{proof}

The next example combines an eval auction, a mark auction and the Covering Lemma.
\begin{example}\label{exm:square2}
\textbf{Dividing a square between two agents who want square pieces.}
The cake $C$ is a square, $S$ is the family of squares and there are $n=2$ agents. Example \ref{exm:intro} shows that the maximum utility that can be guaranteed to both agents is $1/4$ of the total value. We now present a division procedure that guarantees this utility. We normalize the valuations of both agents such that their value of $C$ is 4 and give each agent a square with a value of at least 1. 

\psset{unit=2cm}
\newcommand{\landcakesquare}{
 \psframe[linecolor=brown](0,0)(1,1)
}

Partition the cake to a $2\times2$ grid. Denote one of the four quarters as $C'$, e.g.:
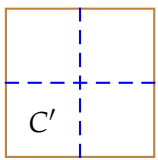
\begin{center}
\begin{pspicture}(1,1)
\landcakesquare
\psline[linecolor=blue,linestyle=dashed](0.5,0)(0.5,1) 
\psline[linecolor=blue,linestyle=dashed](0,0.5)(1,0.5) 
\rput(0.25,0.25){$C'$}
\end{pspicture}
\end{center}
Do an \textbf{eval auction} on $C'$. Let $n'$ be the number of agents whose bid is at least 1.

\emph{Case \#1:} $n'=0$ (both agents value $C'$ as less than 1). Denote another quarter as $C'$ and do an eval auction again. Because the total cake value is 4, this can happen at most three times; eventually one of the other cases must happen.

\emph{Case \#2:} $n'=1$. The single agent who values $C'$ as at least 1 wins $C'$ and goes home. The losing agent values $C'$ as less then 1 so he values $C\setminus C'$ as more than 3. $C\setminus C'$ is a union of 3 squares, so by the Covering Lemma the losing agent can get from it a square with a value of at least 1.

\emph{Case \#3:} $n'=2$. Do a \textbf{mark auction} with the following constraint: \emph{mark a square with a value of 1 contained in $C'$ and adjacent to a corner of $C$}. Both agents can bid as required, since they value $C'$ as at least 1 so they have a square with a value of exactly 1 inside $C'$. An example is illustrated below, where the two bids are marked by dotted lines:
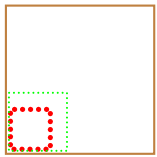
\begin{center}
\begin{pspicture}(1,1)
\landcakesquare
\Square[linecolor=red,linestyle=dotted,linewidth=2pt](.02,.02){.3}
\Square[linecolor=green,linestyle=dotted](.02,.02){.4}
\end{pspicture}
\end{center}
The winning bid (the smallest square) is marked with thicker dots. It is given to the winner, who walks home with a square worth 1. The remaining cake is an L-shape similar to the one in Figure \ref{fig:cover-numbers}. Its cover number is 3 and its value for the loser is at least 3. By the Covering Lemma, it contains a square whose value to the loser is at least 1.\footnote{Combining the lower bound proved by Example \ref{exm:square2} with the upper bound proved by Claim \ref{claim:neg-square-2} gives a tight result for two agents: $\prop(Square,\,\n{2},\,Squares)=1/4$.}  The final allocation may look like:
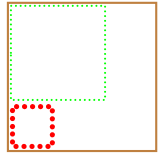
\begin{center}
\begin{pspicture}(1,1)
\landcakesquare
\Square[linecolor=red,linestyle=dotted,linewidth=2pt](.02,.02){.3}
\Square[linecolor=green,linestyle=dotted](.02,.34){.64}
\end{pspicture}
\end{center}
The fairness of this allocation is evident: both agents agree that the south-west is the most valuable district, so the agent who has to go to a less valuable district is compensated by a larger plot.

Note that some land remains unallocated. This is unavoidable if the pieces have to be square. Moreover, in realistic land-division scenarios it is common to leave some land unallocated and available for public use.
\qed 
\end{example}

\section{Division procedures}  \label{sec:procedures}
\psset{unit=2.6cm}
\newcommand{\cutatx}[2] {
	\psline[linestyle=solid,linecolor=blue](#1,-0.07)(#1,1)
	\rput(#1,-0.2){#2}
}
\newcommand{\tickx}[2] {
	\psline[linestyle=solid,linecolor=black](#1,0)(#1,-0.1)
	\rput(#1,-0.2){#2}
}
\newcommand{\landcakefourwalls}{
	\psframe[linecolor=brown,linewidth=2pt](0,0)(1.6,1)
	\tickx{0}{0}
	\tickx{1.6}{$L$}
}
\newcommand{\landcakethreewalls}{
	\psframe[linestyle=none](0,0)(0.8,1)
	\psline[linecolor=brown,linewidth=2pt](0.8,0)(0,0)(0,1)(0.8,1)
	\psline[linestyle=dotted,linecolor=brown](0.8,1)(0.8,0)
	\tickx{0}{$0$}
	\tickx{0.8}{$L$}
}

In this section we use the building-blocks developed in Section \ref{sec:auctions} to create various division procedures.
\subsection{Four and three walls, guillotine cuts} \label{sub:Guillotine-algorithms}
We develop simultaneously a pair of division procedures. Both procedures accept a cake $C$ which is assumed to be the rectangle $[0,L]\times[0,1]$, and return $n$ disjoint square pieces $\{X_i\}_{i=1}^{n}$ such that for every agent $i$: $V_i(X_i)\geq 1$.

The two procedures differ in their requirement on $L$ (the length/width ratio of the cake) and in the number of ``walls'' (bounded sides) they assume on the cake: 
\begin{itemize}
\item The \emph{3-walls procedure} requires that $L\in[0,1]$ and it guarantees that the allocated squares are contained in $[0,\infty]\times[0,1]$ (in other words, there is no wall in the rightmost edge of the cake).
\item The \emph{4-walls procedure} requires that $L\in[1,2]$ (i.e, the cake is a 2-fat rectangle) and it guarantees that all allocated squares are contained in $C$.
\end{itemize}
Additionally, the two procedures differ in their requirement on the total cake value:
\begin{itemize}
\item The 3-walls procedure requires that for every agent $i$: $V_i(C)\geq\max(1,4n-5)$.
\item The 4-walls procedure requires that for every agent $i$: $V_i(C)\geq\max(2,4n-4)$.
\end{itemize}
The procedures are developed by induction on the number of agents. We first consider the base case in which there is a single agent ($n=1$).

In the 3-walls procedure, the single agent values $C$ as at least 1. The square $[0,1]\times[0,1]$ contains all the value of $C$ and it is contained within its three walls, so it can be given to the single agent:
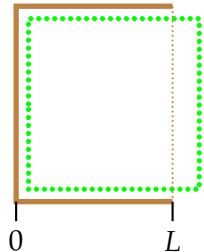
\begin{center}
	\begin{pspicture}(-1,-.3)(1.1,1)
	\Square[linestyle=dotted,linewidth=2pt,linecolor=green](0.05,0.05){0.9}
	\landcakethreewalls
	\end{pspicture}
\end{center}

In the 4-walls procedure, the single agent values $C$ as at least 2. The requirement on $L$ guarantees that the cake can be covered by at most 2 squares:
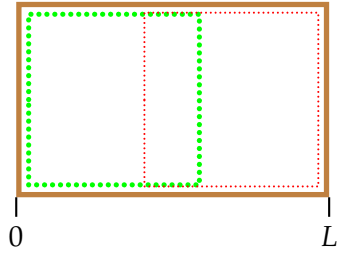
\begin{center}
	\begin{pspicture}(-1,-.5)(2,1)
	\landcakefourwalls
	\Square[linestyle=dotted,linewidth=2pt,linecolor=green](0.05,0.05){0.9}
	\Square[linestyle=dotted,linecolor=red](0.65,0.05){0.9}
	\end{pspicture}
\end{center}
Hence, by the Covering Lemma, the single agent can be given a square with a value of at least 1.

We now assume that we can handle any number of agents less than $n$. Given $n$ agents ($n\geq 2$), we proceed as follows.

\subsubsection{3 Walls procedure}
At this point, there are $n\geq 2$ agents who value the cake as at least $4n-5$. 

(1) \textbf{Mark auction}. Ask each agent to mark a rectangle with a value of exactly 1 adjacent to the rightmost edge of the cake (the edge without the wall):
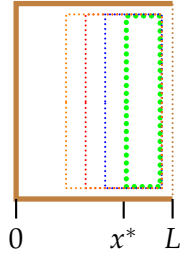
\begin{center}
	\begin{pspicture}(-1,-.3)(1.1,1)
	\psframe[linestyle=dotted,linewidth=2pt,linecolor=green](0.55,0.05)(.75,0.95)
	\psframe[linestyle=dotted,linecolor=blue](0.45,0.05)(.75,0.95)
	\psframe[linestyle=dotted,linecolor=red](0.35,0.05)(.75,0.95)
	\psframe[linestyle=dotted,linecolor=orange](0.25,0.05)(.75,0.95)
	\tickx{0.55}{$x^*$}
	\landcakethreewalls
	\end{pspicture}
\end{center}
The winning bid (marked by thicker dots above) is a rectangle $[x^*,L]\times[0,1]$. There are two cases:
\begin{itemize}
\item \emph{Easy case}: $x^*\geq1/2$. Make a vertical guillotine cut at $x^*$. Give to the winner the square $[x^*,x^*+1]\times[0,1]$. This square contains the winning bid, so its value for the winner is at least 1. The remaining cake is a 2-fat rectangle and its value for the remaining $n-1$ agents is at least $V(C)-1\geq4n-6\geq\max(2,4(n-1)-4)$. Use the \emph{4 walls procedure} to divide the remainder among the losers.
	\item \emph{Hard case}: $x^{*}<1/2$. Now we cannot let the winner have the winning bid, since the remainder will be too thin for the remaining agents. Our solution relies on the following observation: the fact that $x^*<1/2$ means that all agents value the rectangle $[1/2,L]\times[0,1]$ as less than 1. Therefore, they value the rectangle $[0,1/2]\times[0,1]$ as at least $4n-6$. Since all agents believe that this ``far left'' rectangle is so valuable, we are going to do an \emph{eval auction} inside it. 
\end{itemize}

(2) \textbf{Eval auction}. Let $C'=[0,1/2]\times[1/2,1]$ and $C''=[0,1/2]\times[0,1/2]$:
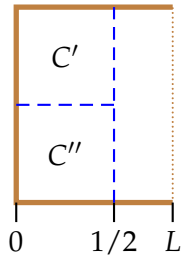
\begin{center}
	\begin{pspicture}(-1,-.3)(2,1)
	\landcakethreewalls
	\tickx{0.5}{$1/2$}
	\psline[linestyle=dashed,linecolor=blue](0.5,0)(0.5,1)
	\psline[linestyle=dashed,linecolor=blue](0,0.5)(0.5,0.5)
	\rput(0.25,0.75){$C'$}
	\rput(0.25,0.25){$C''$}
	\end{pspicture}
\end{center}
Do an eval auction on $C'$. Order the agents in a descending  order of their bid, $V_1(C')\geq \cdots \geq V_n(C')$, and let $n'$ be the largest integer with:
\begin{align*}
	V_{n'}(C')\geq \max(4 n'-5,1)
\end{align*}
If $n'=n$ then all agents value $C'$ as the entire cake, so the other parts of the cake can be discarded and the division procedure can start again with $C'$ as the cake. Hence, we assume that $n'<n$. There are several cases to consider:
\begin{itemize}
\item \emph{Easy case}: $1\leq n'\leq n-2$. 
	Make a horizontal guillotine cut between $C'$ and $C''$.
	Use the 3-walls procedure to divide $C'$ among the $n'$ winners.
	
	The losers value $C'$ as less than $\max(4(n'+1)-5,1)=4 n'-1$. At this point all agents value the rectangle $C'\cup C''$ as at least $4n-6$; hence, all losers value $C''$ as at least $(4n-6)-(4 n'-1) = 4(n-n')-5$. Since $n-n'\geq 2$, this value is also larger than 1, so we can use the 3-walls procedure to divide $C''$ among the $n-n'$ losers.
	
	Note that no square is allocated to the right of the line $x=1/2$, so we can assume that the rightmost border of both $C'$ and $C''$ is open and use the 3-walls procedure to divide them.
	
\item \emph{Hard case}: $n'=0$. This means that all agents value $C'$ as less than 1, so they value $C''$ as at least $4n-7$. 
	Now we have a problem: we cannot give $C'$ even to a single agent since it is not sufficiently valuable, but we also cannot divide $C''$ among all $n$ agents since it too is not sufficiently valuable. 
	
	Our solution is to shrink $C''$ towards the corner, until one of the agents decides that it is better to take a piece outside $C''$ and leave $C''$ to the remaining $n-1$ agents. This solution is implemented using a \emph{mark auction}, which is described in detail in step (3) below. But before proceeding there is one more case that must be handled:
	
\item \emph{Mixed case}: $n'=n-1$. This is handled according to the bid of the single losing agent (agent $n$): if $V_n(C')<4n-7$, then the losing agent values $C''$ as at least 1, so we can proceed as in the Easy case (the winning agents receive $C'$ and the losing agent receives $C''$). Otherwise, $V_n(C')\geq 4n-7$, so all agents value $C'$ as at least $4n-7$ (because the agents are ordered in descending  order of their bid). Switch the roles of $C'$ and $C''$ (e.g. by reflecting the cake about the line $y=1/2$), and proceed as in the hard case to the next auction.
\end{itemize}

(3) \textbf{Mark auction.} Ask each agent to mark an L-shape with a value of exactly 2, the complement of which is a square inside $C''$ with a value of $4n-7$ cornered at the corner of $C$, like this:
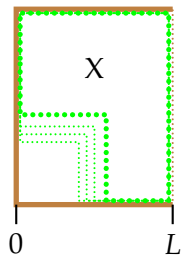
\providecommand{\lshape}{}
\renewcommand{\lshape}[1]{
	\pspolygon[linecolor=green](0.02,0.98)(0.02,#1)(#1,#1)(#1,0.02)(0.78,0.02)(0.78,0.98)(0.02,0.98)
}
\begin{center}
	\begin{pspicture}(-1,-.3)(2,1)
	\landcakethreewalls
	\psset{linestyle=dotted}
	\lshape{0.32}
	\lshape{0.36}
	\lshape{0.40}
	\psset{linewidth=2pt} \lshape{0.46}
	\rput(0.4,0.7){X}
	\end{pspicture}
\end{center}
Let $X$ be the winning bid. $X$ can be covered by two  overlapping pieces: a square near the top-left corner of $C$ (denoted by $Y$ below) and a square near the right edge of $C$ (denoted by $Z$ below):
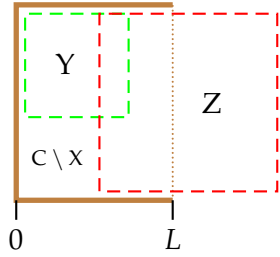
\begin{center}
	\begin{pspicture}(-1,-.3)(2,1)
	\landcakethreewalls

 \Square[linestyle=dashed,linecolor=green](0.04,0.42){0.54}
 \rput(0.25,0.7){Y}
 \psframe[linestyle=dashed,linecolor=red](0.42,0.04)(1.34,0.96)
 \rput(1,0.5){Z}
 \rput(.21,.21){\scriptsize C \textbackslash{} X} 
	\end{pspicture}
\end{center}
At least one of these squares must have a value of at least 1 to the winner. If Y has value 1 then give Y to the winner and leave Z unallocated; otherwise, give Z to the winner, leave Y unallocated and rotate $C$ clockwise $90^{\circ}$. In both cases, $C\setminus X$ can be separated from the piece given to the winner using a horizontal guillotine cut. Moreover, in both cases the cake to the right of $C\setminus X$ is unallocated. The remaining $n-1$ agents
value $C\setminus X$ as at least $(4n-5)-2$, which is more than $\max(1,4(n-1)-5)$. Use the 3 walls procedure to divide $C\setminus X$ among them.  \qed

\subsubsection{4 Walls procedure}
\psset{unit=2cm}

At this point, there are $n\geq 2$ agents who value the cake as at least $4n-4$. 

The 4-walls procedure is similar to the 3-walls procedure except that it has one additional \emph{eval} auction at the beginning. If this auction succeeds, then it effectively cuts the cake to two halves each of which is a 2-fat rectangle, so each half can be divided recursively using the 4-walls procedure. If this auction fails (as will be explained below), then the situation is similar to the 3-walls procedure and we can use a similar sequence of three auctions.

(0) \textbf{Eval auction.} Let $C'=[{L/2},1]\times[0,1]=$ the rightmost half of $C$. Note that both $C'$ and its complement are 2-fat rectangles:
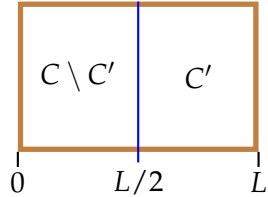
\begin{center}
\begin{pspicture}(-1,-.5)(2,1)
  \landcakefourwalls
  \cutatx{.8}{$L/2$}
  \rput(0.4,0.5){$C\setminus C'$}
  \rput(1.2,0.5){$C'$}
\end{pspicture}
\end{center}
Do an eval auction on $C'$. Order the agents in a descending  order of their bid, $V_1(C')\geq \cdots \geq V_n(C')$, and let $n'$ be the largest integer with:
\begin{align*}
V_{n'}(C')\geq \max(4 n'-4,2)
\end{align*}
If $n'=n$ then for all agents $V_i(C')=V_i(C)$, so $C\setminus C'$ can be ignored and the procedure can be restarted with $C'$ as the entire cake. Hence, we assume $n'<n$. There are several cases to consider:
\begin{itemize}
\item \emph{Easy case}: $1\leq n'\leq  n-2$. Make a vertical guillotine cut between $C'$ and $C\setminus C'$. Use the 4-walls procedure to divide $C'$ among the $n'$ winners. This is possible since $C'$ is a 2-fat rectangle and all winners value it as at least $\max(4 n'-4,2)$. 

The losers value $C'$ as less than $\max(4(n'+1)-4,2)=4 n'$, so they value the remaining half $C\setminus C'$ as more than $(4n-4)-4 n' = 4(n-n')-4$. Since $n-n'\geq 2$, this value is also larger than 2. Use the 4-walls procedure to divide $C\setminus C'$ among the $n-n'$ losers; this is possible  since $C\setminus C'$ is a 2-fat rectangle and all losers value it as at least $\max(4(n-n')-4,2)$. 

\item \emph{Hard case}: $n'=0$. This means that all agents value $C'$ as less than 2 so they value the remainder $C\setminus C'$ as at least $4n-6$. We are going to enlarge $C'$ leftwards, until it becomes sufficiently valuable such that some agent is willing to accept it. We implement this solution using a \emph{mark auction}, described in step (1) below. But beforehand, one more case must be handled:

\item \emph{Mixed case}: $n'=n-1$. This case is handled according to the bid of the losing agent: if $V_n(C')<4n-6$, then the losing agent values $C\setminus C'$ as at least 2, so we can proceed as in the Easy case (the winning agents receive $C'$ and the losing agent receives $C\setminus C'$). Otherwise, $V_n(C')\geq 4n-6$, so all agents value $C'$ as at least $4n-6$. Switch the roles of $C'$ and $C\setminus C'$ (e.g. by reflecting the cake $C$ about the line $x=L/2$), and proceed as in the hard case to the next auction.
\end{itemize}

(1) \textbf{Mark auction.} Ask each agent to mark a rectangle with a value of exactly 2 adjacent to the rightmost edge of $C$:
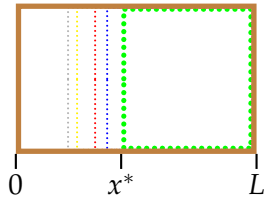
\begin{center}
\begin{pspicture}(-1,-.3)(2,1)
     \psframe[linestyle=dotted,linewidth=2pt,linecolor=green](0.7,0.02)(1.58,0.98)
     \psframe[linestyle=dotted,linecolor=blue](0.6,0.02)(1.58,0.98)
     \psframe[linestyle=dotted,linecolor=red](0.52,0.02)(1.58,0.98)
     \psframe[linestyle=dotted,linecolor=yellow](0.4,0.02)(1.58,0.98)
     \psframe[linestyle=dotted,linecolor=black!30](0.34,0.02)(1.58,0.98)
     \tickx{0.7}{$x^*$}
     \landcakefourwalls
\end{pspicture}
\end{center}
The smallest rectangle wins. Let $x^*$ be the x coordinate of its leftmost edge, so the winning bid is $[x^*,L]\times[0,1]$. Since all agents value $C'$ as less than 2, all bids must contain $C'$, so $x^*\leq L/2$. There are two cases:
\begin{itemize}
\item \emph{Easy case}: $x^*\geq {1/2}$. Make a vertical guillotine cut at  $x^*$. Both the winning bid and its complement are 2-fat rectangles. By the Covering Lemma, the winner can be allocated from its bid a square with a value of at least 1. The $n-1$ losers value the remaining cake, $[0,x^*]\times[0,1]$, as at least $4n-6$, which is at least $\max(2,4(n-1)-4)$. Hence, the 4-walls procedure can be used to divide the remainder among the losers.
\item \emph{Hard case}: $x^*<1/2$. Now we cannot let the winner have the winning bid, since the remainder will be too thin for the remaining agents. But we know that all agents value the rectangle $[1/2,L]\times[0,1]$ as less than 2 so they value the rectangle $[0,1/2]\times[0,1]$ as at least $4n-6$. Since all agents believe that this rectangle is so valuable, we are going to do an \emph{eval auction} inside it. 
\end{itemize}
(2) \textbf{Eval auction}. Let $C'=[0,1/2]\times[1/2,1]$ and $C''=[0,1/2]\times[0,1/2]$:
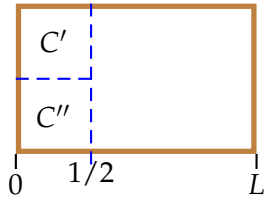
\begin{center}
\begin{pspicture}(-1,-.5)(2,1)
     \landcakefourwalls
     \psline[linestyle=dashed](0.5,-0.07)(0.5,1)
     \rput(0.5,-0.13){1/2}
     \psline[linestyle=dashed](0,0.5)(0.5,0.5)
     \rput(0.25,0.75){$C'$}
     \rput(0.25,0.25){$C''$}
\end{pspicture}
\end{center}
Do an eval auction on $C'$ and let $n'$ be the largest integer with:
\begin{align*}
V_{n'}(C')\geq \max(4 n'-5,1)
\end{align*}
As in step (0), the case $n'=n$ is trivial and can be ignored. The non-trivial cases are:
\begin{itemize}
\item \emph{Easy case}: $1\leq n'\leq n-2$. 
Make a horizontal guillotine cut between $C'$ and $C''$.
Use the \emph{3-walls procedure} to divide $C'$ among the $n'$ winners. The 3-walls procedure might allocate pieces that flow over the right boundary of $C'$ (the line $x=1/2$). This does not cause any problem because the side-length of these rectangles is at most $1/2$, so they are still contained in the original cake $C$.

The losers value $C'$ as less than $\max(4(n'+1)-5,1)=4 n'-1$. At this point of the procedure, all agents value the rectangle $C'\cup C''$ as at least $4n-6$; hence, all losers value $C''$ as at least $(4n-6)-(4 n'-1) = 4(n-n')-5$. Since $n-n'\geq 2$, this value is also larger than 1, so we can use the 3-walls procedure to divide $C''$ among the $n-n'$ losers.

\item \emph{Hard case}: $n'=0$. This means that all agents value $C'$ as less than 1 and value $C''$ as at least $4n-7$. We are going to ``shrink'' $C''$ using a mark-auction in step (3). But beforehand we handle the remaining case:

\item \emph{Mixed case}: $n'=n-1$. Proceed according to the bid of the losing agent: if $V_n(C')<4n-7$, then the losing agent values $C''$ as at least 1, so we can proceed as in the Easy case (the winning agents receive $C'$ and the losing agent receives $C''$). Otherwise, $V_n(C')\geq 4n-7$, so all agents value $C'$ as at least $4n-7$. Switch the roles of $C'$ and $C''$, and proceed as in the hard case to the next auction.
\end{itemize}

(3) \textbf{Mark auction.} Ask each agent to mark an L-shape with a value of 3, whose complement is a square inside $C''$ cornered at the corner of $C$, like this:

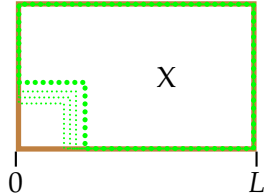
\providecommand{\lshape}{}
\renewcommand{\lshape}[1]{
\pspolygon[linecolor=green](0.02,0.98)(0.02,#1)(#1,#1)(#1,0.02)(1.58,0.02)(1.58,0.98)(0.02,0.98)
}
\begin{center}
\begin{pspicture}(-1,-.3)(2,1)
 \landcakefourwalls
 \psset{linestyle=dotted}
 \lshape{0.32}
 \lshape{0.36}
 \lshape{0.40}
 \psset{linewidth=2pt} \lshape{0.46}
 \rput(1,0.5){X}
\end{pspicture}
\end{center}
Since all agents value $C''$ as at least $4n-7 = (4n-4)-3$ they can indeed bid as required. Let $X$ be the winning bid. $X$ is an L-shape that can be covered by two overlapping pieces: a square near the top-left corner of $C$ (denoted by $Y$ below) and a rectangle near the right edge of $C$ (denoted by $Z$ below):
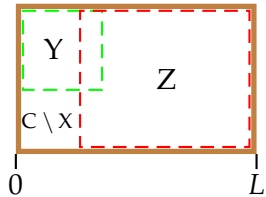
\begin{center}
\begin{pspicture}(-1,-.3)(2,1)
 \landcakefourwalls
 \Square[linestyle=dashed,linecolor=green](0.04,0.42){0.54}
 \rput(0.25,0.7){Y}
 \psframe[linestyle=dashed,linecolor=red](0.42,0.04)(1.56,0.96)
 \rput(1,0.5){Z}

 \rput(.21,.21){\scriptsize C \textbackslash{} X} 
\end{pspicture}
\end{center}
Since the winner values $X$ as 3, at least one of the following must hold:
\begin{itemize}
\item The winner values $Y$ as at least 1; if this is the case then the winner receives $Y$, and $Z$ remains unallocated.
\item The winner values $Z$ as at least 2; if this is the
case then the winner selects a square from $Z$ with a value of at least 1 (this is possible by the Covering Lemma since $Z$ is a 2-fat rectangle), and $Y$ remains unallocated. If this is the case, then rotate $C$ clockwise $90^{\circ}$. 
\end{itemize}
In both cases, $C\setminus X$ can be separated from the piece given to the winner using a horizontal guillotine cut.
In both cases, the cake to the right of $C\setminus X$ is unallocated. The $n-1$ losers value $X$ as at most 3 so they value $C\setminus X$ as at least $(4n-4)-3$,
which is at least $\max(1,4(n-1)-5)$. Therefore, the 3 walls procedure can be used to divide $C\setminus X$ among them.
\qed

The above pair of procedures prove the following pair of positive results $\forall n\geq2$: 
\begin{align*}
\prop(2\,fat\,rectangle\,with\,all\,sides\,bounded,\,n,\,Squares)&\geq\frac{1}{4n-4}
\\
\prop(Rectangle\,with\,a\,long\,side\,unbounded,\,n,\,Squares)&\geq\frac{1}{4n-5}
\end{align*}
Since a square is a 2-fat rectangle:
\begin{align*}
\prop(Square\,with\,4\,walls,\,n,\,Squares)&\geq\frac{1}{4n-4}
\\
\prop(Square\,with\,3\,walls,\,n,\,Squares)&\geq\frac{1}{4n-5}
\end{align*}

\subsubsection{Fat rectangle pieces}
When the pieces are allowed to be $R$-fat rectangles, the above lower bounds are of course still true, since a square is an $R$-fat rectangle. But when $R\geq 2$, the 4-walls division procedure can give slightly stronger guarantees --- the required value is $\max(1,4n-5)$ instead of $\max(2,4n-4)$ (this is analogous to the fact that in Subsection \ref{sub:negative-fat}, when the pieces are allowed to be $R$-fat rectangles with $R\geq 2$, our upper bound for a cake with 4 walls is slightly weaker --- the denominator is $2n-1$ instead of $2n$). The required modifications are briefly explained below:
\begin{itemize}
\item In the base case ($n=1$), since the cake is 2-fat, the single agent can have it all, so it is sufficient that its value be 1.
\item In step (0), after the Eval auction, $n'$ is  the largest integer with $V_{n'}(C')\geq \max\bm{(4 n'-5,1)}$. In the easy case, the $n'$ winners value their share $C'$ as at least $\max(4 n'-5,1)$ and the $n-n'$ losers value their share $C\setminus C'$ as at least $\max(4(n-n')-5,1)$, so each part can be divided recursively using the 4-walls procedure. In the hard case, all agents value $C'$ as less than 1 so they value the remainder $C\setminus C'$ as at least $4n-6$; proceed to the next step.
\item In step (1), the Mark auction asks each agent to mark a rectangle with a value of exactly \textbf{1} adjacent to the rightmost edge of $C$. In the easy case, both the winning bid and its complement are 2-fat rectangles. The winning bid can be given entirely to the winner; the $n-1$ losers value the remaining cake as at least $4n-6$, which is at least $\max(1,4(n-1)-5)$, so the 4-walls procedure can be used to divide the remainder among them.
In the hard case, all agents value the rectangle $[1/2,L]\times[0,1]$ as less than 1 so they value the rectangle $[0,1/2]\times[0,1]$ as at least $4n-6$; proceed to the next step.
\item In step (2), the Eval auction proceeds exactly as in the case of square pieces. The values are sufficient for using the 3-walls procedure.
\item In step (3), the Mark auction asks each agent to mark an L-shape with a value of exactly \textbf{2}. Let $X$ be the winning bid. Since the winner values $X$ as 2, he values either its topmost part or its rightmost part as at least 1; both these parts are 2-fat rectangles so the winner can pick one of them and get a fair share. In both cases, $C\setminus X$ (which is a square) can be separated from the piece given to the winner using a horizontal guillotine cut.
In both cases, the $n-1$ losers value $X$ as at most 2 so they value $C\setminus X$ as at least $(4n-5)-2$,
which is at least $\max(1,4(n-1)-5)$. Therefore, the 3 walls procedure can be used to divide $C\setminus X$ among them.
\item The 3-walls procedure remains unchanged. 
\end{itemize}
So for every $n\geq 2$ and $R\geq 2$:
\begin{align*}
\prop(2\,fat\,rectangle\,with\,all\,sides\,bounded,\,n,\,R\,fat\,rectangles)&\geq\frac{1}{4n-5}
\end{align*}

\subsection{Two walls} \label{sub:alg-2-walls}
We present a division procedure for dividing the top-right quarter-plane, i.e, the cake is a square with two walls and two unbounded sides. We would like to do a \emph{mark auction} in which each agent is asked to mark a square adjacent to the bottom-left corner. Then, the smallest square should be allocated to its bidder and the remaining cake should be divided among the remaining agents. However, when we try to do this we run into trouble, as the remaining cake is no longer a quarter-plane.

As it often happens, the solution is to generalize the problem. Instead of dividing a quarter-plane, we divide a \emph{rectilinear polygonal domain unbounded in two directions}, which for brevity we call ``staircase'' because of its shape (see Figure \ref{fig:staircase}). 

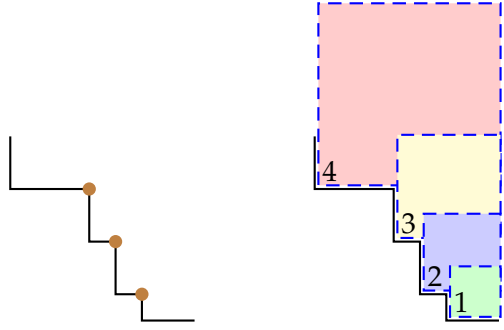
\begin{figure}
\psset{unit=0.35mm}
\def\landcakestaircase{\psline[linecolor=black](70,0)(50,0)(50,10)(40,10)(40,30)(30,30)(30,50)(0,50)(0,70)}
\centering
\begin{pspicture}(70,70)
\landcakestaircase
\psdot[dotsize=5pt,linecolor=brown](50,10)
\psdot[dotsize=5pt,linecolor=brown](40,30)
\psdot[dotsize=5pt,linecolor=brown](30,50)
\end{pspicture}
\hskip 15mm
\begin{pspicture}(100,120)
\landcakestaircase
\Square[linestyle=dashed,fillcolor=red!20,fillstyle=solid](1,51){70}
\Square[linestyle=dashed,fillcolor=yellow!20,fillstyle=solid](31,31){40}
\Square[linestyle=dashed,fillcolor=blue!20,fillstyle=solid](41,11){30}
\Square[linestyle=dashed,fillcolor=green!20,fillstyle=solid](51,1){20}
\rput[bl](53,3){1}
\rput[bl](43,13){2}
\rput[bl](33,33){3}
\rput[bl](3,53){4}
\end{pspicture}

\protect\caption{\label{fig:staircase} A staircase with $T=3$ teeth marked by discs (Left). It has $T+1=4$ corners and can be covered by 4 squares (Right).}
\end{figure}
Each staircase has vertexes with inner angle $90^\circ$ and vertexes with inner angle $270^\circ$; we call the former \emph{corners} and the latter \emph{teeth}.\footnote{Other common names are \emph{convex vertexes} vs. \emph{concave/reflex vertexes}, or \emph{inner corners} vs. \emph{outer corners}.} A staircase with $T$ teeth has $T+1$ corners. A quarter-plane is a staircase with $T=0$ teeth.

By putting the arrangement of Claim \ref{claim:neg-quart-n} in one of the corners and adding a pool in each of the other $T$ corners, the following upper bound is obtained:
\begin{align*}
\prop(T\,staircase,n,Squares)\leq \frac{1}{2n-1+T}
\end{align*}
We normalize the valuations of all agents such that the value of the entire cake is $2n-1+T$. We use a sequence of mark auctions to give each agent a square with a value of at least 1.

We proceed by induction on the number of agents $n$. When $n=1$, the cake value for the single agent is at least $T+1$. The cake can be covered by $T+1$ sufficiently large squares --- one square per corner (see Figure \ref{fig:staircase}/Right). By the Covering Lemma, the agent can get a square with a value of at least 1.

Suppose we already know how to divide a $T$-staircase to $n-1$ agents, for every integer $T\geq 0$. Now there are $n$ agents. Start by doing \textbf{$T+1$ mark auctions:} for each corner $j\range{1}{T+1}$, ask each agent to mark a square with a value of exactly 1 adjacent to corner $j$. If the total value of the agent in corner $j$ is less than 1, then the agent is allowed to not participate in that auction, or equivalently mark a square with an infinite side-length. By the Covering Lemma, each agent can mark at least one finite square.

In each corner, the ``corner-winning-bid'' is the smallest square (contained in all other bids in that corner). We now have $T+1$ corner-winners, and we have to select a single global-winner. There are two cases.

\emph{Easy case:} there is a $j\range{1}{T+1}$ such that the corner-$j$ winning-bid is smaller than the two edges of $C$ adjacent to corner $j$. An example is the square in corner 4 in Figure \ref{fig:staircase-alg-easy}. Select one such square arbitrarily as the global ``winning bid''. Give the winning bid to its bidder. The remaining cake is a staircase with $T+1$ teeth (see Figure \ref{fig:staircase-alg-easy}). The $n-1$ losing agents value the allocated square as at most 1, so they value the remaining staircase as at least $(2n-1+T)-1 = 2(n-1)-1+(T+1)$. Hence, by induction we can divide the remainder among the losers.

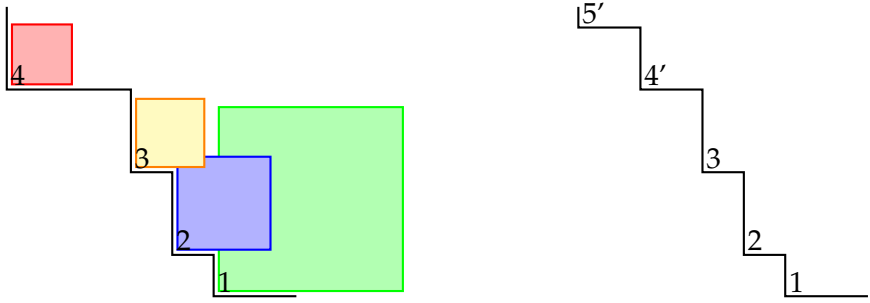
\begin{figure}
\psset{unit=0.55mm,hatchangle=0,hatchsep=2pt}
\def\landcakestaircase{
\psline[linecolor=black](70,0)(50,0)(50,10)(40,10)(40,30)(30,30)(30,50)(0,50)(0,70)
}
\centering
\begin{pspicture}(100,100)
\landcakestaircase
\Square[linecolor=green,fillcolor=green!30,fillstyle=solid](51,1){45}
\Square[linecolor=blue,fillcolor=blue!30,fillstyle=solid](41,11){23}
\Square[linecolor=orange,fillcolor=yellow!30,fillstyle=solid](31,31){17}
\Square[linecolor=red,fillcolor=red!30,fillstyle=solid](1,51){15}
\rput[bl](51,1){1}
\rput[bl](41,11){2}
\rput[bl](31,31){3}
\rput[bl](1,51){4}
\end{pspicture}
\hskip 2cm
\begin{pspicture}(80,100)
\psline[linecolor=black](70,0)(50,0)(50,10)(40,10)(40,30)(30,30)(30,50)(15,50)(15,65)(0,65)(0,70)
\rput[bl](51,1){1}
\rput[bl](41,11){2}
\rput[bl](31,31){3}
\rput[bl](16,51){4'}
\rput[bl](1,66){5'}
\end{pspicture}
\protect\caption{\label{fig:staircase-alg-easy}
The square at corner 4 is entirely contained in the corner (left). After it is allocated, the remaining cake is a staircase with 4 teeth and 5 corners (right).
}
\end{figure}

\emph{Hard case:} all corner-winning-bids are larger than the edges adjacent to their corners, as in Figure \ref{fig:staircase-alg}. Now, when a square is allocated, the remainder is no longer a staircase. In order to restore the staircase shape, we have to remove an additional part of $C$. We do this by cutting, from the top-right corner of the allocated square, a straight line downwards to the bottom boundary of $C$, and a straight line leftwards to the leftmost boundary of $C$. 
The parts of $C$ that are removed besides the allocated square are called the \emph{shadows} of the square. An example is illustrated in Figure \ref{fig:staircase-alg}, where the square at corner 2 has two shadows denoted by dotted lines.
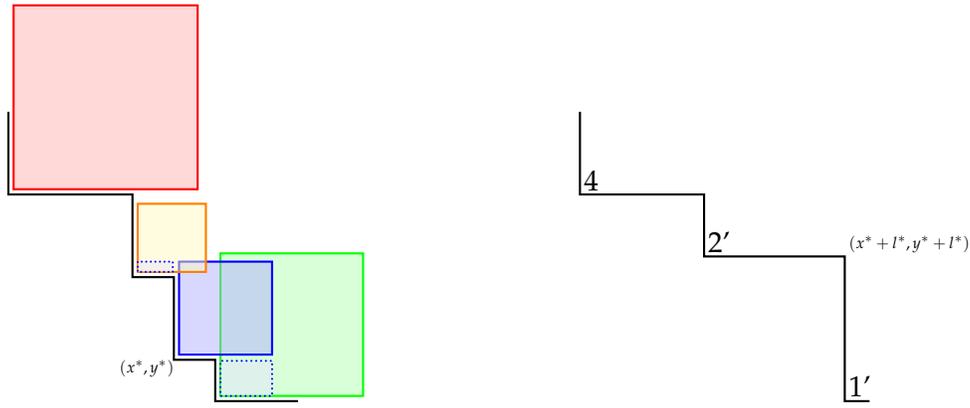
\begin{figure}
\psset{unit=0.55mm,hatchangle=0,hatchsep=2pt}
\def\landcakestaircase{
\psline[linecolor=black](70,0)(50,0)(50,10)(40,10)(40,30)(30,30)(30,50)(0,50)(0,70)
\Square[linecolor=green,fillcolor=green!30,fillstyle=solid,opacity=0.5](51,1){35}
\Square[linecolor=blue,fillcolor=blue!30,fillstyle=solid,opacity=0.5](41,11){23}
\Square[linecolor=orange,fillcolor=yellow!30,fillstyle=solid,opacity=0.5](31,31){17}
\Square[linecolor=red,fillcolor=red!30,fillstyle=solid,opacity=0.5](1,51){45}
}
\centering
%
%
\begin{pspicture}(100,70)
\landcakestaircase
\rput[tr](40,10){\tiny{$(x^*,y^*)$}}
\psframe[linecolor=blue,fillcolor=blue!10,fillstyle=solid,hatchcolor=blue,opacity=0.5,linestyle=dotted](31,31)(40,34)
\psframe[linecolor=blue,fillcolor=blue!10,fillstyle=solid,hatchcolor=blue,opacity=0.5,linestyle=dotted](51,1)(64,10)
\end{pspicture}
\hskip 2cm
\begin{pspicture}(100,70)
\psline[linecolor=black](70,0)(64,0)(64,35)(30,35)(30,50)(0,50)(0,70)
\rput[bl](65,1){1'}
\rput[bl](31,36){2'}
\rput[bl](65,36){\tiny{$(x^*+l^*,y^*+l^*)$}}
\rput[bl](1,51){4}
\end{pspicture}
\protect\caption{\label{fig:staircase-alg}
The square at corner 2 (second from the bottom-right) satisfies the Staircase Lemma, since its ``shadows'' (dotted) are contained in the other squares. After it is allocated, the remaining cake is a staircase with 2 teeth and 3 corners (right).}
\end{figure}

We now need the following geometric lemma, which is formally stated and proved in Appendix \ref{sec:staircase-lemma}:
\begin{lemma}
(Staircase Lemma)
Given a staircase in which a square is located in each corner, there exists a square whose shadows are contained in the union of the other squares.
\end{lemma}
Based on the Staircase Lemma, we proceed as follows. From the $T+1$ corner-winning-bids, select one square whose shadows are contained in the other squares (e.g. the square in corner 2 in Figure \ref{fig:staircase-alg}). Declare this square as the global winning square, give it to its bidder, and remove its shadows from $C$.

We have to prove that the remaining cake is sufficiently valuable for each losing agent. The number of agents changes by $\Delta n = -1$ since the winning agent leaves. The cake value for a losing agent changes by $\Delta V$ (a negative quantity). The number of teeth changes by $\Delta T$ which may be positive or negative. Looking at the value requirement $V \geq 2n+T-1$, we see that in order to use the induction assumption, it is sufficient to prove that for every loser:
\begin{align*}
\Delta V \geq 2 \Delta n + \Delta T = \Delta T - 2
\end{align*}
I.e, the value of the remaining agents should drop by at most two units, plus one unit for each removed tooth.

The shadows of the winning square can be partitioned to $m$ disjoint rectangular components, to its top-left and to its bottom-right, such that each component is located in a different corner (e.g. in Figure  \ref{fig:staircase-alg}, $m=2$). After the shadows are removed, $m$ teeth disappear. One tooth is added at the top-right of the winning square. Hence, $\Delta T = 1-m$.

The winning square is worth at most $1$ for the remaining agents, since it is contained in all other squares in its corner. By the selection of the global-winning-bid, each of the $m$ shadows is contained in a corner-winning-bid, so its value for the losing agents is at most 1. Hence, the total value of the removed region to the $n-1$ losers is at most $m+1$, so $\Delta V\geq -1-m = \Delta T-2$, as required. Hence, by the induction assumption we can proceed and divide the remainder among the losers. \footnote{The easy case is, in fact, contained in the hard case, since a square smaller than the edges adjacent to its corner has an empty shadow (so $m=0$). The split to easy and hard cases is done for presentation purposes only.} \qed

The above procedure proves that, for every $n\geq 1,T\geq 0$:

\begin{align*}
\prop(T\,staircase,n,Squares)=\frac{1}{2n-1+T}
\end{align*}

By letting $T=0$ we get: 
\begin{align*}
\prop(Quarter\,plane,\,n,\,Squares)=\frac{1}{2n-1}
\end{align*}

\subsection{One and zero walls}
A half-plane can be divided by partitioning it to two quarter-planes:
\begin{claim}\label{pos:half-plane}
For every $n\geq2$:
\end{claim}
\begin{align*}
\prop(Half\,plane,\,n,\,Squares)\geq\frac{1}{2n-2}
\end{align*}
\begin{proof}
Assume the cake is the half-plane $y\geq0$ and there are $n$ agents who value it as $2n-2$. Do the following \emph{mark auction}: ask each agent to mark a quarter-plane open to the top-left, whose bottom edge is adjacent to the bottom edge of $C$ and its value is exactly 1. An example is illustrated below, where the winning bid is --- as usual --- marked by thicker dots:
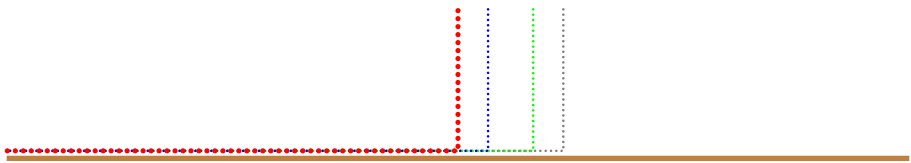
\begin{center}
\psset{unit=2cm,linewidth=1pt}
\newcommand{\quarterplane}[1]{\psline[linestyle=dotted](#1,1)(#1,0.05)(-3,0.05)}
\begin{pspicture}(-3,0)(3,1)
\psline[linecolor=brown,linestyle=solid,linewidth=2pt](-3,0)(3,0)
\psset{linecolor=gray}\quarterplane{0.7}
\psset{linecolor=green}\quarterplane{0.5}
\psset{linecolor=blue}\quarterplane{0.2}
\psset{linecolor=red,linewidth=2pt}\quarterplane{0}
\end{pspicture}
\end{center}
After the winning bid is allocated to its winner, the $n-1$ losers value the remaining quarter-plane as at least $(2n-2)-1=2(n-1)-1$; divide it among
them using the procedure of Subsection \ref{sub:alg-2-walls}.
\end{proof}

An unbounded plane can be divided by partitioning it to two half-planes.
\begin{claim}\label{pos:plane}
For every $n\geq 4$:
\end{claim}
\begin{align*}
\prop(Plane,\,n,\,Squares)\geq\frac{1}{2n-4}
\end{align*}
\begin{proof}
Normalize the cake value to $2n-4$. Do the following \emph{mark auction}: ask each agent to mark a half-plane bounded at its top, with a value of exactly 2 (so each agent $i$ marks a half-plane $Y_i = [-\infty,\infty]\times[-\infty,y_i]$). Order the bids by containment, so that $Y_1 \subseteq Y_2 \subseteq \dots \subseteq Y_n$. Select \emph{two} winners --- the agents with the two smallest bids ($Y_1$ and $Y_2$):
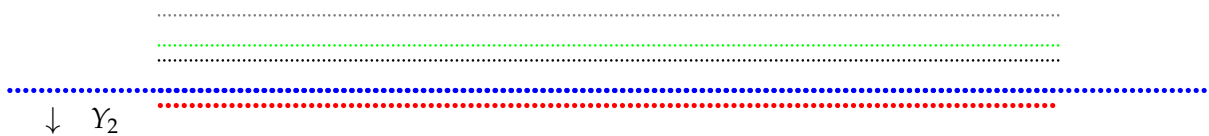
\begin{center}
	\psset{unit=2cm,linewidth=1pt}
	\newcommand{\halfplane}[1]{\psline[linestyle=dotted](-3,#1)(3,#1)}
	\begin{pspicture}(-4,-.2)(4,1)
	\psset{linecolor=gray}\halfplane{0.7}
	\psset{linecolor=green}\halfplane{0.5}
	\psset{linecolor=black}\halfplane{0.4}
	\psset{linecolor=blue,linewidth=2pt}\halfplane{0.2}
	\psline[linestyle=dotted](-4,.2)(4,.2)
	\psset{linecolor=red,linewidth=2pt}\halfplane{0.1}
	\rput(-3.5,0){$\downarrow ~~~Y_2$}
	\end{pspicture}
\end{center}
Both winners value $Y_2$ as at least 2; divide it among them using cut-and-choose. Each of them receives a quarter-plane with a value of at least 1. The remaining cake is a half-plane bounded at its bottom, which the $n-2$ losers value as at least $(2n-4)-2 = 2(n-2)-2$; divide it among them using the procedure of  Claim \ref{pos:half-plane}.
\footnote{
When $n=3$, the cake-value should be normalized to 3. The single losing agent values the remaining cake as 1 and takes it entirely.
When $n=2$, the cake-value should be normalized to 2. The agents simply divide the plane to two half-planes using cut-and-choose. In both these cases, $\prop(Plane,\,n,\,Squares)=1/n$.
}
\end{proof}

The lower bounds for one and zero walls do not match the upper bounds proved in Section \ref{sec:Impossibility-Results}: the proportionality coefficient (the coefficient of $n$ in the denominator) is 2 in both cases, while the coefficients in the upper bounds are 3/2 for a half-plane and almost 1 for an unbounded plane. We believe that the procedures presented above are tight and the ``real'' coefficient is 2. The reason is that, whenever a plane is cut by even a single straight line, the remainder is a half-plane, and when a half-plane is cut, the remainder is a quarter-plane, and for a quarter-plane we know that the proportionality coefficient is 2. In future work we plan to look for tighter impossibility results showing that the proportionality coefficient is indeed 2 in half-planes and unbounded planes, too.

\subsection{Three walls}
\label{sub:valleys}
\newcommand{\xmin}{x^{\min}}
\newcommand{\xmax}{x^{\max}}
\newcommand{\xleft}{x^{\text{left}}}
\newcommand{\xright}{x^{\text{right}}}
\newcommand{\yleft}{y^{\text{left}}}
\newcommand{\yright}{y^{\text{right}}}
Our next goal is to divide a square bounded by three walls. We already presented a procedure for a square with three walls in Subsection \ref{sub:Guillotine-algorithms}, but the value guarantee of the present procedure is better and it matches the upper bound of $1/(2n-1)$. On the other hand, the present procedure uses general  (non-guillotine) cuts.

Similarly to the two-walls case, we have to generalize the problem and divide a \emph{rectilinear polygonal domain unbounded in one direction}, which for brevity we call a ``valley''. Again the number of teeth is denoted by $T$; see Figure \ref{fig:valley}. We assume that the valley is entirely contained in the unit square $[0,1]\times[0,1]$.

\psset{unit=3.5mm,hatchsep=1pt}
\def\valleyfour{\psline[linecolor=black,linewidth=2pt]
	(0,10) (0,8)(1,8) (1,9)(5,9) (5,6)(7,6) (7,4)(10,4) (10,10)
}
\def\walls{\psline[linestyle=dotted](0,10)(0,0)(10,0)(10,10)(0,10)}
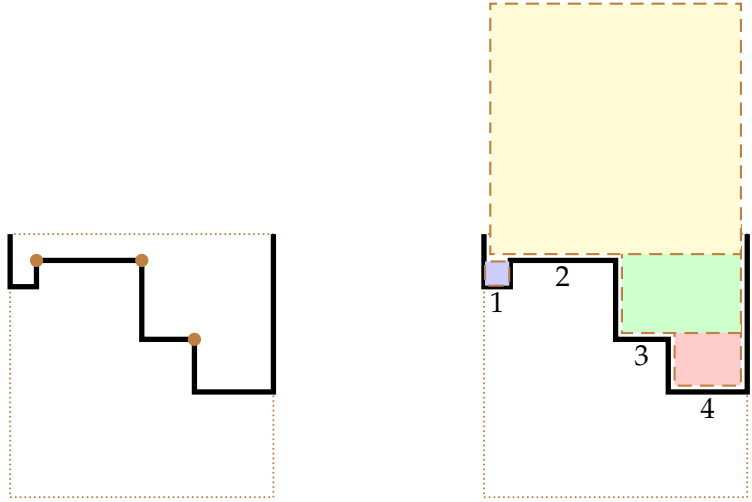
\begin{figure}
\begin{center}
\psset{linecolor=brown}
\begin{pspicture}(12,20)
\walls\valleyfour
\psdot[dotsize=5pt,linecolor=brown](1,9)
\psdot[dotsize=5pt,linecolor=brown](5,9)
\psdot[dotsize=5pt,linecolor=brown](7,6)
\end{pspicture}
\hskip 2cm
\begin{pspicture}(12,20)
\walls\valleyfour
\rput(0.5,7.4){1}
\rput(3,8.4){2}
\rput(6,5.4){3}
\rput(8.5,3.4){4}
\psframe[fillstyle=solid,fillcolor=red!20,linestyle=dashed](7.2,4.2)(9.8,6.8)
\psframe[fillstyle=solid,fillcolor=green!20,linestyle=dashed](5.2,6.2)(9.8,10.8)
\psframe[fillstyle=solid,fillcolor=blue!20,linestyle=dashed](0,8)(1,9)
\psframe[fillstyle=solid,fillcolor=yellow!20,linestyle=dashed](0.2,9.2)(9.8,18.8)
\end{pspicture}
\end{center}
\protect\caption{
\label{fig:valley}
A valley with $T=3$ teeth marked by discs (Left). It has $T+1=4$ levels and can be covered by 4 squares (Right). The levels coordinates are: $[0,.1]\times.8,\, 
[.1,.5]\times.9,\, [.5,.7]\times.6,\, [.7,1.0]\times.4$. The levels are covered from bottom to top: 4, then 3, then 1, then 2. In each level, the bottom rectangle, which is not overlapped by higher squares, is the \emph{covering rectangle} of that level.
}
\end{figure}

We require the valley to have the \emph{Sunlight property}, which means that light coming from the top can reach all parts of the bottom border. In other words: no part of the valley lies below a wall; the bottom border of a valley goes from the left wall (at $x=0$) to the right wall (at $x=1$) in stairs climbing to the top-right or bottom-right, but never back to the left. Hence a valley can be represented as a sequence of $T+1$ \emph{levels} $\{[\xmin,\xmax_i]\times y_i\}_{i=1}^{T+1}$, where (see Figure \ref{fig:valley}):
\begin{align*}
0 = \xmin_1 < \xmax_1 = \xmin_2 < \xmax_2 & \cdots < \xmax_{T} = \xmin_{T+1} < \xmax_{T+1} = 1
\\
&\textrm{and}
\\
\forall i\in\{1,\ldots, T+1\}&:~~~
0 \leq y_i \leq 1 
\end{align*}
Our valley-division procedure is essentially similar to the staircase-division procedure: a mark-auction is performed in each ``corner'' of the valley; the smallest bid in each corner is the corner-winning-bid; and a global winning-bid is selected such that its ``shadows'' are contained in all other bids. We have to carefully define the ``corners'' and the ``shadows'', and this requires several definitions.

\subsubsection{The structure of a valley}
For every level $i\range{1}{T+1}$, when we look from $(\xmin_i,y_i)$ leftwards, we see a wall. Let $\xleft_i$ be the $x$ coordinate of that
wall and $\yleft_i$ be the $y$ coordinate of the level at the
top of the wall $(\yleft_i>y_i)$. If $(\xmin_i,y_i)$ is a bottom-left corner (such as in levels 1 and 3 and 4 in Figure \ref{fig:valley}), then $\xleft_i=\xmin_i$ and $\yleft_i=y_{i-1}$ (if $\xleft_i=0$, i.e. we hit the left boundary, then we define $\yleft_i=1$). Otherwise (as in level 2), $\xleft_i<\xmin_i$.

Similarly, define $\xright_i$ as the $x$ coordinate of the wall we see at the right and $\yright_i$ as the $y$ coordinate of the level at the top of the wall $(\yright_i>y_i)$. If $(\xmax_i,y_i)$ is a bottom-right corner (such as in levels 1 and 4 in the figure), then $\xright_i=\xmax_i$ and $\yright_i=y_{i+1}$ (if $\xright_i=1$, i.e. we hit the right boundary, then we define $\yright_i=1$). Otherwise (as in levels 2 and 3), $\xright_i>\xmax_i$.

The horizontal distance between the two walls surrounding a level is denoted: 
\begin{align*}
dx_i := \xright_i-\xleft_i
\end{align*}
In the figure, the values of $dx_{i}$ for the 4 levels are: $0.1, 1.0, 0.5, 0.3$.
The vertical depth of a level is denoted by:
\begin{align*}
dy_i := \min(\yright_i,\yleft_i)-y_{i}
\end{align*}
It is the height to which one has to climb in order to move to another level, or to exit the unit square. In the figure, the values of $dy_{i}$ for the 4 levels (from left to right) are: $0.1, 0.1, 0.3, 0.2$. 

Initially we handle the case of a single agent. This requires a bound on the square-cover-number of the valley, as a function of $T$. In general, the square-cover-number of a valley can be arbitrarily large, e.g, if the valley has a single level $[0,1/m]\times 0$, then $m$ squares are required to cover it, for every integer $m$. For our purposes, we can restrict our attention to valleys that do not have such deep levels. Formally, we require the valley to have the \emph{Shallowness property}, which means that for every level $i$: 
\begin{align*}
dy_i\leq dx_i
\end{align*}
This property guarantees that the valley can be covered by at most $T+1$ squares, as we show in the following subsection.

\subsubsection{Covering a valley with squares}\label{sub:cover-valley}
\begin{lemma}\label{lem:cover-valley}
If $C$ is a valley with $T$ teeth satisfying the Shallowness property, then:
\begin{align*}
\covernum(C,Squares)\leq T+1
\end{align*}
\end{lemma}
\begin{proof}
Consider the lowest level --- the level $i$ with the smallest $y_i$. Consider the square:
\begin{align*}
S_i :=[\xleft_i,\,\xleft_i+dx_{i}]\times[y_{i},\, y_{i}+dx_{i}]
\end{align*}
Because this is the lowest level, both its endpoints are inner corners, so $\xleft_i=\xmin_i$
and $\xleft_i+dx_{i}=\xright_i=\xmax_i$.

The Shallowness property guarantees that $dx_i\geq dy_i$. Hence, $y_i + dx_i \geq y_i+dy_i = \min(\yright_i,\yleft_i)$. Hence, $S_i$ contains the rectangle:
\begin{align*}
R_i := [\xleft_i,\, \xright_i]\times[y_{i},\,\min(\yleft_i,\, \yright_i)]
\end{align*}
Call $R_i$ the \emph{covering rectangle of level i} (see Figure \ref{fig:valley}/Right). If we remove from the valley the covering rectangle of $i$ (the lowest level), then at least one of the teeth adjacent to it (from the left or from the right) is flattened, and we remain with at most $T-1$ teeth. In some remaining levels $j$, the $\xmin_j$ and $\xmax_j$ values might change, but the $\xleft_j$ and $\xright_j$ do not change since the removed level was lower than all surrounding levels. Hence, $dx_j$ and $dy_j$ do not change, the Shallowness property is preserved, and we can continue this process iteratively until all the valley is covered. The number of squares in the covering is at most the number of levels, $T+1$.
\end{proof}

\subsubsection{The division procedure}
We are now ready to present the valley-division procedure.

We normalize the valuations of all agents such that the value of the entire valley for each agent is $2n-1+T$. We use a sequence of mark auctions to give each agent a square with a value of at least 1. 

We proceed by induction on the number of agents $n$. When $n=1$, the value for the single agent is at least $T+1$. By Lemma \ref{lem:cover-valley} the valley can be covered by $T+1$ squares, so by the Covering Lemma the agent can get a square with a value of at least 1.

Suppose we already know how to divide a $T$-valley to $n-1$ agents, for every integer $T\geq 0$. Now there are $n$ agents. Start by doing \textbf{$2 (T+1)$ mark auctions}. There are two auctions per level: one on the left and one on the right of its covering rectangle. For every level $i\range{1}{T+1}$, ask each agent to mark two squares with a value of exactly 1: a square with its bottom-left corner at the bottom-left corner of $R_i$ $(\xleft_i,\, y_i)$ and a square with its bottom-right corner at the bottom-right corner of $R_i$ $(\xright_i,\, y_i)$. The squares may overlap. An agent can refrain from participating in an auction if the largest square he can mark at this corner has a value of less than 1. By the Covering Lemma, each agent can participate in at least one auction.

In each corner, there are at most $n$ squares. From these, we select a smallest square as the ``corner-winning-bid''. Now we have at most $2 (T+1)$ corner-winners. The global-winner is \textbf{the square with a lowest top side}. I.e, if the side-length of the $i$-level winning-bid is $l_i$, then the global winner is a square with a smallest $y_i+l_i$.

In the illustration below, the index of each level is written below the level. There are squares only in 7 out of 10 corners, since no agents participated in the auction for the corner $(\xleft_3,y_3)$ (marked with x) and for level 5. The global-winner (marked with thicker dots) is the corner-winner at the corner $(\xright_1,y_1)$:

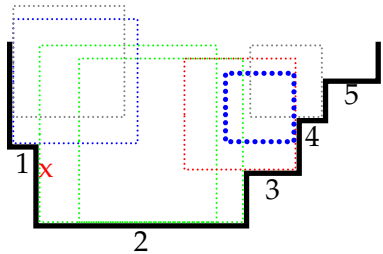
\newcommand{\valleydeep}{
\psline[linecolor=black,linewidth=2pt]
	(0,12) (0,8)(1,8) (1,5)(9,5) (9,7)(11,7) (11,9)(12,9) 
	(12,10.5)(14,10.5) (14,12)
	\rput(0.5,7.4){1}
	\rput(5,4.4){2}
	\rput(10,6.4){3}
	\rput(11.5,8.4){4}
	\rput(13,10){5}
}
\newcommand{\losingsquares}{
\Square[linecolor=blue](0.1,8.1){4.8}

\Square[linecolor=green](1.1,5.1){6.8}
\Square[linecolor=green](2.6,5.1){6.3}

\Square[linecolor=red](6.6,7.1){4.3}
\rput[l](1.1,7.1){\color{red}x}

\Square[linecolor=gray](0.1,9.1){4.3}
\Square[linecolor=gray](9.1,9.1){2.8}
}
\begin{center}
\begin{pspicture}(0,4)(14,14)
\valleydeep
\psset{dotsep=1pt,linestyle=dotted,linecolor=blue}
\losingsquares{}
\Square[linecolor=blue,linewidth=2pt](8.1,8.1){2.8}
\end{pspicture}
\end{center}
In addition to the winning square, we may have to remove some other parts of the valley, in order to ensure that the remaining valley satisfies the two properties defined above: the \emph{Sunlight property} and the \emph{Shallowness property}. We have to prove that this allocation leaves a sufficiently high value for the losing agents.

After all the removals, the number of agents changes by $\Delta n = -1$ since one agent leaves; the cake value for a losing agent changes by $\Delta V$ (a negative quantity); and the number of teeth changes by $\Delta T$ which may be positive or negative.  Looking at the value requirement $V \geq 2n+T-1$, we see that in order to use the induction assumption, it is sufficient to prove that for every loser:
\begin{align*}
\Delta V \geq 2 \Delta n + \Delta T = \Delta T - 2
\end{align*}
so the value of each loser should drop by at most two units, plus one unit for each removed tooth.

The following analysis depends on whether the winning square is adjacent to a right corner $(\xright_i,y_i)$ as in the illustration above, or a left corner $(\xleft_i,y_i)$. The two cases are entirely symmetric; henceforth we assume that the winning square is adjacent to a right corner.

First, we handle the \emph{Sunlight property} by cutting from the left edge of the winning square down to the bottom border of $C$:
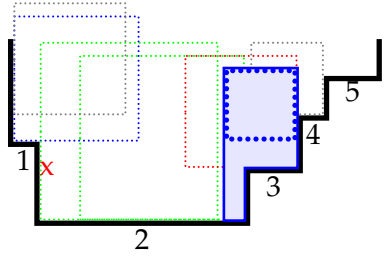
\begin{center}
\begin{pspicture}(0,4)(14,14)
\valleydeep{}
\psset{dotsep=1pt,linestyle=dotted,linecolor=blue}
\losingsquares{}
\pspolygon[linecolor=blue,linewidth=1pt,linestyle=solid,fillcolor=blue!10,fillstyle=solid](10.9,7.1)(10.9,10.9)(8.1,10.9)(8.1,5.1)(8.9,5.1)(8.9,7.1)(10.9,7.1)
\Square[linecolor=blue,linewidth=2pt](8.1,8.1){2.8}
\end{pspicture}
\end{center}
The winning square casts a shadow on $m\geq 0$ teeth below it, which are all removed. In the illustration above, $m=1$. Additionally, a new tooth is added at the top-left of the winning square. Additionally, if the winning square is higher than the tooth at its right (as in the figure), then that tooth is removed and a new tooth is added at the top-right of the winning square. All in all, $\Delta T = 1-m$. 

The winning square casts a shadow on $1+m$ levels. All squares of the losing agents in these levels are higher than the winning square; hence, the shadows of the winning square are contained in the losers' squares, and the total value of the shadows is at most $1+m$. All in all, $\Delta V \geq -1-m = \Delta T-2$, as required.

Next, we have to handle the \emph{Shallowness property} by removing \emph{deep levels} --- levels for which $dy_j>dx_j$ or equivalently:
\begin{align}
\label{eq:depth}
\min(\yright_j,\yleft_j)-y_j > \xright_j-\xleft_j
\end{align}
This is done separately to the left and to the right of the winning square:
\begin{itemize}
\item A level to the left of the winning square ($j<i$) may become deep if the left edge of the winning square, and the cut from that edge downwards, becomes its rightmost wall: 
\begin{align*}
\xright_j \leftarrow \xleft_i 
&&
\yright_j \leftarrow y_i+l_i
\\
&(y_i+l_i)-y_j > \xleft_i -\xleft_j&
\end{align*}
\item A level to the right of the winning square ($j>i$) may become deep if the right edge of the winning square becomes its leftmost wall:
\begin{align*}
\xleft_j \leftarrow \xright_i
&&
\yleft_j \leftarrow y_i+l_i
\\
&(y_i+l_i)-y_j > \xright_j - \xright_i&
\end{align*}
\end{itemize}

In each side, we remove the highest deep level, and all levels below it. In the illustration below, only level 4 (to the right of the winning square) is removed:
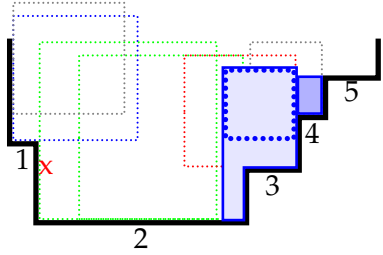
\begin{center}
\begin{pspicture}(0,4)(14,14)
\valleydeep{}
\psset{dotsep=1pt,linestyle=dotted,linecolor=blue}
\losingsquares{}
\pspolygon[linecolor=blue,linewidth=1pt,linestyle=solid,fillcolor=blue!10,fillstyle=solid](10.9,7.1)(10.9,10.9)(8.1,10.9)(8.1,5.1)(8.9,5.1)(8.9,7.1)(10.9,7.1)
\psframe[linecolor=blue,linewidth=1pt,linestyle=solid,fillcolor=blue!30,fillstyle=solid](10.9,9.1)(11.9,10.6)
\Square[linecolor=blue,linewidth=2pt](8.1,8.1){2.8}
\end{pspicture}
\end{center}
By selection of the global winner: $y_j+l_j\geq y_i+l_i$, which implies:
\begin{align}
\label{eq:lj}
l_j \geq (y_i+l_i) - y_j
\end{align}
If a level $j<i$ becomes deep, then (\ref{eq:lj}) implies:
\begin{align*}
& l_j > \xleft_i - \xleft_j
\\
\implies & \xleft_j+l_j > \xleft_i .
\end{align*}
In addition to $y_j+l_j \geq y_i+l_i$, this implies that the removed rectangle $[\xleft_j,\xleft_i]\times[y_j,\min(y_i+l_i,\yleft_j)]$ is contained in the corner-winner: $[\xleft_j,\xleft_j+l_j]\times[y_j,y_j+l_j]$. 
Hence, the value of the removed rectangle is at most 1. At most one unit of value is removed, and one tooth is removed. Hence, the balance between $\Delta V$ and $\Delta T$ is maintained.

If a level $j>i$ becomes deep, then (\ref{eq:lj}) implies:
\begin{align*}
& l_j > \xright_j - \xright_i
\\
\implies & \xright_j-l_j < \xright_i .
\end{align*}
In addition to $y_j+l_j \geq y_i+l_i$, this implies that the removed rectangle $[\xright_i,\xright_j]\times[y_j,\min(y_i+l_i,\yright_j)]$ is contained in the corner-winner: $[\xright_j-l_j,\xright_j]\times[y_j,y_j+l_j]$. Hence, the value of the removed rectangle is at most 1. At most one unit of value is removed, and one tooth is removed. The balance between $\Delta V$ and $\Delta T$ is maintained.

Finally, we have to handle the \emph{Sunlight property} again by removing all levels below the levels removed in the previous step. We now prove that in all such levels, no agent marked any square. Indeed, let $j$ be a level that became deep, and $k$ be a level shadowed by it. So $y_k<y_j$ and $\xleft_k>\xleft_j$ and $\xright_k < \xright_j$. The side-length of any square marked in level $k$ is at most $\xright_k-\xleft_k$, so $l_k<\xright_k-\xleft_k<\xright_j-\xleft_j$ and:
\begin{align*}
y_k+l_k < y_j+(\xright_j-\xleft_j)
\end{align*}
Combining this with (\ref{eq:depth}) gives:
\begin{align*}
y_k+l_k < \min(\yright_j,\yleft_j) \leq y_i+l_i
\end{align*}
but this contradicts the assumption that $i$ is the global-winning-square. Therefore, all levels below a deep level have a value of less than 1 to all agents. At most one unit of value is removed per level, so the balance between $\Delta V$ and $\Delta T$ is maintained. 

To summarize: after allocating the winning square to the winner and removing some parts of the valley, we have a new valley with $T+\Delta T$ teeth satisfying the Sunlight and the Shallowness properties, and each losing agent values it as at least $((2n-1+T)+\Delta V) \geq ((2n-1+T)+(\Delta T-2)) = 2(n-1)-1+(T+\Delta T)$. Therefore, by the induction assumption we can continue to divide it among the $n-1$ losers.
\qed

The above procedure proves that, for every $n\geq 1,T\geq 1$:

\begin{align*}
\prop(T\,valley,n,Squares)=\frac{1}{2n-1+T}
\end{align*}

A square with 3 walls is a valley with no teeth. It obviously satisfies the Sunlight property and the Shallowness property. Letting $T=0$ in the above formula yields:
\begin{align*}
\prop(Square\,with\,three\,walls,\,n,\,Squares)=\frac{1}{2n-1}
\end{align*}
matching the upper bound.

\subsubsection{Remark}


We divided a 2-walls square by generalizing it to a staircase, and divided a 3-walls square by generalizing it to a valley. The natural next step is to divide a 4-walls square by generalizing it to a rectilinear polygon. This is a much more challenging task even for a single agent. The algorithmic problem of finding a minimal square-covering for a rectilinear polygon has been solved by \citet{BarYehuda1996Lineartime}, and we believe that their algorithm can be used for developing a rectilinear polygon division procedure. However, this algorithm is much more complicated than our covering algorithm of Subsection \ref{sub:cover-valley}, so the division procedure will probably also be much more complicated. 

In the next subsection we present a procedure for dividing a square using a different approach, which works only when the value measures are identical.

\subsection{Four walls, guillotine cuts, identical valuations} \label{sec:Half-proportional-division-identical-valuations}
Our procedures for identical valuations differ from the other procedures in that they do not use auctions, since all agents would make the same bids anyway.

We develop simultaneously a pair of division procedures. Both procedures accept a cake $C$ which is assumed to be the rectangle $[0,1]\times[0,L]$, and a \emph{single} continuous value measure $V$. They return some disjoint square pieces $\{X_{i}\}$ such that for every $i$: $V(X_{i})\geq 1$.

The two procedures differ in their requirement on $L$ (the height/length ratio of the cake) and in the number of ``walls'' (bounded sides) they assume on the cake: 
\begin{itemize}
\item The \emph{fat-procedure} requires that $L\in[1,2]$ (i.e, the cake is a 2-fat rectangle) and it guarantees that all allocated squares are contained in $C$;
\item The \emph{thin-procedure} requires that $L\in [2,\infty)$ (i.e, the cake is a ``2-thin'' rectangle) and it returns one of the following two outcomes:
\begin{enumerate}
\item $n-1$ squares contained in $C$ (i.e, bounded by the 4 walls of the cake), or -
\item $n$ squares contained in $[0,\infty]\times[0,L]$, i.e, bounded by only 3 walls but may flow over the rightmost border. Every square that flows over the rightmost border is guaranteed to have its leftmost edge adjacent to the leftmost edge of $C$ and its side-length at most $L-1$ (the longer side of the cake minus its shorter side), so that all squares are contained in $[0,L-1]\times[0,L]$.
\end{enumerate}
\end{itemize}
Additionally, the two procedures differ in their requirement on the total cake value:
\begin{itemize}
\item The fat-procedure requires that $V(C)\geq 2n$.
\item The thin-procedure requires that $V(C)\geq 2n-2$.
\end{itemize}
The procedures are developed by induction on $n$. We first consider the base case $n=1$:
\begin{itemize}
\item In the fat-procedure, the cake value is 2 and the cake is 2-fat, so by the Covering Lemma it contains a square with a value of at least 1.
\item The thin-procedure can just return an empty set. This is an instance of the first outcome --- $n-1$ squares contained in $C$.
\end{itemize}
We now assume that both procedures work well for any number less than $n$. Given $n\geq 2$, we proceed as in the following subsections.

Henceforth, we make the following \textbf{positivity assumption}: every piece with positive area has positive value. This assumption is only for convenience: it simplifies the presentation and reduces the number of cases to consider. It can be dropped by adding sub-cases to each case in the procedures.
\subsubsection{Fat procedure} \label{sub:fat-procedures}
At this point, the cake is a 2-fat rectangle with width 1 and height $L\in [1,2]$. Its total value is $2n$, and $n\geq 2$.

For every integer $u\in[0,2n]$, let $y_u$ be the value $y\in[0,L]$ such that the cake below $y$ has value $u$: $V([0,1]\times[0,y_{u}])=u$. By the positivity assumption, $y_u$ is unique, $y_0=0$ and $y_{2n}=L\geq1$. Therefore, there exists a smallest $k\in[1,n]$ such that: $y_{2k}\geq 1/2$. Let $Bottom := [0,1]\times[0,y_{2k}] = $ the cake below $y_{2k}$; note that it is a 2-fat rectangle. Let $Top:=C\setminus Bottom = [0,1]\times[y_{2k},L] = $ the cake above $y_{2k}$. We have $V(Bottom)=2k$ and $V(Top)=2(n-k)$. Now there are two cases:

\psset{unit=2.5cm}
\newcommand{\yvalue}[2]{\rput[r](-0.1,#1){$#2$}}
\newcommand{\yvalueline}[2]{
 \rput[r](-0.1,#1){$#2$}
 \psline[linestyle=dotted](0,#1)(1.1,#1)
}
\newcommand{\landcakefat}{
\psframe[linecolor=brown](0,0)(1,1.6)
\yvalue{1.6}{L}
\yvalue{0}{0}
}

\textbf{Case A}: $L-y_{2k}\geq 1/2$ (this implies $k<n$). Thus $Bottom$ and $Top$ are both 2-fat rectangles: 
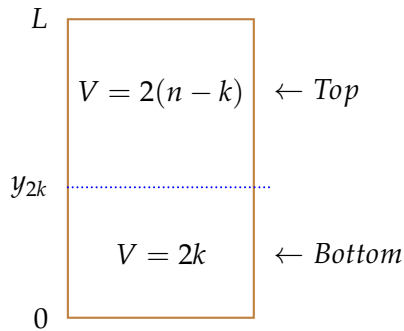
\begin{center}
\begin{pspicture}(-1,-0.1)(3,1.7)
 \landcakefat
 \psline[linestyle=dotted](0,0.7)(1.1,0.7)
 \yvalue{0.7}{y_{2k}}
 \rput[l](1.1,1.2){$\leftarrow Top$} \rput(0.5,1.2){\shortstack{$V=2(n-k)$}}
 \rput[l](1.1,0.35){$\leftarrow Bottom$}  \rput(0.5,0.35){\shortstack{$V=2k$}}
\end{pspicture}
\end{center}
Apply the fat procedure to $Bottom$ and $Top$ and get $k+(n-k)=n$ squares contained in $C$.

\textbf{Case B}: $L-y_{2k}<\frac{1}{2}$, so $Bottom$ is 2-fat and $Top$ is 2-thin. Now consider $y_{2k-2}$. By definition of $k$,
$y_{2k-2}<\frac{1}{2}$. Let $Bottom':=[0,1]\times[0,y_{2k-2}]$ and
$Top':=C\setminus Bottom'=[0,1]\times[y_{2k-2},L]$, so $V(Bottom')=2(k-1)=V(Bottom)-2$
and $V(Top') = 2(n-k+1)=V(Top)+2$. Note that $Bottom'$ is 2-thin and is contained in $Bottom$, and $Top'$ is 2-fat and contains $Top$:

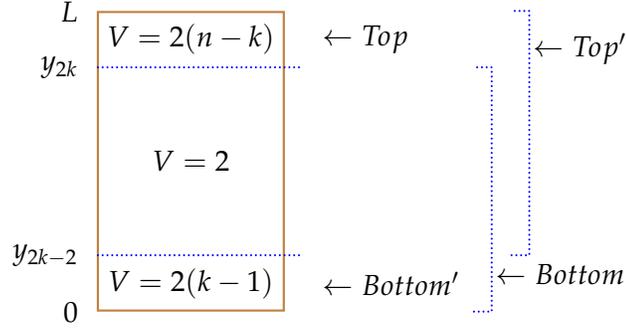
\newcommand{\landcaketopbottom}{
 \landcakefat
 \yvalue{1.3}{y_{2k}}
 \psline[linestyle=dotted](0,1.3)(1.1,1.3)
 \yvalue{0.3}{y_{2k-2}}
 \psline[linestyle=dotted](0,0.3)(1.1,0.3)
}
\begin{center}
\begin{pspicture}(-1,-0.1)(3,1.7)
 \landcaketopbottom
 \rput(0.5,1.45){$V=2(n-k)$}
 \rput(0.5,0.8){$V=2$}
 \rput(0.5,0.15){$V=2(k-1)$}
 \rput[l](1.2,1.45){$\leftarrow Top$}
 \rput[l](1.2,0.15){$\leftarrow Bottom'$}
 \psline[linecolor=blue,linestyle=dotted](2,0)(2.1,0)(2.1,1.3)(2,1.3)
 \rput[l](2.12,0.2){$\leftarrow Bottom$}
 \psline[linecolor=blue,linestyle=dotted](2.2,0.3)(2.3,0.3)(2.3,1.6)(2.2,1.6)
 \rput[l](2.32,1.4){$\leftarrow Top'$}
\end{pspicture}
\end{center}

Because here $n\geq2$, either $n-k\geq1$ or $k-1\geq1$ or both. Hence, at least one of the two 2-thin parts ($Top$, $Bottom'$) is non-empty and with value at least 2. Use the \emph{thin procedure} to divide the non-empty thin part/s. In each part there are two possible outcomes: a smaller number of squares within 4 walls or a larger number of squares within 3 walls. There are several cases to consider.

--- One easy case is that we get the 4-walls outcome in at least one of the parts --- either in $Top$ or in $Bottom'$ or in both. Suppose that we get the 4-walls outcome in $Bottom'$. So we have $k-1$ squares within the 4 walls of $Bottom'$. Ignore the outcome on $Top$ and apply the fat procedure to $Top'$. This results in $n-k+1$ additional squares, so we have the required $n$ squares. The situation is analogous if we get the 4-walls outcome in $Top$.

--- Another easy case is that we get the 3-walls outcome in one part, and the other part is empty. Suppose that $Top$ is empty (this implies $k=n$) and we get the 3-walls outcome in $Bottom'$. So we have $(k-1)+1=n$ squares contained in $[0,1]\times [0,1-y_{2k-2}] \subseteq C$, as required. The situation is analogous if $Bottom'$ is empty and we get the 3-walls outcome in $Top$.

--- The hard case is that both $Top$ and $Bottom'$ are non-empty and the thin procedure on both of them returns the 3-walls outcome. Now we have $k$ bottom squares and $n-k+1$ top squares, for a total of $n+1$ squares, e.g:
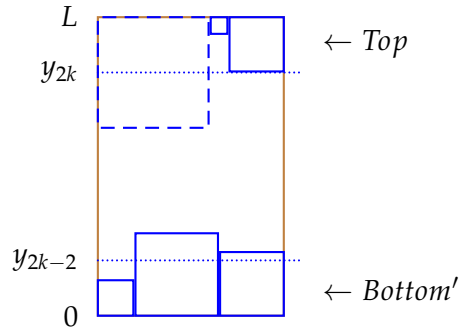
\begin{center}
\begin{pspicture}(-1,-0.1)(3,1.7)
 \landcaketopbottom
 \rput[l](1.2,1.45){$\leftarrow Top$}
 \rput[l](1.2,0.15){$\leftarrow Bottom'$}
 
 \Square(0,0){0.2}
 \Square(0.2,0){0.45}
 \Square(0.65,0){0.35}

 \Square[linestyle=dashed](0,1){0.6}
 \Square(0.6,1.5){0.1}
 \Square(0.7,1.3){0.3}
\end{pspicture}
\end{center}
A potential problem in the last step is that some of the squares might overlap: some top squares might flow over the lower boundary of Top and overlap a bottom square, or some bottom squares might flow over the upper boundary of Bottom' and overlap a top square. To prevent an overlap, we remove a single square --- the largest of the $n+1$ squares (dashed square in the illustration above) --- and return the remaining $n$ squares. 

It remains to prove that, indeed, after the largest square is removed, the remaining $n$ squares do not overlap. The proof is purely geometric and it is delegated to Appendix \ref{sec:dove-hawk}. \qed

\subsubsection{Thin procedure} \label{sub:thin}
At this point, the cake is a 2-thin rectangle with width 1 and height $L\in [2,\infty)$. Its total value is $2n-2$, and $n\geq 2$. The procedure is allowed to return one of two outcomes:

\textbf{Outcome \#1:} $n-1$ squares bounded by the 4 walls of $C$, i.e, contained in $[0,1]\times[0,L]$, or ---

\textbf{Outcome \#2:} $n$ squares bounded by the 3 walls of $C$,
i.e, contained in $[0,\infty]\times[0,L]$. 
In this case, every square that flows over the rightmost border must have its leftmost edge adjacent to the leftmost edge of $C$ (the edge $x=0$), and its side-length must be at most $L-1$ (the longer side of $C$ minus its shorter side). This means that all $n$ squares must be contained in $[0,L-1]\times[0,L]$.

We first handle the case $n=2$, in which $V=2$.

Select $y\in[0,L]$ such that $V([0,1]\times[0,y])=V([0,1]\times[y,L])=1$. Proceed according to the value of $y$:

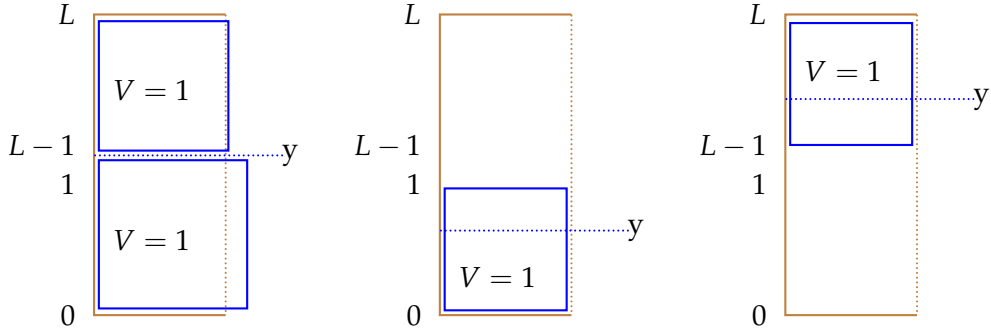
\newcommand{\landcakethin}{
\psline[linestyle=solid,linecolor=brown](0.7,0)(0,0)(0,1.6)(0.7,1.6)
\psline[linestyle=dotted,linecolor=brown](0.7,1.6)(0.7,0)
\yvalue{1.6}{L}
\yvalue{0}{0}
}

\begin{center}
\begin{pspicture}(-0.5,-0.1)(1.3,1.7)
\landcakethin
\psline[linestyle=dotted](0,0.85)(1,0.85)
\yvalue{0.9}{L-1}
\yvalue{0.7}{1}
\rput[l](1,0.85){y}
\Square(0.02,0.03){0.8} \rput[l](0.1,0.4){$V=1$}
\Square(0.02,0.87){0.7} \rput[l](0.1,1.2){$V=1$}
\end{pspicture}
\begin{pspicture}(-0.5,-0.1)(1.3,1.7)
\landcakethin
\psline[linestyle=dotted](0,0.45)(1,0.45)
\yvalue{0.9}{L-1}
\yvalue{0.7}{1}
\rput[l](1,0.45){y}
\Square(0.02,0.02){0.66}  \rput[l](0.1,0.2){$V=1$}
\end{pspicture}
\begin{pspicture}(-0.5,-0.1)(1.3,1.7)
\landcakethin
\psline[linestyle=dotted](0,1.15)(1,1.15)
\yvalue{0.9}{L-1}
\yvalue{0.7}{1}
\rput[l](1,1.15){y}
\Square(0.02,0.9){0.66}  \rput[l](0.1,1.3){$V=1$}
\end{pspicture}
\end{center}

\begin{itemize}
\item If $y\in[1,L-1]$ (left) then return the two squares $[0,y]\times[0,y]$ and $[0,L-y]\times[y,L]$. Both squares are in $[0,L-1]\times[0,L]$
with their \West  side at $x=0$; this is an instance of outcome \#2.
\item If $y\in[0,1)$ (middle) then return $[0,1]\times[0,1]$; if $y\in(L-1,L]$ (right) then return $[0,1]\times[L-1,L]$. Both cases are instances of outcome \#1.
\end{itemize}

From now on we assume that $n\geq3$.

For every $u\in[0,2n-2]$, define $y_{u}$ as the value $y\in[0,L]$ such that the cake below $y$ has value $u$: $V([0,1]\times[0,y_{u}])=u$. By the positivity assumption, $y_u$ is unique and $y_0=0$ and $y_{2n-2}=L$. Therefore, there exists a smallest $k\in[1,n-1]$ such that: $y_{2k}\geq\frac{1}{2}$. Mark the cake below $y_{2k}$ ($[0,1]\times[0,y_{2k}]$) as $Bottom$ and the part above it ($[0,1]\times[y_{2k},L]$) as $Top$.
We have $V(Bottom)=2k$ and $V(Top)=2(n-k-1)$. 

Now there are two cases:

\textbf{Case A}: $L-y_{2k}\geq\frac{1}{2}$ (this implies $k<n-1$). Thus each of $Bottom$ and $Top$ is either 2-fat, or 2-thin with its longer side vertical --- parallel to the open side of $C$ (this means that we can divide it using the Thin Procedure letting the pieces flow over its rightmost border). 
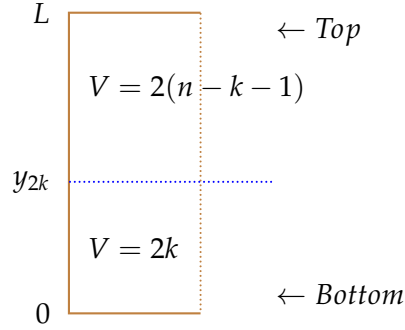
\newcommand{\landcaketopbottomthin}{
 \landcakethin
 \yvalue{1.3}{y_{2k}}
 \psline[linestyle=dotted](0,1.3)(1.1,1.3)
 \yvalue{0.3}{y_{2k-2}}
 \psline[linestyle=dotted](0,0.3)(1.1,0.3)
}

\begin{center}
\begin{pspicture}(-0.5,-0.1)(2,1.7)
 \landcakethin
 \psline[linestyle=dotted](0,0.7)(1.1,0.7)
 \yvalue{0.7}{y_{2k}}
 \rput[l](1.1,1.5){$\leftarrow Top$} \rput[l](0.1,1.2){\shortstack{$V=2(n-k-1)$}}
 \rput[l](1.1,0.1){$\leftarrow Bottom$}  \rput[l](0.1,0.35){\shortstack{$V=2k$}}
\end{pspicture}
\end{center}
Apply the fat procedure or the thin procedure, whichever is appropriate, to $Bottom$ and $Top$. In each part there are two possible outcomes: a smaller number of squares within 4 walls, or a larger number of squares within 3 walls.  

--- If we get the 4-walls outcome in both parts, then we have $k+(n-k-1) = n-1$ squares within the 4 walls of $C$, which is an instance of Outcome \#1.

--- If we get the 4-walls outcome in one part and the 3-walls outcome in the other part, then we have $k+(n-k)=n$ or $(k+1)+(n-k-1)=n$ squares within 3 walls. By the induction assumption, the thin procedure guarantees that all squares flowing over the rightmost border have their leftmost edge adjacent to the leftmost wall $x=0$, and their side-length at most the longer side minus the shorter side. Here, the longer side of both Bottom and Top is less than $L$ and their shorter side is 1, so all these squares are contained in $[0,L-1]\times[0,L]$, so we have an instance of Outcome \#2.

--- If we get the 3-walls outcome in both parts, then we have $k+(n-k)+1 = n+1$ squares within 3 walls. We can discard one square arbitrarily and remain with $n$ squares as in the above case, which is again an instance of Outcome \#2.

\textbf{Case B}: $L-y_{2k}<\frac{1}{2}$, so $Bottom$ is 2-fat or 2-thin with a vertical long side  (parallel to the open side of $C$), and $Top$ is 2-thin with horizontal long side (perpendicular to the open side of $C$). Now consider $y_{2k-2}$. By definition of $k$, $y_{2k-2}<\frac{1}{2}$. let $Bottom'=[0,1]\times[0,y_{2k-2}]$
and $Top'=[0,1]\times[y_{2k-2},L]$, so $V(Bottom')=2(k-1)=V(Bottom)-2$
and $V(Top')=2(n-k)=V(Top)+2$. Note that $Bottom'$ is 2-thin with a horizontal long side and ir is contained in $Bottom$, and $Top'$ is 2-fat or 2-thin with a vertical long side and it contains $Top$:
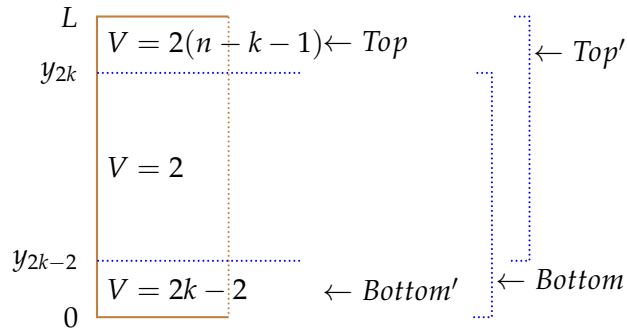
\begin{center}
\begin{pspicture}(-0.5,-0.1)(3,1.7)
 \landcaketopbottomthin
 \rput[l](1.2,1.45){$\leftarrow Top$}
 \rput[l](0.05,1.45){$V=2(n-k-1)$}
 \rput[l](0.05,0.8){$V=2$}
 \rput[l](0.05,0.15){$V=2k-2$}
 \rput[l](1.2,0.15){$\leftarrow Bottom'$}

 \psline[linecolor=blue,linestyle=dotted](2,0)(2.1,0)(2.1,1.3)(2,1.3)
 \rput[l](2.12,0.2){$\leftarrow Bottom$}
 \psline[linecolor=blue,linestyle=dotted](2.2,0.3)(2.3,0.3)(2.3,1.6)(2.2,1.6)
 \rput[l](2.32,1.4){$\leftarrow Top'$}
\end{pspicture}
\end{center}
At this point $n\geq3$, so either $n-k-1\geq1$ or $k-1\geq1$ or both. Hence, at least one of the two horizontal thin parts ($Top$, $Bottom'$) is non-empty and with value at least 2. Use the \emph{thin procedure} on the non-empty horizontal part/s. In each part there are two possible outcomes: a smaller number of squares within 4 walls or a larger number of squares within 3 walls. There are several cases to consider.

--- One easy case is that we get the 4-walls outcome in at least one of the parts --- either in $Top$ or in $Bottom'$ or in both. Suppose that we get the 4-walls outcome in $Bottom'$ (the situation is analogous if we get the 4-walls outcome in $Top$). So we have $k-1$ squares within the 4 walls of $Bottom'$. We ignore the outcome on $Top$ and proceed to get additional squares from $Top'$. Apply to $Top'$ either the fat procedure (if it is 2-fat) or the thin procedure (if it is 2-thin with a vertical long side). One possibility is that we get $n-k$ additional squares contained in $Top'$; then we have a total of $n-1$ squares contained in $C$, which is an instance of Outcome \#1. Another possibility is that we get $n-k+1$ additional squares bounded by only three walls of $Top'$; by the induction assumption and the guarantees of the Thin Procedure, the squares that flow over the rightmost border of $Top'$ are adjacent to its leftmost wall, which coincides with the leftmost wall of $C$. Their side-length is at most the longer side-length of $Top'$ minus its shorter side-length; the longer side-length of $Top'$ is less than $L$ and its shorter side-length is 1, so the side-length of all the additional squares is at most $L-1$, and we have an instance of Outcome \#2.

--- Another easy case is that we get the 3-walls outcome in one part, and the other part is empty. Suppose that $Top$ is empty (this implies $k=n-1$) and we get the 3-walls outcome in $Bottom'$. So we have $(k-1)+1=n-1$ squares contained in $[0,1]\times [0,1-y_{2k-2}] \subseteq C$, which is an instance of Outcome \#1. The situation is analogous if $Bottom'$ is empty and we get the 3-walls outcome in $Top$.

--- The hard case is that both $Top$ and $Bottom'$ are non-empty and the thin procedure on both of them returns the 3-walls outcome. We now have the following squares:
\begin{itemize}
\item $k\geq1$ bottom squares in $[0,1]\times[0,1-y_{2k-2}]$;
\item $n-k\geq1$ top squares in $[0,1]\times[L-1+(L-y_{2k}),L]$.
\end{itemize}
Because $L\geq2$, no squares overlap:
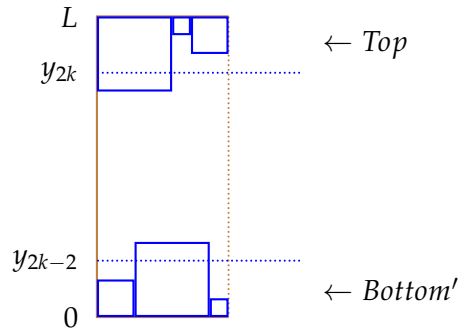
\begin{center}
\begin{pspicture}(-0.5,-0.1)(2,1.7)
 \landcaketopbottomthin
 \rput[l](1.2,1.45){$\leftarrow Top$}
 \rput[l](1.2,0.15){$\leftarrow Bottom'$}
  
 \Square(0,0){0.2}
 \Square(0.2,0){0.4}
 \Square(0.6,0){0.1}
 \Square(0,1.2){0.4}
 \Square(0.4,1.5){0.1}
 \Square(0.5,1.4){0.2}
\end{pspicture}
\end{center}
We now have $n$ squares within the 4 walls of $C$, which is more than we need for Outcome \#1. \qed
\\
\\
The guarantees of the Fat Procedure imply that, for all $n\geq 2$:
\begin{align*}
\propsame(Square\,with\,4\,walls,\,n,\,Squares)\geq\frac{1}{2n}
\end{align*}
which exactly matches the upper bound of Claim \ref{claim:neg-square}.

\subsubsection{Fat rectangle pieces}
When the pieces are allowed to be $R$-fat rectangles, the above lower bound is of course still valid. But when $R\geq 2$, 
the Fat Procedure can give a slightly stronger guarantee - the required value is $2n-1$ instead of $2n$ (the Thin Procedure is unchanged). The required modifications in the Fat Procedure are briefly explained below:
\begin{itemize}
\item In the base case ($n=1$), the cake value is 1 and it is 2-fat, so the procedure returns the entire cake as a single piece.
\item In the main procedure ($n\geq 2$), we first try to cut the cake horizontally to two 2-fat rectangles and apply the Fat Procedure to each of them. For this, we need to find some $y\in[1/2,L-1/2]$ such that the value below $y$ is at least $2k-1$ and the value above $y$ is at least $2(n-k)-1$, for some integer $k\geq 1$. Then, both the part below $y$ and the part above $y$ are 2-fat. By the induction assumption, the Fat Procedure finds $k$ 2-fat-rectangles in the bottom part and $n-k$ 2-fat-rectangles in the top part, so we are done.
\item If we cannot find such $y$, this means that for all $y\in[1/2,L-1/2]$ and every integer $k'$, either the value below $y$ is less than $2k'-1$ or the value above $y$ is less than $2(n-k')-1$. But the latter condition implies that the value below $y$ is more than $2k'$, so the condition becomes: for all $y\in[1/2,L-1/2]$ and every integer $k'$, the value below $y$ is either less than $2k'-1$ or more than $2k'$. So for all $y\in[1/2,L-1/2]$, the value below $y$ is in the open interval $(2k-2,2k-1)$ for some integer $k\geq 1$. This means that the cake looks like this, for some integer $k$:
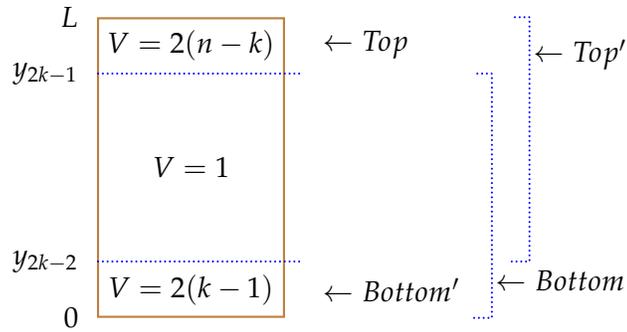
\begin{center}
\psset{unit=2.5cm}	
\begin{pspicture}(-.5,-0.1)(3,2)
\landcakefat
\yvalueline{0.3}{y_{2k-2}}
\rput(.5,.15){\shortstack{$V=2(k-1)$}}

\yvalueline{1.3}{y_{2k-1}}
\rput(.5,1.45){\shortstack{$V=2(n-k)$}}

\rput(.5,.8){\shortstack{$V=1$}}
 \rput[l](1.2,1.45){$\leftarrow Top$}
 \rput[l](1.2,0.15){$\leftarrow Bottom'$}

 \psline[linecolor=blue,linestyle=dotted](2,0)(2.1,0)(2.1,1.3)(2,1.3)
 \rput[l](2.12,0.2){$\leftarrow Bottom$}
 \psline[linecolor=blue,linestyle=dotted](2.2,0.3)(2.3,0.3)(2.3,1.6)(2.2,1.6)
 \rput[l](2.32,1.4){$\leftarrow Top'$}
\end{pspicture}
\end{center}
where $y_{2k-2}<1/2$ and $y_{2k-1}>L-1/2$. Hence, the parts $Top:=[0,1]\times[y_{2k-1},L]$ and $Bottom':=[0,1]\times[0,y_{2k-2}]$ are both 2-thin rectangles (one of these parts may be empty). $V(Top)=2(n-k)$ and $V(Bottom')=2(k-1)$. This is exactly the same situation as in Case B of the original procedure. We can now apply the Thin Procedure to $Top$ and to $Bottom'$ and proceed according to the outcomes. \qed
\end{itemize}
Therefore, for all $n\geq 2$ and $R\geq 2$:
\begin{align*}
\propsame(Square\,with\,4\,walls,\,n,\,R\,fat\,rectangles)\geq\frac{1}{2n-1}
\end{align*}
which exactly matches the upper bound of Claim \ref{claim:neg-square-fatrects}.

\subsubsection{Remark}
The above procedures work only when the value measures are identical. The main reason is that the Thin procedure may return one of two outcomes. When there is a single value measure, the returned outcome is unique. But when there are different value measures, each value measure may induce a different outcome, and the different outcomes may be incompatible. 

\subsection{Compact cakes of any shape} \label{sub:Selection-algorithms}
As explained in Subsection \ref{sub:intro-any-shape}, when the cake can be of an arbitrary shape, $\prop(C,n,S)$ may be arbitrarily small. Hence it makes sense to assess the fairness of an allocation for a particular agent relative to the total utility that this agent can get in an $S$-piece when given the entire cake. This intuition is captured by the following definition. It is an analogue of Definition \ref{def:abs-prop}, the only difference being that the normalization factor is the cake utility $V^S(C)$ instead of the cake value $V(C)$:
\begin{defn}
\emph{\label{def:rel-prop}(Relative proportionality)} For a cake $C$, a family of usable pieces $S$ and an integer $n\geq1$:

(a) The \emph{relative proportionality level} of $C$, $S$ and $n$, marked $\proprel(C,n,S)$, is the largest fraction $r\in[0,1]$ such that, for every set of  $n$ value measures $(V_i,...,V_n)$,
there exists an $S$-allocation $(X_1,...,X_n)$ for which $\forall i:\,V_i(X_i)/V_i^{S}(C)\geq r$.

(b) The \emph{same-value relative proportionality level} of $C$, $S$ and $n$, marked $\proprelsame(C,n,S)$, is the largest fraction $r\in[0,1]$ such that, for every single value measure $V$, there exists an $S$-allocation $(X_1,...,X_n)$ for which $\forall i:\,V(X_i)/V^S(C)\geq r$.
\end{defn}
Our first result involves parallel squares.
\begin{claim}
\label{claim:8n-6} For every cake $C$ which is a compact subset of $\mathbb{R}^{2}$:
\begin{align*}
\proprel(C,\,n,\,Parallel\,squares)\geq\frac{1}{8n-6}
\end{align*}
\end{claim}
\begin{proof}
We normalize the valuations of all agents such that, for every agent $i$, $V_i^S(C)=8n-6$. We show a division procedure giving each agent a square with a value of at least 1.

(1) \textbf{Preparation:} Each agent $i$ draws a ``best square'' in $C$ --- a square $q_{i}$ that maximizes $V_i$. The existence of such a square can be proved based on the compactness of the set of squares in $C$; this is done in Appendix \ref{sec:Conditions-for-Existence}. By definition of the utility function $V^S$, for every $i$: $V_i(q_i)=V_i^S(C) = 8n-6$.

(2) \textbf{Mark auction:} Let $N:=4n-3$. Ask each agent $i$ to mark, inside $q_i$, $N$ pairwise-disjoint parallel squares with a value of 1 (the agent can do so by using the division procedure for identical value measures described in Subsection  \ref{sec:Half-proportional-division-identical-valuations}: this procedure finds $N$ squares in $q_i$, each of which has a value of at least $V_i(q_i)/(2N) = 1$). Let $Q_i$ be the collection of $N$ squares marked by $i$. 

An agent's bid is interpreted as saying ``I am willing to give my entitlement to a piece of $C$ in exchange for any square in $Q_i$". Our goal now is to allocate to each agent $i$ a single piece from the collection $Q_i$ such that the $n$ allocated pieces are pairwise-disjoint. 

(3) \textbf{Winner selection:} a smallest square in $\cup_i Q_i$ is selected as the winning bid (if there several smallest squares, one is selected arbitrarily). Denote the selected smallest square by $q^*$ and suppose it belongs to agent $i$. Agent $i$ now receives $q^*$ and goes home.

(4) \textbf{Bid adjustment:} For each agent $j\neq i$, remove from $Q_j$ all squares that overlap $q^*$. Since the squares in $Q_j$ are all pairwise-disjoint and not smaller than $q^*$, the number of squares removed is at most 4. This is based on the following geometric fact: given a square $q$, there are at most 4 parallel squares that are larger than $q$, overlap $q$ and do not overlap each other. This is because each square larger than $q$ which overlaps $q$, must overlap one of its 4 corners, so there can be at most 4 such squares:
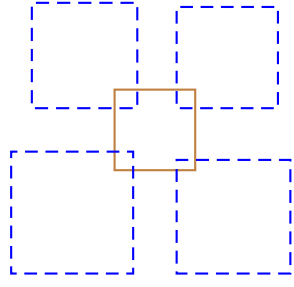
\begin{center}
\psset{unit=0.55mm}
\begin{pspicture}(80,71)
\Square[linecolor=brown](30,30){20}
\Square[linecolor=blue,linestyle=dashed](5,5){30}
\Square[linecolor=blue,linestyle=dashed](45,45){25}
\Square[linecolor=blue,linestyle=dashed](45,5){28}
\Square[linecolor=blue,linestyle=dashed](10,45){26}
\end{pspicture}
\end{center}
After the removal, each of the remaining $n-1$ agents has a collection of at least $4(n-1)-3$ squares. If only a single agent remains, then his collection contains at least 1 square; allocate this square to the single agent and finish. Otherwise, go back to step (3) and select the next winner from the remaining $n-1$ agents.

Finally, each agent $i\range{1}{n}$ holds a square from the collection $Q_i$. This square has a value of at least 1, which proves the claim.
\end{proof}
The proof of Claim \ref{claim:8n-6} can be generalized to other families of usable pieces:
\begin{claim}
\label{claim:prel}
For a family of pieces $S$, define:
\begin{itemize}
\item $O_S =$ the largest number of pairwise-disjoint \spieces{} that overlap an \spiece{} with a smaller diameter.
\item $\propsame(S,n,S)=\inf_{C\in S}\allowbreak\propsame(C,n,S)$.
\end{itemize}
Then for every compact cake $C$ and every $n\geq1$:
\begin{align*}
\proprel(C,\,n,\,S)\,\,\geq\,\,\propsame(S,\,O_S \cdot (n - 1) + 1,\,S)
\end{align*}
\end{claim}
The proof is exactly the same as that of Claim \ref{claim:8n-6}, with only the constant 4 replaced by $O_S$, 3 replaced by $O_S-1$ and the function $1/(2 N)$ replaced by $\propsame(S,N,S)$.

When $S$ is the family of general (rotated) squares, $O_S = 8$.\footnote{We are grateful to Mark Bennet, Martigan, calculus, Red, Peter Woolfitt
and Dejan Govc for their help in calculating this number in http://math.stackexchange.com/q/1085687/29780
. Image credit: Dejan Govc. Licensed under CC-BY-SA 3.0.} 
\begin{center}
\end{center}
\begin{corollary}\label{claim:14n-12}
For every cake $C$ which is a compact subset of $\mathbb{R}^2$:
\begin{align*}
\proprel(C,\,n,\,Squares)\geq\frac{1}{16n-14}
\end{align*}
\end{corollary}
When $S$ is the family of parallel $R$-fat rectangles, $O_S = \lceil 2R+2 \rceil$:
\begin{corollary}\label{claim:R-fat-rel}
For every cake $C$ which is a compact subset of $\mathbb{R}^{2}$:
\begin{align*}
\proprel(C,\,n,\,Parallel\,R\,fat\,rectangles)
\geq
\frac{1}{2\lceil 2R+2 \rceil (n-1)+2}
\end{align*}
\end{corollary}
For completeness, we present the following trivial result regarding identical value measures:
\begin{claim}
\label{claim:prop-rel-same}
For every cake $C$ which is a compact subset of $\mathbb{R}^{2}$:
\begin{align*}
\proprelsame(C,\,n,\,Squares)=\frac{1}{2n}
\end{align*}
\end{claim}
\begin{proof}
Suppose the value measure of all $n$ agents is $V$. Let $q$ be a best square in $C$ --- a square that maximizes $V$. By definition of the utility function, $V(q)=V^{S}(C)$. Because $q$ is a square,
it is possible to allocate within it $n$ disjoint squares with a value of at least $V(q)/(2n) = V^S(C)/(2n)$.
\end{proof}

\subsubsection{Remarks}
1. The constant $O_S$ --- the largest number of pairwise-disjoint \spieces{} that overlap an \spiece{} with a smaller diameter --- has been used for developing approximation procedures for the problem of finding a maximum non-overlapping set \citep{Marathe1995Simple}.
The approximation factors are not tight. For example, for $n=2$, in step (b) we create $4n-3=5$ axis-parallel squares for each agent, but it is possible to prove that 3 squares per agent suffice for guaranteeing that a pair of disjoint squares exists. Hence, $\proprel(C,\,\n{2},\,Axis\,parallel\,squares)\geq1/6$.
What is the smallest number of squares required to guarantee the existence of $n$ disjoint squares? This open question is interesting because it affects both the proportionality coefficient in our fair cake-cutting
procedure and the approximation coefficient in the maximum disjoint set algorithm of \citet{Marathe1995Simple}. 

2. The Winner Selection procedure (step 3 in the proof) can be used even when the value functions of the agents are not additive or even not monotone (i.e. some parts of the land have negative utility to some agents). As long as every agent can draw $N$ disjoint squares, the procedure guarantees that he receives one of these pieces.

3. \citet{Iyer2009Procedure} present a division procedure in which each agent marks $n$ desired \emph{rectangles}. Their goal is to allocate each agent a single desired rectangle. However, because the rectangles might be arbitrarily thin, it is possible that a single rectangle will intersect all other rectangles. In this case, the procedure fails and no allocations are returned. In contrast, our procedure requires the agents to draw fat pieces. This guarantees that it always succeeds.

\section{Future Work}
\label{sec:future}
The challenge of fair cake-cutting with geometric constraints has a large potential for future research. Some possible directions are suggested below.

We would like to close the gaps between the possibility and impossibility results in Tables \ref{tab:Proportionality-bounds-for-square-cakes}
and \ref{tab:Relative-proportionality-bounds-}. The most interesting gap, in our opinion, is related to an unbounded plane. Our impossibility result assumes that the squares are parallel to each other; if the squares are allowed to rotate arbitrarily, then we do not have an impossibility result, and we do not know whether a proportional division is possible.

Based on our current results, and some other results which we had to omit in order to keep the paper length at a reasonable level, we make the following conjecture:
\begin{conjecture*}
When a cake $C$ is divided to $n$ agents each of whom must receive a fat rectangle, the attainable proportionality is:

\begin{align*}
\frac{1}{2n+Geom(C)}
\end{align*}
Where $Geom(C)$ is a (positive or negative) constant that depends only on the geometric shape of the cake.
\end{conjecture*}
In other words: the move from a one-dimensional division to a two-dimensional division asymptotically decreases the fraction that can be guaranteed
to every agent by a factor of 2.

Another direction is extending the results to cakes in three or more dimensions. We have some preliminary results in this direction.

It may be interesting to study cakes of different topologies, such as cylinders and spheres. We mention, in particular, the following potentially practical open question: \emph{is it possible to divide Earth (a sphere) in a fair-and-square way?}

The two auction types used by our procedures (see Subsection \ref{sub:auctions-intro}) can possibly be generalized. For example, it may be interesting to see what can be attained if each agent receives two entitlements instead of one. This is common in some rural settlements, in which each settler receives two plots --- one for housing and one for farming.

The present paper focused on constraints related to \emph{geometric shape} --- squareness or fatness. One could also consider constraints related to \emph{size}, e.g. by defining the family $S$ to be the family of all rectangles of length above 10 meters or area above 100 square meters. A problem with these constraints is that they are not scalable. For example, if the cake is 200-by-200 meters and there is either a length-minimum of 10 or an area-minimum of 100, then it is impossible to divide the land to more than 400 agents. Governments often cope with this problem by putting an upper bound on the number of people allowed to settle in a certain location. However, this limitation prevents people from taking advantage of new possibilities that become available as the number of people increases. For example, while in rural areas a land-plot of less than 10-by-10 meters may be considered useless because it
cannot be efficiently cultivated, in densely populated cities even a land-plot as small as 2-by-2 meters can be used as a parking lot for rent or as a lemonade selling spot. Limiting the number of agents assures that each agent gets a land-plot that can be cultivated efficiently,
but it may prevent more profitable ways of using the land-plots. In contrast, the squareness/fatness constraint is scalable because it does not depend on the absolute size of the land-cake. It is equally meaningful in both densely and sparsely populated areas.

The division problem can be extended by allowing each agent to have a different geometric constraint (a different family $S$ of usable shapes) or even to have utility functions which combine different families of usable shapes (with an agent-specific weight for each family).

This paper focuses on the basic fairness criterion of proportionality. We already started to study the stronger criterion of \emph{envy-freeness}\citep{SegalHalevi2015EnvyFree}, using substantially different techniques. It may be interesting to survey other results in the cake-cutting literature and see if and how they can be generalized to a two-dimensional cake.

\section{Acknowledgments}
\noindent This research was funded in part by the following institutions: 
\begin{itemize}
	\item The Doctoral Fellowships of Excellence Program at Bar-Ilan University;
	\item The Adar Foundation of the Economics Department and Wolfson Chair
	at Bar-Ilan Univ.;
	\item The ISF grants 1083/13 and 1241/12;
	\item The BSF grant 2012344.
\end{itemize}
These funds had no influence on the contents of the paper and do not necessarily endorse it.

This paper has greatly benefited from the comments and suggestions of anonymous referees on its earlier versions.
Particularly, the three rounds of comments and suggestions received from the editor and reviewer of the Journal of Mathematical Economics were extremely helpful in improving the quality of both substance and exposition of the paper, and we are very grateful for this most valuable input.

This paper has also benefited from comments by participants in the various seminars and conferences in which it was presented: Bar-Ilan game theory seminar, Hebrew University rationality center seminar, BISFAI 2015 conference, Technion industrial engineering department seminar, Technion computer science department seminar, Tel-Aviv university computational geometry seminar, Israel game-theory day, Paris game theory seminar, EuroCG 2016 conference and Corvinus university game-theory seminar.

Many members of the stackexchange.com communities kindly helped in coping with specific issues arising in this paper. In particular, we are grateful to the following members: Abel, Abishanka Saha, Acarchau, AcomPi, Adam P. Goucher, Adriano, aes, Alecos Papadopoulos, András
Salamon, Anthony Carapetis, Aravind, Barry Cipra, Barry Johnson, BKay, Boris Bukh, Boris Novikov, BradHards, Brian M. Scott, calculus, Chris Culter, Christian Blatter, Claudio Fiandrino, coffeemath, Cowboy, D.W., dafinguzman, Dale M, David Eppstein, David Richerby, Debashish,
Dejan Govc, Dömötör Pálvölgyi (domotorp), dtldarek, Emanuele Paolini, Flambino, Foo Barrigno, frafl, George Daccache, Hagen von Eitzen, Henno Brandsma, Henrik, Henry, Herbert Voss, Ilya Bogdanov, Ittay Weiss, J.R., Jan Kyncl, Jared Warner, Jeff Snider, John Gowers, John Hughes, JohnWO, Joonas Ilmavirta, Joseph O'Rourke, Kenny LJ, Kris Williams, Macavity, MappaGnosis, Marc van Leeuwen, Mari-Lou A, Mariusz Nowak, Mark Bennet, Martigan, Martin Van der Linden, Matt, Maulik V, Merk, Marcus Hum (mhum), Michael Albanese, Mikaël Mayer, Mike Jury, mjqxxxx, mvcouwen, Martin von Gagern (MvG), Nick Gill, njaja, Normal Human, NovaDenizen, Peter Franek, Peter Michor, Peter Woolfitt, pew, PhoemueX, Phoenix, Raphael Reitzig, R Salimi, Raymond Manzoni, Realz Slaw, Red, Robert Israel, Ross Millikan, sds, Stephen, Steven Landsburg, Steven Taschuk, StoneyB, tomasz, TonyK, Travis, vkurchatkin, Włodzimierz Holsztyński, Yury, Yuval Filmus and Zur Luria.

Last but not least: Erel is grateful to his wife Galya for the inspiration to this research, the creative ideas and the tasty cakes.
\appendix
\section{Staircase Lemma}\label{sec:staircase-lemma}
\begin{figure}
\psset{unit=0.55mm,hatchangle=0,hatchsep=2pt}
\def\landcakestaircase{
\psline[linecolor=black](70,0)(50,0)(50,10)(40,10)(40,30)(30,30)(30,50)(0,50)(0,70)
\Square[linecolor=green,fillcolor=green!30,fillstyle=solid,opacity=0.5](51,1){35}
\Square[linecolor=blue,fillcolor=blue!30,fillstyle=solid,opacity=0.5](41,11){23}
\Square[linecolor=orange,fillcolor=yellow!30,fillstyle=solid,opacity=0.5](31,31){17}
\Square[linecolor=red,fillcolor=red!30,fillstyle=solid,opacity=0.5](1,51){45}
}
\centering
\begin{pspicture}(80,100)
\landcakestaircase
\squarediagonal[linecolor=green,fillcolor=green!30,fillstyle=solid](51,1){35}
\squarediagonal[linecolor=blue,fillcolor=blue!30,fillstyle=solid](41,11){23}
\squarediagonal[linecolor=orange,fillcolor=yellow!30,fillstyle=solid](31,31){17}
\squarediagonal[linecolor=red,fillcolor=red!30,fillstyle=solid](1,51){45}
\rput[bl](51,1){1}
\rput[bl](41,11){2}
\rput[bl](31,31){3}
\rput[bl](1,51){4}
\end{pspicture}
\begin{pspicture}(80,70)
\landcakestaircase
\rput[tr](40,10){\tiny{$(x^*,y^*)$}}
\psframe[linecolor=blue,fillcolor=blue!10,fillstyle=solid,hatchcolor=blue,opacity=0.5,linestyle=dotted](31,31)(40,34)
\psframe[linecolor=blue,fillcolor=blue!10,fillstyle=solid,hatchcolor=blue,opacity=0.5,linestyle=dotted](51,1)(64,10)
\end{pspicture}
\protect\caption{\label{fig:staircase-lemma}
\textbf{a.} A staircase with $T=3$ teeth and $T+1=4$ corners and a square in each corner. The diagonal (dashed) represents $t_j$ --- the taxicab distance from the origin to the square center. The square at corner 2 is the winning square as its taxicab distance is minimal (the diagonal is closest to the origin). 
\protect \\
\textbf{b.} The shadow of the winning square (dotted). Note that each rectangular component of the shadow is entirely contained in the square of the corresponding corner.
}
\end{figure}
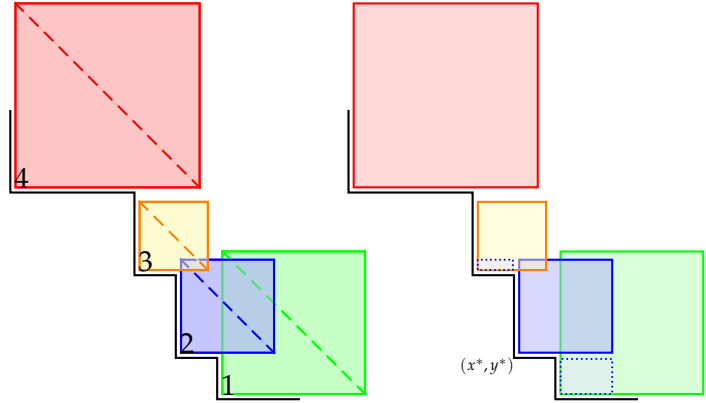

This appendix proves the following geometric lemma, which is used in Section \ref{sub:alg-2-walls}:

\begin{lemma}
(Staircase Lemma)
Let $C$ be a staircase-shaped polygonal domain with $T$ teeth (and $T+1$ corners). Suppose that in each inner corner $j\range{1}{T+1}$, with coordinates $(x_j,y_j)$, there is a square with side-length $l_j$ (the square $[x_j,x_j+l_j]\times[y_j,y_j+l_j]$).

Define the \emph{shadow} of square $j$ as the intersection of $C$ with the rectangle $[0,x_j+l_j]\times[0,y_j+l_j]$ (this is the area of $C$ that is removed when cutting from the top-right corner of square $j$ towards the bottom and left boundaries of $C$; see Figure \ref{fig:staircase-lemma}/b).

There exists a corner $j$ such that the shadow of square $j$ is contained in the union of the $T+1$ squares.
\end{lemma}
\begin{proof}
For every $j\range{1}{T+1}$, define:
\begin{align*}
t_j := x_j+y_j+l_j
\end{align*}
$t_j$ can be interpreted as the ``taxicab distance'' ($\ell_{1}$ distance) from the origin to the center of the square at corner $j$, or equivalently to its bottom-right or top-left corner; 

Define the \emph{winning square} as the square $j$ \emph{for which $t_j$ is minimized}. Denote its corner coordinates by $(x^*,y^*)$ and its side-length by $l^*$. We now prove that the shadows of the winning square are contained in the other squares. We decompose the shadows of the winning square to pairwise-disjoint rectangular components in the following way.
\begin{itemize}
\item For each corner $j$ to the top-left of the winning square, the component is a rectangle with coordinates: $[x_j,x^{*}]\times[y_j,y^*+l^*]$. Note that this component is empty if $y_j\geq y^*+l^*$, as in  corner 4 in Figure \ref{fig:staircase-lemma}.
\item For each corner $j$ to the bottom-right of the winning square, the component is a rectangle with coordinates: $[x_j,x^*+l^*]\times[y_j,y^*]$. This component is empty if $x_j\geq x^*+l^*$.
\end{itemize}
By definition of the winning square, for every $j\range{1}{T+1}$:
\begin{align}\label{eq:winning}
x_j+y_j+l_j \geq x^*+y^*+l^*
\end{align}
Now:
\begin{itemize}
\item For each corner $j$ to the top-left of the winning square, we have $x_j<x^*$. Combining this with (\ref{eq:winning}) gives $y^*+l^*<y_j+l_j$. Moreover, if the component in that corner is not empty, then necessarily $y_j<y^*+l^*$. Combining this with (\ref{eq:winning}) gives $x^*<x_j+l_j$. Hence, the component $[x_j,x^{*}]\times[y_j,y^*+l^*]$ is contained in the square $[x_j,x_j+l_j]\times[y_j,y_j+l_j]$.
\item For each corner $j$ to the bottom-right of the winning square, we have $y_j<y^*$. Combining this with (\ref{eq:winning}) gives $x^*+l^*<x_j+l_j$. Moreover, if the component in that corner is not empty, then necessarily $x_j<x^*+l^*$. Combining this with (\ref{eq:winning}) gives $y^*<y_j+l_j$. Hence, the component $[x_j,x^*+l^*]\times[y_j,y^*]$ is contained in the square $[x_j,x_j+l_j]\times[y_j,y_j+l_j]$.
\end{itemize}
We proved that every component of the shadow of the winning square is contained in one of the $T+1$ squares; hence, the winning square satisfies the requirement of lemma.
\end{proof}

\section{Non-intersection of Squares in Fat Procedure} \label{sec:dove-hawk}
This appendix proves that in the last step of the Fat Procedure (Subsection \ref{sub:fat-procedures}), the $n$ returned squares do not overlap.

\renewcommand{\landcakefat}{
\psframe[linecolor=brown](0,0)(1,1.1)
\yvalue{1.1}{L}
\yvalue{0}{0}
}
\newcommand{\landcaketopbottomshort}{
 \landcakefat
 \rput[l](1.2,0.95){$\leftarrow Top$}
 \yvalue{0.8}{y_t}
 \psline[linestyle=dotted](0,0.8)(1.1,0.8)
 \yvalue{0.3}{y_b}
 \psline[linestyle=dotted](0,0.3)(1.1,0.3)
 \rput[l](1.2,0.15){$\leftarrow Bottom'$}
}
Recall that at this step, the cake has two distinguished regions: $Bottom':= [0,1]\times[0,y_b]$ and $Top := [0,1]\times[y_t,L]$, both of which are 2-thin rectangles, i.e, $0<y_b<1/2\leq L-1/2 <y_t<L$. In each region there is a family of squares: the \emph{bottom squares} were returned by applying the Thin Procedure to Bottom', and the \emph{top squares} were returned by applying the Thin Procedure to Top. The squares in each family are pairwise-disjoint, but squares from different families might overlap. Our goal is to prove that, after a single largest square is removed, the remaining squares do not overlap, as in the following illustration:
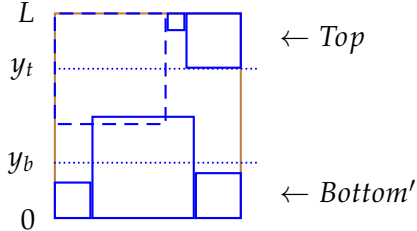
\begin{center}
\begin{pspicture}(-1,-0.1)(3,1.2)
 \landcaketopbottomshort
 \Square(0,0){0.2}
 \Square(0.2,0){0.55}
 \Square(0.75,0){0.25}

 \Square[linestyle=dashed](0,0.5){0.6}
 \Square(0.6,1){0.1}
 \Square(0.7,0.8){0.3}
\end{pspicture}
\end{center}
Recall that, by the specification of the Thin Procedure (Subsection \ref{sub:thin}), the squares in each family can be divided to two types, which we call ``doves'' and ``hawks'':
\begin{itemize}
\item \emph{Doves} are squares generated by Outcome \#1 of the Thin Procedure (or by recursive calls to the Fat Procedure). They are contained within the four walls of their rectangle: the bottom doves are contained in $[0,1]\times[0,y_b]$, and the top doves are contained in $[0,1]\times[y_t,L]$.
\item \emph{Hawks} are squares generated by Outcome \#2 of the Thin Procedure. They are contained within only three walls of their rectangle, with one of their edges adjacent to the wall opposite the open side: the bottom edge of all bottom hawks is at $y=0$, and the top edge of all top hawks is at $y=L$. Moreover, the side-length of each hawk is at most the longer side of its rectangle minus the shorter side of its rectangle; hence, the side-length of all bottom hawks is at most $1-y_b$ and their top edge is in $y\in[y_b,1-y_b]$, and the side-length of all top hawks is at most $1-(L-y_t)$
and their bottom edge is in $y\in[L-(1-L+y_t), y_t]$.
\end{itemize}

\begin{claim}
\label{claim:sum1}
In each family, the sum of the side-lengths of all hawks is at most 1.
\end{claim}
\begin{proof}
The bottom hawks are all bounded in a rectangle of length 1: $[0,1]\times[0,1-y_b]$. Their bottom side is at $y=0$. Since they do not overlap, the sum of their side-lengths must be at most 1. A similar argument holds for the top hawks.
\end{proof}
An immediate corollary of Claim \ref{claim:sum1} is that at most one hawk from each side has side-length more than 1/2. We call each of these two hawks (if it exists) the \emph{dangerous hawk}.

We say that a square $q$ \emph{attacks} a square $q'$ if $q$ is larger than $q'$ and $q$ overlaps $q'$. This is possible only if $q$ and $q'$ are in two opposite families, since the squares in each family are pairwise-disjoint. The doves obviously do not attack each other because $y_b<y_t$. So the only possible attacks are: top hawks attacking bottom hawks/doves, or bottom hawks attacking top hawks/doves.

After removing the largest square, at most one dangerous hawk remains; it is only this hawk that might attack other squares in the opposite side. We now prove that even this dangerous hawk does not attack other squares.
\begin{claim}
No remaining hawk attacks any dove.
\end{claim}
\begin{proof}
We prove that no remaining hawk even enters the rectangle of the opposite family (no remaining bottom-hawk enters $Top$ and no remaining top-hawk enters $Bottom'$). Since all doves are contained in their rectangle, they are safe. There are two cases:

\textbf{Case 1}: $y_t\geq L-y_b$. Then also $y_t\geq1-y_b$. The side-length of all bottom hawks is at most $1-y_b$, so no bottom hawk enters $Top$. If the top dangerous hawk enters $Bottom'$, then its side-length must be more than $L-y_b$, so it is larger than all bottom hawks. Hence, it is the largest square and it is removed.

\textbf{Case 2}: $y_t<L-y_b$. Then also $1-(L-y_t)<1-y_b\leq L-y_b$. The side-length of all top hawks is at most $1-(L-y_t)$, no top hawk enters $Bottom'$. If the bottom dangerous hawk enters $Top$, then its side-length must be more than $y_t$, so it is larger than all top hawks. Hence, it is the largest square and it is removed.
\end{proof}
\begin{claim}
No remaining hawk attacks any hawk.
\end{claim}
\begin{proof}
There are two cases: 

\textbf{Case 1}: There is only one hawk (either bottom \emph{or} top) with side-length more than $1/2$. This is the largest square so it is removed. The remaining squares have side-length at most $1/2$ and thus do not attack each other.

\textbf{Case 2}: There are two hawks (bottom \emph{and} top) with side-length more than $1/2$. W.l.o.g, assume the top hawk is the largest, with a side-length of $h_t\geq h_b$. By Claim \ref{claim:sum1}, the sum of the side-lengths of all other top hawks is at most $1-h_t$, hence the side-length of any single other top hawk is at most $1-h_t$ which is at most $1-h_b$ which is at most $L-h_b$. Hence, the bottom side of all remaining top hawks is above $h_b$. Hence the remaining bottom hawk cannot attack any of them. 
\end{proof}

\section{Existence of Best pieces}
\label{sec:Conditions-for-Existence}

This appendix shows how to prove the existence of a usable piece with
a maximum value (this is used in the proof of Claim \ref{claim:8n-6}).
We start by defining a metric space of pieces (recall that a \emph{piece
}is a Borel subset of $\mathbb{R}^{2}$ and $Area$ is its Lebesgue
measure).
\begin{defn}
The \emph{symmetric difference (SD) pseudo-metric }is defined by:

\[
d_{SD}(X,Y)=Area[(X\setminus Y)\cup(Y\setminus X)]
\]

\end{defn}
$d_{SD}$ is not a metric because there may be different pieces whose
symmetric difference has an area of 0, e.g, a square with an additional
point and a square with a missing point. To make SD a metric, we consider
only pieces $X$ that are \emph{regularly open}, i.e, the interior
of the closure of themselves: $X=Int[Cl[X]]$. 
\begin{claim}
SD is a metric on the set of all regularly-open pieces.\end{claim}
\begin{proof}
\footnote{We are thankful to Tony K., Phoemue X., Dafin Guzman, Henno Brandsma
and Ittay Weiss for contributing to this proof via discussions in
the math.stackexchange.com website (http://math.stackexchange.com/a/1099461/29780).} Let $X$ and $Y$ be two regularly-open sets such that $d_{SD}(X,Y)=0$.
We prove that $X=Y$.

$d_{SD}(X,Y)=0$ implies $Area[X\setminus Y]=Area[Y\setminus X]=0$.

$Y\subseteq Cl[Y]$ so $X\setminus Y\supseteq X\setminus Cl[Y]$.
Hence also \textbf{$Area[X\setminus Cl[Y]]=0$.}

$X$ is open and $Cl[Y]$ is closed; hence $X\setminus Cl[Y]$ is
open (it is an intersection of two open sets).

The only open set with an area of 0 is the empty set (because any
non-empty open set contains a ball with a positive measure). Hence:\textbf{
$X\setminus Cl[Y]=\emptyset$.}

Equivalently: $X\,\subseteq\,Cl[Y]$.

By taking the Cl of both sides: $Cl[X]\subseteq Cl[Y]$

By a symmetric argument: $Cl[Y]\subseteq Cl[X]$

Hence: $Cl[Y]=Cl[X]$

By taking the Int of both sides and by the fact that they are regularly-open:
$Y=X$.
\end{proof}
Thus when we allocate a square we actually allocate only its interior.
This has no effect on the utility of the agents since the boundary
has an area of 0 and so its value is 0 for all agents.
\begin{claim}
Let $D$ be the metric space defined by $d_{SD}$. Let $V$ be a measure
absolutely continuous with respect to area. Then $V$ is a uniformly
continuous function from $D$ to $\mathbb{R}$.\end{claim}
\begin{proof}
The fact that $V$ is an absolutely continuous measure implies that,
for every $\epsilon>0$ there is a $\delta>0$ such that every piece
$X$ with $Area(X)<\delta$ has $V(X)<\epsilon$ \citep[Proposition 15.5 on page 251]{Nielsen1997Introduction}.
Hence, for every two pieces $X$ and $Y$, if $d_{SD}(X,Y)<\delta$
then $Area(X\setminus Y)<\delta$ and $Area(Y\setminus X)<\delta$,
then $V(X\setminus Y)<\epsilon$ and $V(Y\setminus X)<\epsilon$,
then $|V(X)-V(Y)|=|V(X\setminus Y)-V(Y\setminus X)|<\epsilon$.\end{proof}
\begin{claim}
Let $V$ be a measure absolutely continuous with respect to area and
$Q$ a set of pieces which is compact in the SD metric space. Then
there exists a piece $q\in Q$ for which $V$ is maximized.\end{claim}
\begin{proof}
By the previous claim, $V$ is a uniformly continuous and hence a
continuous real-valued function. By the extreme value theorem, it
attains a maximum in every compact set.
\end{proof}
The value measures considered in this paper are always absolutely continuous with respect to area. Hence, to prove that a certain set of pieces $Q$ contains a ``best piece'' it is sufficient to prove
that $Q$ is compact. We do this now for the special case in which $Q$ is the set of open squares contained in a given cake (note that the same proof could be used for the set of closed squares):
\begin{claim}
Let $C$ be a closed, bounded subset of $\mathbb{R}^{2}$. Let $Q$
be the set of all open squares contained in $C$. Then $Q$ is compact
in the SD metric space.\end{claim}
\begin{proof}
It is sufficient to prove that $Q$ is sequentially compact, i.e.
every infinite sequence of open squares in $C$ has a subsequence
converging to an open square in $C$. Let $\{q_{i}\}_{i=1}^{\infty}$
be an infinite sequence of open squares in $C$. For every $q_{i}$,
let $(A_{i},B_{i})$ be a pair of opposite corners. Because $C$ is
compact, it contains $Cl[q]$ and hence contains the points $A_{i}$
and $B_{i}$. Hence the infinite sequence of pairs of points, $\{(A_{i},B_{i})\}_{i=1}^{\infty}$,
is an infinite sequence in $C\times C$. $C\times C$ is compact because
it is a finite product of compact sets. Hence, the sequence has a
subsequence converging to a limit point $(A^{*},B^{*})\in C$. From
now on we assume that $\{(A_{i},B_{i})\}_{i=1}^{\infty}$ is that
converging subsequence. Let $q^{*}$ be the open square having $A^{*}$
and $B^{*}$ as two opposite corners. We show that: (a) $q^{*}$ is
an open square in $C$; (b) The subsequence $\{q_{i}\}_{i=1}^{\infty}$
converges to $q^{*}$.

(a) $q^{*}$ is a obviously an open square by definition. We have
to show that each point in $q^{*}$ is also a point of $C$. To every
square $q_{i}$, attach a local coordinate system in which corner
$A_{i}$ has coordinates $0,0$ and corner $B_{i}$ has coordinates
$1,1$ and every other point in $Cl[q_{i}]$ has coordinates in $[0,1]\times[0,1]$.
For every coordinate $(x,y)\in[0,1]\times[0,1]$, let $q_{i}(x,y)$
be the unique point with these coordinates in $Cl[q_{i}]$ (e.g. $A_{i}=q_{i}(0,0)$
and $B_{i}=q_{i}(1,1)$). 

For every $(x,y)$, The sequence $\{q_{i}(x,y)\}_{i=1}^{\infty}$ is a sequence of points which are all in $C$, and they converge to
$q^{*}(x,y)$. Since $C$ is closed, $q^{*}(x,y)\in C$.

(b) For every $i$, the area of the symmetric difference between $q^{*}$
and $q_{i}$ is bounded and satisfies the following inequality:
\[
d_{SD}(q^{*},q_{i})\leq4\cdot\max(d(A^{*},A_{i}),d(B^{*},B_{i}))\cdot\max(d(A^{*},B^{*}),d(A^{*},B_{i}),d(A_{i},B^{*}),d(A_{i},B_{i}))
\]
 Since all distances are bounded and $d(A^{*},A_{i})$, $d(B^{*},B_{i})$
converge to 0, the same is true for $d_{SD}(q^{*},q_{i})$. Hence,
the subsequence $\{q_{i}\}_{i=1}^{\infty}$ converges to $q$. 

The previous paragraph proved that $Q$ is sequentially compact. Hence
it is compact.
\end{proof}
In a similar way it is possible to prove similar results for other
families $S$, such as the family of $R$-fat rectangles or cubes.

\section*{References}
\bibliographystyle{elsarticle-num-names-alpha}
\bibliography{FairAndSquare}

\end{document}